\documentclass[letterpaper,twocolumn,10pt]{article}
\usepackage{usenix}
\usepackage[english]{babel}
\usepackage{balance}
\usepackage{bm}     
\usepackage{xspace}
\usepackage{array}
\usepackage{subcaption}
\newcolumntype{C}[1]{>{\centering\arraybackslash}m{#1}}
\usepackage{multirow}
\usepackage{amssymb}
\usepackage{amsmath}
\usepackage{amsfonts}
\providecommand{\Description}[1]{}
\usepackage{accsupp}
\usepackage{accessibility}
\usepackage{setspace}
\usepackage{amsthm}
\newtheorem{theorem}{Theorem}[section]
\newtheorem{lemma}[theorem]{Lemma}
\theoremstyle{definition}
\newtheorem{definition}[theorem]{Definition}
\usepackage{threeparttable}
\usepackage{colortbl} 
\usepackage{enumerate}
\usepackage{pifont}     
\usepackage{bbding}     
\usepackage{fontawesome}
\usepackage{algorithm}
\usepackage[noend]{algpseudocode}
\definecolor{commentcolor}{RGB}{0, 100, 100} 
\definecolor{actioncolor}{RGB}{150, 0, 0}    
\definecolor{statecolor}{RGB}{0, 0, 150}     
\usepackage{booktabs}

\newcommand{\AttackAction}[1]{\textcolor{actioncolor}{\textbf{#1}}}
\newcommand{\StateUpdate}[1]{\textcolor{statecolor}{#1}}
\newcommand{\MyComment}[1]{\hfill \textcolor{commentcolor}{\(\triangleright\) \textit{#1}}}

\newcommand{\heading}[1]{{\vspace{0pt}\noindent{\textbf{#1}}}}

\newcommand{\tabref}[1]{\mbox{Table~\ref{#1}}}

\usepackage{tikz}
\usepackage{filecontents}

\begin{document}

\date{}

\title{\Large \bf PDoS: A Profitable Denial-of-Service Attack against Proof-of-Work Blockchain Liveness}

\author{
{\rm Junjie Hu}\\
Shanghai Jiao Tong University\\
nakamoto@sjtu.edu.cn
\and
{\rm Tianzhu Han}\\
Shanghai Jiao Tong University\\
v.viktor@sjtu.edu.cn
\and
{\rm Na Ruan}\\
Shanghai Jiao Tong University\\
naruan@sjtu.edu.cn
} 

\maketitle

\begin{abstract}
The security and liveness of Proof-of-Work (PoW) blockchains fundamentally depend on the economic rationality of miners. Existing incentive-driven denial-of-service attacks, such as BDoS, can deter rational miners from participating, but require the attacker to continuously absorb substantial economic losses and are therefore difficult to sustain in high-value networks. Meanwhile, prior infiltration-based withholding attacks are designed to extract revenue rather than to directly disrupt chain liveness.

We present PDoS, a hybrid attack that combines block header signal deterrence with parasitic revenue extraction. PDoS disrupts blockchain liveness while exploiting the victim pool's share-reward mechanism to subsidize the attack cost, thereby lowering the adversarial hash-power threshold required to induce rational miners to shut down. We further show a counterintuitive result: in high-fee or high-MEV environments, higher block value can make PoW systems less secure by increasing the attacker's parasitic revenue and pushing the attack across the break-even point into a self-sustaining, or even profitable, regime. To the best of our knowledge, PDoS is the first attack to demonstrate that disrupting PoW blockchain liveness can be economically self-sustaining and even profitable.
\end{abstract}

\section{Introduction}

\heading{Nakamoto Consensus.}
The security and liveness of Nakamoto consensus fundamentally rely on the economic rationality of miners. Prior work has extensively studied the safety, consistency, and confirmation properties of longest-chain PoW protocols under synchronous, partially synchronous, and bounded-delay models~\cite{nakamoto2008bitcoin,garay2015backbone,pass2017analysis,pass2017sleepy,kiffer2018better,gazi2020tight,li2021close,kiffer2024bounded}. Yet the incentive landscape of PoW blockchains has changed substantially. Mining-pool industrialization, rising transaction-fee share, and the rapid growth of \emph{miner-extractable value} (MEV) have made block rewards increasingly heterogeneous. Fee volatility, bursty high-value transactions, and MEV opportunities can materially reshape the profitability boundary of deviations from honest mining~\cite{daian2020flashboys,carlsten2016instability,qin2022darkforest,barzur2023werlman,sarenche2025volatile}. Conventional wisdom suggests that higher block value improves security by raising the cost of brute-force attacks. In reality, the same high-value environment may also enable subtler and more damaging forms of incentive manipulation.

Existing attacks largely fall into two lines of work. The first focuses on \emph{economic theft}: \emph{Selfish Mining}, \emph{Stubborn Mining}, and related variants increase adversarial revenue through strategic withholding and release~\cite{eyal2014majority,nayak2016stubborn,sapirshtein2016optimal,bai2023multiple}, while \emph{BWH}, \emph{FAW}, \emph{PAW}, and their extensions exploit mining-pool payout rules to continuously extract victim-funded revenue~\cite{rosenfeld2011pooled,eyal2015miners,kwon2017faw,gao2019paw,yang2022fwap}. These attacks can be profitable, but they do not aim to directly disrupt chain liveness. The second line targets \emph{liveness degradation}. In \emph{Blockchain DoS (BDoS)}, for example, the attacker releases only a block header while withholding the body, creating a credible threat of branch competition and inducing rational miners to reconsider whether continued mining is worthwhile~\cite{mirkin2020bdos}. More broadly, recent attacks on the mempool, execution load, storage, bandwidth, and sequencing layers show that blockchain systems can be disrupted along multiple resource dimensions~\cite{li2021deter,yaish2024speculative,he2024nurgle,tsuchiya2025amplification,li2025denialsequencing}. We summarize these attacks in \tabref{tab:attack_comparison}.

\heading{Motivation.}
Prior liveness attacks are generally considered cost-prohibitive. The original BDoS model, for instance, assumes an adversary dedicated to disrupting system liveness at any cost, concluding that the attack inevitably incurs an economic loss compared to honest mining. Despite this severe opportunity cost, our comprehensive utility modeling refines this view by shifting to an \emph{accounting profit} perspective: when properly incorporating race-winning block rewards and MEV, we demonstrate that BDoS can theoretically offset its explicit operational costs in high-value networks. However, this demands an unrealistically high hash-power threshold and extreme market conditions. Following the April 2024 reward halving, the declining profitability factor has rendered BDoS operationally unsustainable in current environments.

This leaves a key question unanswered: \emph{Can a PoW adversary simultaneously degrade chain liveness and finance the attack through victim-funded revenue, eventually turning denial of service into a self-sustaining or even profitable strategy?}

\heading{Our Approach.}
We answer this question affirmatively with \emph{Profitable Denial-of-Service} (PDoS), a hybrid attack that combines \emph{header signal deterrence} with \emph{parasitic revenue extraction}. Unlike conventional network-layer DoS attacks, PDoS directly targets the incentive structure of the consensus layer by exploiting two distinct entities: an oblivious \emph{victim pool} serving as the funding source, and the remaining rational network, defined as the \emph{target miners} subjected to liveness degradation. The attacker infiltrates the victim pool and coordinates two mechanisms. First, when it finds a block, it releases only the block header and withholds the body, thereby signaling the possible existence of a competing branch and increasing the perceived orphaning risk of the target miners. Once the expected loss exceeds the expected gain from continued mining, these rational targets are induced to shut down. Second, to make the attack operationally sustainable, the attacker dynamically reallocates hash power between private mining and pool infiltration. By continuously submitting shares while withholding full blocks, it exploits the victim pool's reward-allocation mechanism to extract revenue that offsets, or even exceeds, the explicit operational costs of the attack.

We further develop an MEV-aware dynamic incentive model that jointly captures the coinbase reward, time-varying transaction fees, and bursty high-value rewards. Our analysis reveals a counterintuitive result: high MEV does not necessarily strengthen PoW security. Instead, by amplifying the attacker's infiltration revenue, it can reduce the net cost of liveness disruption and make the attack economically self-sustaining. We refer to this phenomenon as \emph{attack subsidization}: the infiltrated victim pool inadvertently provides both the financial subsidy and the hash-power leverage required to suppress the target miners.

\begin{table}[t]
\caption{Comparison of PDoS with representative prior mining attacks. Prior mining attacks typically achieve either profitability or liveness disruption, but not both. \textbf{PDoS} is the first attack that simultaneously targets PoW blockchain liveness while remaining operationally self-sustaining.}
\label{tab:attack_comparison}
\footnotesize
\centering
\begin{tabular}{@{}|C{4.7cm}|C{1.1cm}|C{1.1cm}|@{}}
\toprule
\textbf{Attack Type} &
\textbf{Profitable} &
\textbf{Liveness} \\ \midrule

\shortstack[c]{Withholding attacks\\
{\footnotesize \cite{eyal2014majority,nayak2016stubborn,marmolejo2019competing,li2021semiselfish,mccorry2018smart,hu2023greedymine,gao2019paw,kwon2017faw,eyal2015miners,rosenfeld2011pooled,yang2022fwap,hu2023bribery,zhang2024maxattestation}}}
& \Checkmark & \ding{55} \\ \midrule

\shortstack[c]{DoS attacks\\
{\footnotesize \cite{yaish2024speculative,li2021deter,tsuchiya2025amplification,he2024nurgle,li2025denialsequencing}}}
& \ding{55} & \Checkmark \\ \midrule

BDoS~\cite{mirkin2020bdos}
& \ding{55} & \Checkmark \\ \midrule

\rowcolor{gray!30}
\textbf{PDoS [Ours]}
& \Checkmark
& \Checkmark \\ \bottomrule
\end{tabular}
\end{table}

\begin{center}
\setlength{\fboxsep}{3pt}
\fbox{
\parbox{0.96\linewidth}{
\noindent\textbf{Responsible disclosure}. For ethical reasons, we responsibly disclosed our findings to the Bitcoin Core Security Team in April 2026. To prevent potential exploitation, the detailed disclosure records are temporarily withheld and will be made publicly available upon the conclusion of the vulnerability handling process.
}}
\end{center}

\heading{Our Contributions.}
Our main contributions are as follows:
\begin{enumerate}
    \item[$\bullet$] We present PDoS, the first denial-of-service attack that targets PoW blockchain liveness while being economically self-sustaining and potentially generating a positive accounting profit in real-world settings.
    \item[$\bullet$] We formalize system evolution under three rational target responses, $S \in \{\mathrm{Mine}, \mathrm{SPV}, \mathrm{Stop}\}$, and derive a unified analytical framework to capture the complex attacker--victim--target interactions.
    \item[$\bullet$] We develop a dynamic incentive model incorporating coinbase rewards, time-varying fees, and bursty MEV-like rewards. By explicitly decoupling the revenue streams into an MEV inflation coefficient ($e^S$), effective block-generation rate ($v_{\mathcal{M}}^S$), share-extraction rate ($s_{\mathcal{M}}^S  $), and active-time ratio ($\theta_{\mathcal{M}}^S$), we analytically characterize the attack's operating break-even point and formally prove that the target miners' subgame converges to a unique strict Nash equilibrium in which all target miners stop mining.
    \item[$\bullet$] We validate PDoS through trace-driven evaluations spanning the 2024 Bitcoin reward halving. Compared to BDoS, PDoS substantially lowers the attacker hash-power threshold and reduces net attack costs, thereby more easily achieving infinite attack endurance. It expands the feasible attack region by establishing a lower operating break-even bound and a higher validity upper bound. Finally, we quantify the effectiveness of PPLNS and delayed-payout defenses against the no-hardening baseline.
\end{enumerate}

\section{Model}\label{Model}
We formalize the PoW blockchain, the attacker capabilities, and the notation used in the subsequent analysis.

\subsection{Mining Model}

\heading{Participants.}
We consider four types of entities:
\begin{itemize}
    \item \textbf{Attacker ($\mathcal{A}$):} controls hash power $\alpha$ and aims to maximize its utility while degrading chain liveness.
    \item \textbf{Victim Pool ($\mathcal{V}$):} controls hash power $\beta$, distributes share rewards proportionally to submitted PPoWs, and always performs full-block validation.
    \item \textbf{Target Miner Population ($\mathcal{T}$):} controls total hash power $\eta$ and is modeled as a continuum of rational miners $\mathcal{T}_i$ with $\eta_i \to 0$ and $\sum_i \eta_i=\eta$.
    \item \textbf{Other Miners ($\mathcal{O}$):} control hash power $\delta$ and follow the Nakamoto consensus.
\end{itemize}
The total network hash power is normalized to one:
$
\alpha+\beta+\eta+\delta=1$.
Let $\mathcal{M}\in\{\mathcal{A},\mathcal{V},\mathcal{T},\mathcal{O}\}$ denote an arbitrary entity. Define the external miner set as $\mathcal{E}=\mathcal{T}\cup\mathcal{O}$ and the honest miner set as $\mathcal{H}=\mathcal{V}\cup\mathcal{T}\cup\mathcal{O}$.

\heading{Network Propagation.}
Following prior analyses of withholding-based attacks~\cite{rosenfeld2011pooled,eyal2015miners,kwon2017faw,gao2019paw}, we ignore ordinary propagation delay and focus on fork-time competition. We use a \emph{rushing ability} parameter $\gamma\in[0,1]$ to capture $\mathcal{A}$'s propagation advantage: when two competing blocks at the same height are released, a neutral miner supports $\mathcal{A}$'s branch with probability $\gamma$ and the honest branch with probability $1-\gamma$. The case $\gamma=0.5$ corresponds to symmetric propagation power.

\heading{Block Generation.}
We assume the attack horizon is shorter than the difficulty-adjustment period and thus ignore difficulty retargeting. Block generation follows a Poisson process. Let $\lambda$ denote the baseline block generation rate, and let $\alpha_{\mathrm{act}}$ denote the active hash power producing FPoWs. Then inter-block time is exponentially distributed with rate $\lambda\alpha_{\mathrm{act}}$. For any entity $\mathcal{M}$ with active hash power $\alpha_{\mathcal{M}}$, its block-generation rate is $\lambda\alpha_{\mathcal{M}}$, and its probability of finding the next block first is $\frac{\alpha_{\mathcal{M}}}{\alpha_{\mathrm{act}}}$.

\heading{Dynamic Economic Incentives.}
We model the expected reward of a valid block as
$
\mathcal{R}_b \triangleq K + (\lambda_t+\lambda_w\mathbb{E}[\mathcal{R}_w])\mathbb{E}[\tau_b]$,
where $K$ is the fixed coinbase reward, $\lambda_t$ is the transaction-fee accumulation rate, and $\mathcal{R}_w$ is a bursty whale reward whose arrivals follow a Poisson process with rate $\lambda_w$ and whose value follows $\mathrm{LogNorm}(\mu,\sigma^2)$, so that
$
\mathbb{E}[\mathcal{R}_w]=e^{\mu+\frac{\sigma^2}{2}}$.
Under the no-attack baseline, $\alpha_{\mathrm{act}}=1$ and $\tau_b\sim\mathrm{Exp}(\lambda)$, hence
$
\mathbb{E}[\tau_b]=\frac{1}{\lambda} ~\text{and}~
\mathcal{R}_b
=
K+\frac{\lambda_t+\lambda_w\mathbb{E}[\mathcal{R}_w]}{\lambda}$.

In addition to block rewards, a pool participant may earn a \emph{share reward} $\mathcal{R}_s$. If its PPoW share is $w$, then
$\mathcal{R}_s = w\cdot \mathcal{R}_{\mathrm{pool}}$,
where $\mathcal{R}_{\mathrm{pool}}$ is the pool's expected total revenue. 

We assume a constant operating cost (OPEX) $c$ per unit hash power per unit time~\cite{ccaf_electricity}. We explicitly ignore the one-time capital expenditure (CAPEX) for hardware acquisition, as such sunk costs do not affect marginal strategic decisions.

\heading{Baseline Utility and Profitability.}
Let $R_{\mathcal{M}}(t)$ and $C_{\mathcal{M}}(t)$ denote the cumulative revenue and cost of entity $\mathcal{M}$ over $[0,t]$, where
$C_{\mathcal{M}}(t)=\alpha_{\mathcal{M}}ct$.
We define the baseline normalized utility as
\begin{align}
U_{\mathcal{M}}^b
\triangleq
\lim_{t\to\infty}
\frac{1}{\alpha_{\mathcal{M}}}
\left(
\frac{R_{\mathcal{M}}(t)}{t}
-
\frac{C_{\mathcal{M}}(t)}{t}
\right)
=
\lambda\mathcal{R}_b-c.
\end{align}
The baseline profitability factor is
$\omega_b \triangleq \frac{\lambda\mathcal{R}_b}{c}$.
Honest mining is profitable iff $\omega_b>1$, equivalently $U_{\mathcal{M}}^b>0$.

\subsection{Threat Model}

\heading{Attacker Capabilities.}
We assume that $\mathcal{A}$ has the following pool- and network-level capabilities: (1) \textit{Dynamic power allocation:} $\mathcal{A}$ can split its total hash power $\alpha$ into private hash power $\alpha_{\mathrm{pri}}$ and infiltration hash power $\alpha_{\mathrm{inf}}$, used for private mining and pool infiltration, respectively. Define $r=\frac{\alpha_{\mathrm{inf}}}{\alpha} ~\text{and}~
    r'=\frac{\alpha_{\mathrm{pri}}}{\alpha}$.
    We allow $r+r'\le 1$ to model partial shutdown. (2) \textit{Withholdable pool work:} As in BWH/FAW, an infiltrating worker can retain an FPoW and submit it later while continuing to submit ordinary PPoWs~\cite{eyal2015miners,kwon2017faw}. A standard mining job exposes the fields needed to construct and verify the candidate header, and a later share submission lets the pool reconstruct and publish the associated block~\cite{stratumv2mining}. (3) \textit{Asymmetric publication:} $\mathcal{A}$ may publish a block header $h$ while withholding the corresponding body $b$, thereby emitting a credible deterrence signal that a competing branch may already exist. Bitcoin already treats headers and missing transaction data as separable relay objects~\cite{bip152}. (4) \textit{Race observation:} $\mathcal{A}$ monitors public block announcements and immediately decides whether to discard or settle a retained FPoW. As in FAW, success is therefore sensitive to stale-share acceptance and propagation delay; we expose this dependence through $\gamma$ and discuss stricter pool policies in Section~\ref{sec:defenses}.

\heading{Target Strategy Space.}
The target miner population $\mathcal{T}$ chooses among
$
S\in\{\mathrm{Mine},\mathrm{SPV},\mathrm{Stop}\}$
according to expected normalized utility $\mathbb{E}[U_{\mathcal{T}}^S]$: (1) $S=\mathrm{Mine}$: continue normal mining on the current branch; (2) $S=\mathrm{SPV}$: extend the published header via SPV mining to save validation time; (3) $S=\mathrm{Stop}$: shut down to avoid further loss. We summarize the notation used in this paper in Table~\ref{tab:notation}.

\begin{table}[t]
\centering
\small
\caption{Core notation}
\label{tab:notation}
\renewcommand{\arraystretch}{1}
\setlength{\tabcolsep}{3pt}
\resizebox{\columnwidth}{!}{%
\begin{tabular}{l p{7.5cm}}
\toprule
\textbf{Notation} & \textbf{Description} \\
\midrule
\multicolumn{2}{l}{\textit{Participants}} \\
$\mathcal{A}, \mathcal{V}, \mathcal{T}, \mathcal{O}$ & Attacker, victim pool, target miner population, and other miners \\
$\alpha, \beta, \eta, \delta$ & Hash power of $\mathcal{A}, \mathcal{V}, \mathcal{T}, \mathcal{O}$, respectively \\
$\mathcal{E}, \mathcal{H}$ & External miner set; honest miner set \\

\midrule
\multicolumn{2}{l}{\textit{Mining System \& Economic Parameters}} \\
$\lambda, \gamma$ & Baseline block generation rate; rushing ability \\
$K, M$ & Coinbase reward; MEV ratio \\
$\lambda_t, \lambda_w$ & Transaction-fee accumulation rate; whale-reward arrival rate \\
$\mathcal{R}_b, \mathcal{R}_w, \mathcal{R}_s$ & Expected block reward, whale reward, and share reward \\
$c, \omega_b$ & Normalized operating cost; baseline profitability factor \\

\midrule
\multicolumn{2}{l}{\textit{Strategies \& Metrics}} \\
$r_1, r_2$ & Infiltration ratios in the initial and deterring states \\
$S$ & Target strategy, $S \in \{\mathrm{Mine}, \mathrm{SPV}, \mathrm{Stop}\}$ \\
$\pi_i^S, D_S$ & Steady-state probability; partition function \\
$p_{\mathcal{M}}^i, \theta_{\mathcal{M}}^S$ & Winning probability in racing state $i$; active-time ratio \\
$v_{\mathcal{M}}^S, s_{\mathcal{M}}^S$ & Effective block-generation rate; share-extraction rate \\
$e^S, U_{\mathcal{M}}^S$ & Reward-inflation coefficient; normalized utility \\
\bottomrule
\end{tabular}%
}
\end{table}

\section{The PDoS Attack}

\subsection{Overview}
The core idea of PDoS is to combine \emph{header signal deterrence} for liveness disruption with \emph{parasitic revenue extraction} from infiltration-based withholding. On the one hand, the attacker creates chain-state uncertainty by publishing only block headers, thereby reducing the expected payoff of rational miners that continue mining. On the other hand, the attacker continuously extracts share rewards from the victim pool to subsidize the attack cost. Under sufficiently profitable conditions, PDoS can therefore evolve into an economically sustainable and even profitable liveness attack.

\subsection{Method}
PDoS can be modeled as a \emph{Continuous Time Markov Chain} (CTMC) \cite{Continuous-time_Markov_chain} driven by \emph{signal injection}, \emph{dynamic power allocation}, and \emph{rational target response}. Algorithm~\ref{alg:pdos_controller} gives the controller logic.

\begin{algorithm}[t]
    \caption{PDoS Controller}
    \label{alg:pdos_controller}
    \small
    \begin{algorithmic}[1]
    \setlength{\itemsep}{0pt}
    \setlength{\parsep}{0pt}
    \setlength{\topsep}{0pt}
    \setstretch{0.1}
        \Require Hash power $\alpha,\beta$; Ratios $r_1, r_2$.
        \Ensure Maximize attacker's utility \& Minimize target's utility.
        \Procedure{Init}{} \MyComment{Initialize or reset system state}
            \State \StateUpdate{$\mathcal{S} \leftarrow 0$}, $B_{\mathrm{held}} \leftarrow \text{null}$, $\alpha_{\mathrm{inf}} \leftarrow r_1\alpha$, $\alpha_{\mathrm{pri}} \leftarrow (1-r_1)\alpha$
        \EndProcedure
        
        \State \Call{Init}{}
        \While{Mining Process Active}
            \State \textbf{Event:} New block $B_\mathrm{new}$ detected.
            \If{$\mathcal{S}=0$}
                \If{$B_\mathrm{new}$ found by $\alpha$}
                    \State $B_{\mathrm{held}} \leftarrow B_\mathrm{new}$
                    \State \AttackAction{Withhold} $\mathrm{Body}(B_{\mathrm{held}})$
                    \State \AttackAction{Broadcast} $\mathrm{Header}(B_{\mathrm{held}})$\MyComment{Deterrence}
                    \State \AttackAction{Adjust:} $\alpha_{\mathrm{inf}} \leftarrow r_2\alpha$, $\alpha_{\mathrm{pri}} \leftarrow 0$.
                    \State \StateUpdate{$\mathcal{S} \leftarrow 1$} \textbf{if} $B_{\mathrm{held}}$ found by $\alpha_{\mathrm{pri}}$ \textbf{else} \StateUpdate{$\mathcal{S} \leftarrow 2$}
                \Else \MyComment{Honest growth}
                    \State \textbf{continue}
                \EndIf
            \ElsIf{$\mathcal{S}\in\{1,2\}$}
                \If{$B_\mathrm{new}$ extends Header($B_{\mathrm{held}}$)} \MyComment{\textbf{SPV Trap}}
                \State \AttackAction{Discard} $B_{\mathrm{held}}$
                \State \Call{Init}{}
            \ElsIf{$\mathcal{S}=1$}
                \State \AttackAction{Adjust:} $\alpha_{\mathrm{inf}} \leftarrow 0$, $\alpha_{\mathrm{pri}} \leftarrow \alpha$.\MyComment{Max pri. win}
                \State \AttackAction{Release} $B_{\mathrm{held}}$ \MyComment{Trigger pri. race}
                \State \StateUpdate{$\mathcal{S} \leftarrow 3$} \textbf{if} $B_{\mathrm{new}}$ found by $\beta$ \textbf{else} \StateUpdate{$\mathcal{S} \leftarrow 4$}
            \Else \MyComment{$\mathcal{S}=2$}
                \If{$B_\mathrm{new}$ found by $\beta$}
                    \State \AttackAction{Discard} $B_{\mathrm{held}}$
                    \State \Call{Init}{}
                \Else
                    \State \AttackAction{Adjust:} $\alpha_{\mathrm{inf}} \leftarrow \alpha$, $\alpha_{\mathrm{pri}} \leftarrow 0$.\MyComment{Max inf. win}
                    \State \AttackAction{Submit} $B_{\mathrm{held}}$ \MyComment{Trigger inf. race}
                    \State \StateUpdate{$\mathcal{S} \leftarrow 5$}
                \EndIf
            \EndIf
            \Else \MyComment{$\mathcal{S} \in \{3,4,5\}$}
                \State \Call{Init}{}\MyComment{Network resolves race via $p_{\mathcal{A}}^i$}
            \EndIf
        \EndWhile
    \end{algorithmic}
\end{algorithm}

The lifecycle of PDoS consists of three phases: an \textit{initial phase}, a \textit{deterring phase}, and a \textit{racing phase}. In the initial phase, the attacker infiltrates the victim pool with part of its hash power to extract share rewards. Once the attacker gains a block lead, the system enters the deterring phase. At this point, the attacker exploits asymmetric publication to broadcast only the block header, injecting a credible yet incomplete state signal into the network and forcing target miners to reassess their expected payoff. The subsequent racing phase depends on the target miners' response strategy, and the attacker then chooses state-dependent actions to maximize liveness disruption.

Concretely, if the target miners continue normal mining, the attacker releases the previously withheld block body to trigger a branch race and thereby reduce the targets' expected payoff. If the targets attempt SPV mining on top of the published header, the attacker simply discards the block body, causing their output to become invalid because the parent block is unavailable. The combination of branch racing and the SPV trap creates a two-layer deterrent that drives rational target miners to shut down. The withdrawal of target hash power further reduces the network's effective block-production rate, prolongs the time spent in deterring states, and amplifies the attack's impact on chain liveness.

\subsection{Markov State Space Definition}\label{markov_state_space}
We model the evolution of PDoS as a CTMC. As shown in Figure~\ref{fig:states}, the state space $\mathcal{S}$ is partitioned into three subsets: the \emph{initial state} $\mathcal{S}_{\mathrm{ini}}=\{0\}$, the \emph{deterring states} $\mathcal{S}_{\mathrm{det}}=\{1,2\}$ in which $\mathcal{A}$ has obtained a lead, and the \emph{racing states} $\mathcal{S}_{\mathrm{rac}}=\{3,4,5\}$ in which branch competition is resolved through network propagation.

\begin{figure}[t]
  \centering
  \includegraphics[width=\linewidth]{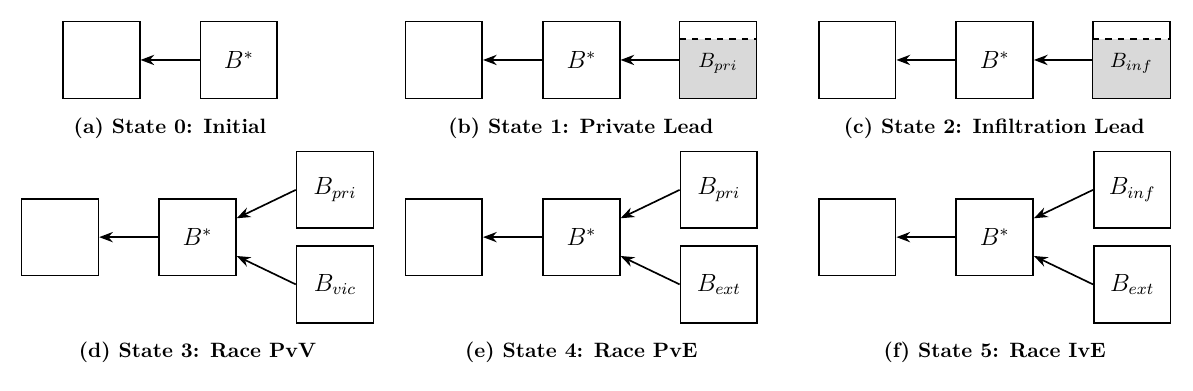}
  \Description{Markov states.}
  \caption{Markov-state topology of the PDoS attack. (a) Initial state. (b)--(c) Deterring states, in which the attacker broadcasts only the block header as a deterrence signal. (d)--(f) Racing states, in which the attacker releases the full block to compete.}
  \label{fig:states}
\end{figure}

\subsubsection{State Definitions}
\begin{itemize}
    \item \textbf{State 0 (Initial):} The baseline state in which all miners extend the public main chain and wait for the next block discovery.
    \item \textbf{State 1 (Private Lead):} $\mathcal{A}$'s private hash power finds a new block $B_{\mathrm{pri}}$ first. The attacker withholds the block body and broadcasts only the header, thereby injecting a deterrence signal into the network.
    \item \textbf{State 2 (Infiltration Lead):} $\mathcal{A}$'s infiltration hash power finds a new block $B_{\mathrm{inf}}$ inside $\mathcal{V}$. The attacker performs double withholding: it neither submits the FPoW to $\mathcal{V}$ nor broadcasts the block body to the network, and releases only the header to maintain deterrence.
    \item \textbf{States 3 and 4 (Race PvV/PvE):} During State~1, if $\mathcal{V}$ or $\mathcal{E}$ finds a new block, $\mathcal{A}$ immediately releases $B_{\mathrm{pri}}$ to trigger a race. State~3 corresponds to a race between the private block and the victim's block (PvV), whereas State~4 corresponds to a race between the private block and an external block (PvE).
    \item \textbf{State 5 (Race IvE):} During State~2, if $\mathcal{E}$ finds a new block $B_{\mathrm{ext}}$, $\mathcal{A}$ immediately submits $B_{\mathrm{inf}}$ to $\mathcal{V}$, thereby triggering a substantive race between the infiltration branch and the external branch.
\end{itemize}

\subsubsection{Dynamic Power Allocation}
$\mathcal{A}$ dynamically adjusts the split of its total hash power between private mining and infiltration mining according to the current Markov state, so as to balance parasitic revenue extraction against liveness degradation. The state-dependent allocation strategy is summarized in Table~\ref{tab:power_allocation}.

\begin{table}[t]
\centering
\small
\caption{Dynamic power-allocation strategies of $\mathcal{A}$.}
\label{tab:power_allocation}
\renewcommand{\arraystretch}{1.15}
\setlength{\tabcolsep}{5pt}
\resizebox{\columnwidth}{!}{%
\begin{tabular}{l c c c l}
\toprule
\textbf{Phase} & $\mathcal{S}$ & $r$ & $r'$ & \textbf{Objective} \\
\midrule
Initial & $0$ & $r_1$ & $1-r_1$ & Global optimization \\
Deterring & $1,2$ & $r_2$ & $0$ & Local optimization \\
PvV \& PvE & $3,4$ & $0$ & $1$ & Maximize private win rate \\
IvE & $5$ & $1$ & $0$ & Maximize infiltration win rate \\
\bottomrule
\end{tabular}%
}
\end{table}

In the initial state, $\mathcal{A}$ allocates a fraction $r_1$ of its hash power to infiltrate $\mathcal{V}$ and uses the remaining $1-r_1$ for private mining. Once the system enters a deterring state, $\mathcal{A}$ shuts down its private mining power to avoid competing with its own leading block and to reduce operating cost, while adjusting the infiltration ratio to $r_2$. The optimization of $r_1$ and $r_2$ is deferred to Appendix~\ref{app:optimization}.

In the racing states, $\mathcal{A}$ adopts extreme allocations to maximize branch-winning probability. In private-branch races (States~3 and~4), it withdraws all infiltration power ($r=0$) and devotes all attack power to the private branch. In the infiltration-branch race (State~5), it allocates all attack power to $\mathcal{V}$ ($r=1$). This not only lets the attacker combine the victim's hash power $\beta$ against external competition, but also maximizes its share payout if the infiltration branch wins.

\subsubsection{Race Winning Probabilities}
Let $p_{\mathcal{A}}^3$, $p_{\mathcal{A}}^4$, and $p_{\mathcal{A}}^5$ denote the winning probabilities of the $\mathcal{A}$-led branch in States~3, 4, and 5, respectively, as summarized in Table~\ref{tab:race_prob}. The honest branch probabilities are complementary, i.e.,
$
p_{\mathcal{V}}^3 = 1-p_{\mathcal{A}}^3, 
p_{\mathcal{E}}^4 = 1-p_{\mathcal{A}}^4, 
p_{\mathcal{E}}^5 = 1-p_{\mathcal{A}}^5$.
The derivation is given in Appendix~\ref{app:race_win_rate}.

\begin{table}[t]
\centering
\small
\caption{Race winning probabilities of the $\mathcal{A}$-led branch.}
\label{tab:race_prob}
\renewcommand{\arraystretch}{1.25}
\setlength{\tabcolsep}{6pt}
\resizebox{\columnwidth}{!}{%
\begin{tabular}{l c l l}
\toprule
\textbf{Phase} & $\mathcal{S}$ & \textbf{Win Rate} & \textbf{Tactical Advantage} \\
\midrule
PvV & $3$ & $p_{\mathcal{A}}^3 = \alpha + \gamma(\eta+\delta)$ & Full private utilization \\
PvE & $4$ & $p_{\mathcal{A}}^4 = \alpha + \gamma(\beta+\eta+\delta)$ & Full private utilization \\
IvE & $5$ & $p_{\mathcal{A}}^5 = \alpha + \beta + \gamma(\eta+\delta)$ & Victim-power hijacking \\
\bottomrule
\end{tabular}%
}
\end{table}

These winning probabilities reveal the key tactical asymmetry of PDoS. In PvV/PvE, $\mathcal{A}$ can fully deploy its own hash power $\alpha$ to support the private branch. In IvE, the victim pool is effectively forced to support the infiltration block. This hijacking effect gives the attacker a significant competitive advantage and correspondingly reduces the expected payoff of the honest branch.

\subsection{Evolution Process}
The system's evolution depends directly on the target response strategy $S$. Different decisions by $\mathcal{T}$ not only change the network's active hash power and therefore the transition rates, but also enable or disable specific evolution paths. The complete derivation of all microscopic transition rates is deferred to Appendix~\ref{app:state_transitions}.

Figure~\ref{fig:evolutions} abstracts four key block-evolution scenarios, while Figure~\ref{fig:markov_chains} shows how the three target strategies reshape the CTMC transition topology.

\begin{figure}[t]
  \centering
  \includegraphics[width=1.0\linewidth]{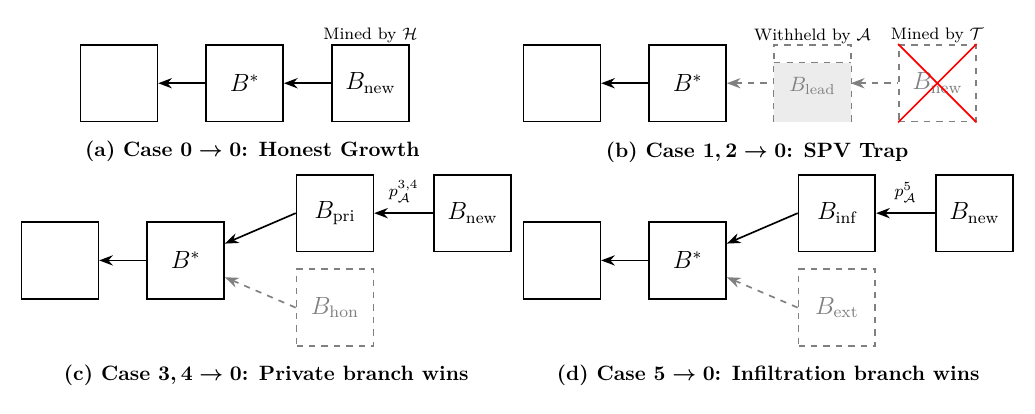}
  \caption{Block-evolution scenarios under PDoS. (a) Honest block growth. (b) SPV trap: the target mines on the attacker's header, but the output is invalidated because the parent body is withheld. (c) Private-race resolution in favor of the attacker. (d) Infiltration-race resolution in favor of the attacker through victim-power hijacking.}
  \Description{Markov evolutions.}
  \label{fig:evolutions}
\end{figure}

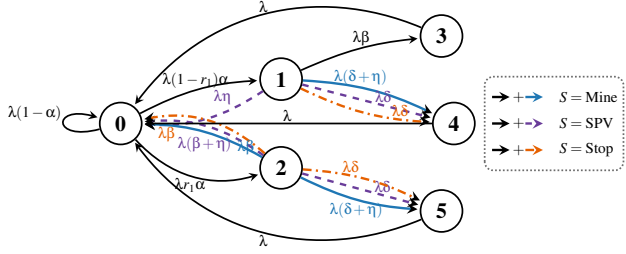
\begin{figure}[t]
    \centering
    \resizebox{\linewidth}{!}{\definecolor{minepath}{RGB}{31,119,180}
\definecolor{spvpath}{RGB}{106,61,154}
\definecolor{stoppath}{RGB}{230,97,1}

\begin{tikzpicture}[
    >=stealth,
    shorten >=1pt,
    thick,
    state_node/.style={
        circle,
        minimum size=0.68cm,
        inner sep=0pt,
        fill=white,
        draw=black,
        font=\normalsize\bfseries
    },
    rate/.style={
        font=\scriptsize,
        inner sep=0.8pt,
        fill=none
    },
    rate_mine/.style={rate, text=minepath},
    rate_spv/.style={rate, text=spvpath},
    rate_stop/.style={rate, text=stoppath},
    common_edge/.style={->, draw=black},
    mine_edge/.style={->, draw=minepath, line width=1.1pt},
    spv_edge/.style={->, draw=spvpath, dashed, line width=1.1pt},
    stop_edge/.style={
        ->,
        draw=stoppath,
        dash pattern=on 4pt off 2pt on 1pt off 2pt,
        line width=1.1pt
    },
    reset_edge/.style={->, draw=black},
    common_key/.style={->, draw=black, line width=1.1pt},
    mine_key/.style={mine_edge},
    spv_key/.style={spv_edge},
    stop_key/.style={stop_edge}
]

    \node[state_node] (s0) at (0, 0) {0};
    \node[state_node] (s1) at (2.6, 0.72) {1};
    \node[state_node] (s2) at (2.6, -0.72) {2};
    \node[state_node] (s3) at (5.2, 1.42) {3};
    \node[state_node] (s4) at (5.4, 0) {4};
    \node[state_node] (s5) at (5.2, -1.42) {5};

    \draw[reset_edge] (s0)
        edge[loop left, looseness=5, min distance=0.8cm]
        node[rate, above left, pos=0.5] {$\lambda(1-\alpha)$} (s0);

    \draw[common_edge] (s0)
        edge[bend left=12]
        node[rate, above, pos=0.48] {$\lambda(1-r_1)\alpha$} (s1);

    \draw[common_edge] (s0)
        edge[bend right=32]
        node[rate, below, pos=0.48] {$\lambda r_1\alpha$} (s2);

    \draw[common_edge] (s1)
        edge[bend left=9]
        node[rate, above, pos=0.52] {$\lambda\beta$} (s3);

    \draw[mine_edge] (s1)
        edge[bend left=11]
        node[rate_mine, above, pos=0.42] {$\lambda(\delta+\eta)$} (s4);
    \draw[spv_edge] (s1)
        edge node[rate_spv, above, pos=0.62] {$\lambda\delta$} (s4);
    \draw[stop_edge] (s1)
        edge[bend right=11]
        node[rate_stop, above, pos=0.78] {$\lambda\delta$} (s4);

    \draw[mine_edge] (s2)
        edge[bend right=11]
        node[rate_mine, below, pos=0.48] {$\lambda(\delta+\eta)$} (s5);
    \draw[spv_edge] (s2)
        edge node[rate_spv, above, pos=0.68] {$\lambda\delta$} (s5);
    \draw[stop_edge] (s2)
        edge[bend left=11]
        node[rate_stop, above, pos=0.38] {$\lambda\delta$} (s5);

    \draw[spv_edge] (s1) edge[bend left=18] (s0);
    \node[rate_spv] at (1.65, 0.43) {$\lambda\eta$};

    \draw[mine_edge] (s2) edge[bend right=14] (s0);
    \draw[spv_edge] (s2) edge[bend right=20] (s0);
    \draw[stop_edge] (s2) edge[bend right=26] (s0);
    \node[rate_stop] at (0.74, -0.17) {$\lambda\beta$};
    \node[rate_spv] at (1.36, -0.34) {$\lambda(\beta+\eta)$};
    \node[rate_mine] at (2.03, -0.38) {$\lambda\beta$};

    \draw[reset_edge] (s3)
        edge[bend right=38]
        node[rate, above, pos=0.50] {$\lambda$} (s0);

    \draw[reset_edge] (s4)
        edge node[rate, above, pos=0.50] {$\lambda$} (s0);

    \draw[reset_edge] (s5)
        edge[bend left=38]
        node[rate, below, pos=0.50] {$\lambda$} (s0);

    \draw[gray!85!black, densely dotted, rounded corners]
        (5.90, -0.76) rectangle (8.15, 0.76);
    \draw[common_key] (6.02, 0.43) -- (6.34, 0.43);
    \node[font=\scriptsize] at (6.45, 0.43) {$+$};
    \draw[mine_key] (6.56, 0.43) -- (6.88, 0.43);
    \node[anchor=west, font=\scriptsize] at (6.98, 0.43)
        {$S=\mathrm{Mine}$};

    \draw[common_key] (6.02, 0.00) -- (6.34, 0.00);
    \node[font=\scriptsize] at (6.45, 0.00) {$+$};
    \draw[spv_key] (6.56, 0.00) -- (6.88, 0.00);
    \node[anchor=west, font=\scriptsize] at (6.98, 0.00)
        {$S=\mathrm{SPV}$};

    \draw[common_key] (6.02, -0.43) -- (6.34, -0.43);
    \node[font=\scriptsize] at (6.45, -0.43) {$+$};
    \draw[stop_key] (6.56, -0.43) -- (6.88, -0.43);
    \node[anchor=west, font=\scriptsize] at (6.98, -0.43)
        {$S=\mathrm{Stop}$};

\end{tikzpicture}}
    \caption{Unified CTMC transition topology for $S\in\{\mathrm{Mine},\mathrm{SPV},\mathrm{Stop}\}$. Each right-hand key entry combines the black shared transitions with one strategy-conditioned family: Mine is blue solid, SPV is purple dashed, and Stop is orange dash--dot. Every colored curve carries its instantiated rate; equal rates remain labeled separately when they belong to distinct strategies. The $1\!\to\!0$ transition exists only for SPV. All black edges are shared by the three scenarios.}
    \Description{One CTMC with strategy-conditioned transition rates for Mine, SPV, and Stop.}
    \label{fig:markov_chains}
\end{figure}

\subsection{Steady-State Distribution}\label{sec:ctmc_steadystate}
Let $\boldsymbol{\pi}^S=[\pi_0^S,\dots,\pi_5^S]$ denote the CTMC steady-state distribution under target strategy $S$, where $\pi_i^S$ is the long-run time proportion spent in state $i$. This distribution is obtained by solving the global balance equations under the corresponding strategy. The full algebraic derivation for all three scenarios is deferred to Appendix~\ref{app:steady_state}.

The steady-state solutions under the three scenarios share a common denominator, which we extract and denote by the partition function $D_S$. This term captures the effective \emph{evolution damping} induced by a given target strategy. Specifically,
$
D_{\mathrm{Mine}} = 1 + \alpha(1-\alpha-r_1\beta), 
D_{\mathrm{SPV}} = D_{\mathrm{Mine}} - \alpha\eta,
~\text{and}~
D_{\mathrm{Stop}} = D_{\mathrm{SPV}} - \eta$.

\begin{lemma}[\textbf{Deterring-State Prolongation}]\label{lem:prob_order}
For the three target response strategies $S\in\{\mathrm{Mine},\mathrm{SPV},\mathrm{Stop}\}$, the total steady-state probability of remaining in the deterring states,
$
\pi_{\mathrm{det}}^S \triangleq \pi_1^S+\pi_2^S$,
satisfies the strict ordering
$
\pi_{\mathrm{det}}^{\mathrm{Mine}} < \pi_{\mathrm{det}}^{\mathrm{SPV}} < \pi_{\mathrm{det}}^{\mathrm{Stop}}$.
\end{lemma}

\begin{proof}
From the closed-form steady-state distribution derived in Appendix~\ref{app:steady_state}, the total probability of the deterring states can be written uniformly as
$
\pi_{\mathrm{det}}^S = \frac{\alpha}{D_S}$.
Because $\alpha,\eta>0$, the partition functions satisfy
$
D_{\mathrm{Mine}} > D_{\mathrm{SPV}} > D_{\mathrm{Stop}} > 0$.
Since $g(D)=\alpha/D$ is strictly decreasing for $D>0$, the claimed ordering follows immediately.
\end{proof}

\subsection{Expected Utility Analysis}\label{sec:expected_utility}
Using a \textit{Markov reward process} (MRP) \cite{markove_model}, we derive the expected utility of each party under different target response strategies. To make the reward flow explicit, we extract three normalized quantities with direct economic meaning. The complete derivation is given in Appendix~\ref{app:formal_extraction}.

\begin{itemize}
    \item \textbf{Effective block-generation rate ($v_{\mathcal{M}}^S$):} the expected main-chain-confirmed block output per unit hash power.
    \item \textbf{Share-extraction rate ($s_{\mathcal{M}}^S$):} the expected parasitic revenue obtained through PPoW-based infiltration of the victim pool.
    \item \textbf{Active-time ratio ($\theta_{\mathcal{M}}^S\in[0,1]$):} the fraction of time during which a miner remains online and incurs operating cost.
\end{itemize}

The normalized utility of any participant can then be written in the unified closed form
\begin{equation}\label{eq:u_m}
    U_{\mathcal{M}}^S = c \Big[ \omega_b \big( v_{\mathcal{M}}^S + s_{\mathcal{M}}^S \big) - \theta_{\mathcal{M}}^S \Big].
\end{equation}

\heading{State-Dependent MEV Inflation.}
Before evaluating utilities, we must account for reward inflation in the deterring states. During deterrence, partial shutdown or invalidation of hash power prolongs block intervals and thereby increases the amount of time-varying MEV accumulated in a single block.

Let $\Delta \alpha_{\mathrm{act}}^S$ denote the fraction of active-hash-power vacuum in the deterring states, and let $M\in[0,1)$ denote the share of MEV in total block value. We define the reward-inflation coefficient as
$
    e^S = 1 + M\left(\frac{\Delta \alpha_{\mathrm{act}}^S}{1-\Delta \alpha_{\mathrm{act}}^S}\right)\ge 1$,
with the full derivation deferred to Appendix~\ref{app:mev_inflation}. Because only target shutdown creates a more severe active-hash-power vacuum, we have $e^{\mathrm{Stop}} > e^{\mathrm{SPV}} = e^{\mathrm{Mine}}$. Thus, defensive retreat by the target endogenously increases the value of blocks generated during deterrence.

\subsubsection{Utility of the Target}
Because $\mathcal{T}$ has no share-extraction capability, $s_{\mathcal{T}}^S=0$, and its normalized utility reduces to
$
U_{\mathcal{T}}^S = c\left(\omega_b v_{\mathcal{T}}^S - \theta_{\mathcal{T}}^S\right)$.

\heading{Scenario 1: Normal Mining ($S=\mathrm{Mine}$).}
The target remains active throughout, so $\theta_{\mathcal{T}}^{\mathrm{Mine}}=1$. Its effective block-generation rate consists of two parts: (i) race-resolving blocks produced in non-deterring states and immediately confirmed on chain; and (ii) triggering blocks produced in deterring states that later win the race:
\begin{equation}\label{eq:v_t_mine}
    v_{\mathcal{T}}^{\mathrm{Mine}} = \sum_{i\in\{0,3,4,5\}} \pi_i^{\mathrm{Mine}} + e^{\mathrm{Mine}}\Big(\pi_1^{\mathrm{Mine}}p_{\mathcal{E}}^4 + \pi_2^{\mathrm{Mine}}p_{\mathcal{E}}^5\Big).
\end{equation}
Its expected utility is 
$
U_{\mathcal{T}}^{\mathrm{Mine}} = c\left(\omega_b v_{\mathcal{T}}^{\mathrm{Mine}} - 1\right)$.

\heading{Scenario 2: SPV Mining ($S=\mathrm{SPV}$).}
The target remains online, so $\theta_{\mathcal{T}}^{\mathrm{SPV}}=1$. However, any SPV-mined output produced during deterrence is necessarily invalid because the parent body is unavailable. Its effective block-generation rate becomes
$
v_{\mathcal{T}}^{\mathrm{SPV}} = \sum_{i\in\{0,3,4,5\}} \pi_i^{\mathrm{SPV}}$,
and its expected utility is
$
U_{\mathcal{T}}^{\mathrm{SPV}} = c\left(\omega_b v_{\mathcal{T}}^{\mathrm{SPV}} - 1\right)$.

Although the steady-state probabilities differ across the two scenarios, Appendix~\ref{app:spv_dominated} proves that
$
U_{\mathcal{T}}^{\mathrm{SPV}} < U_{\mathcal{T}}^{\mathrm{Mine}}$
under all parameter settings. Hence, in game-theoretic terms, $S=\mathrm{SPV}$ is a \emph{strictly dominated strategy}. By iterated elimination of strictly dominated strategies (IESDS), rational miners will abandon it after observing losses. In the remaining analysis, we therefore restrict the effective target strategy space to $S\in\{\mathrm{Mine},\mathrm{Stop}\}$.

\heading{Scenario 3: Shutdown ($S=\mathrm{Stop}$).}
The target disconnects during deterrence, so its active-time ratio shrinks to the non-deterring states. Because it does not participate in external branch competition while offline, its effective block-generation rate equals its active-time ratio:
$
v_{\mathcal{T}}^{\mathrm{Stop}} = \theta_{\mathcal{T}}^{\mathrm{Stop}}
= \sum_{i\in\{0,3,4,5\}} \pi_i^{\mathrm{Stop}}$.
Its expected utility becomes
$
    U_{\mathcal{T}}^{\mathrm{Stop}} = c\cdot \theta_{\mathcal{T}}^{\mathrm{Stop}}(\omega_b-1)
    = \theta_{\mathcal{T}}^{\mathrm{Stop}}\cdot U_{\mathcal{T}}^b$.
Thus, the payoff of shutdown is proportional to both the fraction of time spent outside deterrence and the baseline honest-mining utility.

\subsubsection{Utility of the Adversary}
The attacker's revenue stream combines direct block rewards and parasitic share extraction. To maximize capital efficiency, $\mathcal{A}$ shuts down only its private mining power during deterrence, while keeping infiltration power active throughout. Its global active-time ratio is therefore
$
\theta_{\mathcal{A}}^S = 1-(\pi_1^S+\pi_2^S)(1-r_2)$.
Substituting the two core revenue coefficients below into the unified utility form in Equation~\eqref{eq:u_m} yields the attacker's full theoretical utility under each state.

\heading{Effective Block-Generation Rate.}
The attacker's effective block-generation rate contains two terms: (i) race-resolving blocks confirmed on chain in private-branch racing states; and (ii) triggering blocks found earlier that later win the corresponding private race: $
    v_{\mathcal{A}}^S = \pi_3^S + \pi_4^S + \frac{1}{\alpha}\big(\pi_3^S p_{\mathcal{A}}^3 + \pi_4^S p_{\mathcal{A}}^4\big)$.

\heading{Share-Extraction Rate.}
Share extraction is the core mechanism that allows PDoS to externalize attack cost and sustain the attack lifecycle. Define the local time-weighted infiltration ratio by
\[
\bar{r}_i^S=\frac{r_1\pi_0^S+r_2\pi_i^S}{\pi_0^S+\pi_i^S},
\qquad i\in\{1,2\},
\]
and define the normalized share-extraction function by
$
f(r)=\frac{r}{\beta+r\alpha}$.
Then the share-extraction rate can be written as
\begin{equation*}
    s_{\mathcal{A}}^S =
    \pi_0^S \beta f(r_1)
    + e^S\Big[\pi_3^S p_{\mathcal{V}}^3 f(\bar{r}_1^S) + \pi_2^S \beta f(\bar{r}_2^S)\Big]
    + \pi_5^S\Big[p_{\mathcal{A}}^5 f(\bar{r}_2^S) + 1\Big].
\end{equation*}

\section{Cost-Benefit Analysis}
In this section, we show that PDoS significantly lowers the hash-power threshold required to induce network shutdown, increases the attacker's expected utility through infiltration, and thereby enables longer-lasting disruption under a fixed budget. We further reveal an asymmetric incentive effect induced by high-MEV environments.

\subsection{Attack Threshold Dominance}
To quantify and compare the deterrence power of different attacks, we first introduce the \emph{shutdown utility gap} of the target miner. Define the difference between the normalized utilities of $\mathcal{T}$ under $S=\mathrm{Mine}$ and $S=\mathrm{Stop}$ as
\begin{equation}\label{eq:u_t_diff}
    \Delta U_{\mathcal{T}}^\mathrm{Stop}
    \;\triangleq\;
    U_{\mathcal{T}}^{\mathrm{Mine}} - U_{\mathcal{T}}^{\mathrm{Stop}}
    \;=\;
    c \Big[ \omega_b \big( v_{\mathcal{T}}^{\mathrm{Mine}} - v_{\mathcal{T}}^{\mathrm{Stop}} \big) - \big( 1 - v_{\mathcal{T}}^{\mathrm{Stop}} \big) \Big].
\end{equation}

\begin{definition}[\textbf{Critical Attacker Hash Power $\alpha_{\mathcal{A}}^*$}]\label{def:alpha_crit}
The minimum attacker hash power required to flip the target's optimal response from $S=\mathrm{Mine}$ to $S=\mathrm{Stop}$ is defined as the infimum of $\alpha$ such that the shutdown utility gap becomes strictly negative:
\begin{equation*}
    \alpha_{\mathcal{A}}^*
    \;\triangleq\;
    \inf \left\{ \alpha \mid \Delta U_{\mathcal{T}}^\mathrm{Stop}(\alpha) < 0 \right\}.
\end{equation*}
\end{definition}

This quantity characterizes the security threshold against liveness-oriented deterrence. A smaller $\alpha_{\mathcal{A}}^*$ implies that the network can be paralyzed by a smaller fraction of malicious hash power. The following theorem establishes the threshold advantage of PDoS; the full proof is deferred to Appendix~\ref{app:threshold-proof}.

\begin{theorem}[\textbf{Threshold Dominance}]\label{thm:threshold_dominance}
Under the same network and economic parameters, the critical attacker hash power of PDoS is never larger than that of BDoS, i.e.,
$\alpha_{\mathcal{A}}^* \big|_{\mathrm{PDoS}}
    \;\le\;
    \alpha_{\mathcal{A}}^* \big|_{\mathrm{BDoS}}$.
\end{theorem}

\subsection{Net Cost Dominance}
To quantify the economic burden of the attack, we introduce the \emph{net cost rate}. When the target adopts strategy $S$, define the attacker's net cost rate $\mathcal{C}_{\mathcal{A}}^S$ as the negative of its normalized utility:
\begin{equation}\label{eq:net_cost}
    \mathcal{C}_{\mathcal{A}}^S
    \;\triangleq\;
    -U_{\mathcal{A}}^S
    \;=\;
    c \Big[ \theta_{\mathcal{A}}^S - \omega_b \big( v_{\mathcal{A}}^S + s_{\mathcal{A}}^S \big) \Big].
\end{equation}

\begin{definition}[\textbf{Minimum Attack Net-Cost Rate $\mathcal{C}_{\mathcal{A}}^*$}]\label{def:min_cost}
The minimum normalized net expenditure that the attacker must bear while sustaining the attack is defined as
$\mathcal{C}_{\mathcal{A}}^*
    \;\triangleq\;
    \min \mathcal{C}_{\mathcal{A}}^S$.
\end{definition}

When $\mathcal{C}_{\mathcal{A}}^* \le 0$, the attacker has crossed the break-even point, and the attack becomes a victim-funded self-sustaining attack. The following theorem shows that PDoS dominates BDoS in terms of minimum attack net cost; the full proof is deferred to Appendix~\ref{app:net-cost-proof}.

\begin{theorem}[\textbf{Net Cost Dominance}]\label{thm:net_cost_dominance}
Under the same network and economic parameters, the minimum attack net-cost rate of PDoS is never higher than that of BDoS, i.e.,
$\mathcal{C}_{\mathcal{A}}^* \big|_{\mathrm{PDoS}}(\alpha)
    \;\le\;
    \mathcal{C}_{\mathcal{A}}^* \big|_{\mathrm{BDoS}}(\alpha)$.
\end{theorem}

\subsection{Attack Endurance}
Theorems~\ref{thm:threshold_dominance} and~\ref{thm:net_cost_dominance} establish the advantage of PDoS in both attack threshold and attack cost. In practice, however, rational target miners typically do not shut down immediately; instead, they undergo a period of observation and strategic adjustment. As a result, the eventual success of the attack depends critically on whether the attacker can sustain deterrence before exhausting its budget. We therefore introduce a budget-constrained notion of attack endurance to quantify how cost advantage translates into temporal persistence.

\begin{definition}[\textbf{Attack Endurance $T_{\mathcal{A}}^*$}]\label{def:attack_endurance}
Given an initial budget $W$, the maximum time window for which the attacker can sustain deterrence is defined by
\begin{equation*}
    T_{\mathcal{A}}^*
    \;\triangleq\;
    \begin{cases}
      \frac{W}{\mathcal{C}_{\mathcal{A}}^*}, & \text{if } \mathcal{C}_{\mathcal{A}}^* > 0,\\
      \infty, & \text{if } \mathcal{C}_{\mathcal{A}}^* \le 0.
    \end{cases}
\end{equation*}
\end{definition}

When $\mathcal{C}_{\mathcal{A}}^* > 0$, the attack is a conventional budget-burning attritional strategy. By Theorem~\ref{thm:net_cost_dominance}, since
$
\mathcal{C}_{\mathcal{A}}^* \big|_{\mathrm{PDoS}}
\le
\mathcal{C}_{\mathcal{A}}^* \big|_{\mathrm{BDoS}}$,
it follows that
$
T_{\mathcal{A}}^* \big|_{\mathrm{PDoS}}
\ge
T_{\mathcal{A}}^* \big|_{\mathrm{BDoS}}$.
Thus, under the same budget, PDoS yields a longer attack window.

Moreover, holding one deployed policy $(r_1,r_2)$ and the reward-composition parameter $M$ fixed, differentiating Equation~\eqref{eq:net_cost} with respect to the baseline profitability factor $\omega_b$ gives
$
\frac{\partial \mathcal{C}_{\mathcal{A}}^S}{\partial \omega_b}
=
-c\big(v_{\mathcal{A}}^S+s_{\mathcal{A}}^S\big)
<0$.
Thus, for a fixed feasible policy, higher profitability reduces the attacker's net cost. For classical BDoS, where $s_{\mathcal{A}}^S=0$, this reduction is relatively mild. By contrast, PDoS can capture a proportional share of booming-market rewards through infiltration, yielding a steeper cost-reduction slope. We refer to this effect as the \emph{asymmetric bull-market subsidy effect}. When an economic or defense parameter changes, the feasible set can shrink and the maximizing policy can switch even though every fixed-policy cost curve remains monotone. Figure~\ref{fig:defenses}(c) exposes this active-policy switch under conditional payout. PDoS nevertheless crosses its first break-even point earlier under a milder profitability regime and can obtain victim-funded persistence more easily than classical liveness attacks.

\subsection{Self-Sustaining Attack}
Crossing the break-even point is the key condition for self-sustainability. To characterize the boundary of self-sustaining PDoS, we introduce the following definition and theorem; the full derivation and proof are deferred to Appendix~\ref{app:self-sustain-proof}.

\begin{definition}[\textbf{Joint Self-Sustaining Threshold $\Omega_{\mathcal{A}}^*$}]\label{def:self_sustain_bound}
The minimum profitability factor at which one \emph{same} configuration both deters and breaks even is
\begin{equation*}
    \Omega_{\mathcal{A}}^*
    \triangleq
    \inf\{\omega_b:\mathcal{F}_{D}(\omega_b)\neq\emptyset
    \ \wedge\ \mathcal{C}_{\mathcal{A}}^*(\omega_b)\le0\}.
\end{equation*}
\end{definition}

\begin{theorem}[\textbf{Self-Sustaining Dominance}]\label{thm:self-sustain}
Under the same network environment, the joint profitability threshold required for a deterrence-effective, self-sustaining PDoS is never higher than for BDoS. It is strictly lower whenever $G^{\mathrm{Stop}}(0,0)>0$ at the BDoS break-even point and the constructed perturbed policy remains deterrence-feasible, i.e.,
$\Omega_{\mathcal{A}}^* \big|_{\mathrm{PDoS}}
    \;\le\;
    \Omega_{\mathcal{A}}^* \big|_{\mathrm{BDoS}}$.
\end{theorem}

This theorem identifies the attacker's survival threshold. By lowering the profitability threshold for self-sustainability, PDoS can achieve victim-funded zero-cost disruption in market regimes where BDoS is still engaged in budget burning.

\subsection{Deterrence Boundary}
Although PDoS may be persistent, the ultimate goal of deterrence is to force the target to shut down. High-profitability environments are therefore a double-edged sword: they subsidize the attacker, but they also strengthen the target miners' incentive to remain online. To characterize the validity boundary of the attack, we introduce the following definition and theorem; the full derivation and proof are deferred to Appendix~\ref{app:deterrence-proof}.

\begin{definition}[\textbf{Deterrence-Validity Upper Bound $\Omega_{\mathcal{T}}^*$}]\label{def:deterrence_bound}
The maximum baseline profitability factor under which the attacker can still force the target miner to choose $S=\mathrm{Stop}$ as its optimal response is defined as
\begin{equation*}
    \Omega_{\mathcal{T}}^*
    \;\triangleq\;
    \max_{r_1,r_2}
    \frac{1-v_{\mathcal{T}}^{\mathrm{Stop}}(r_1,r_2)}
         {v_{\mathcal{T}}^{\mathrm{Mine}}(r_1,r_2)-v_{\mathcal{T}}^{\mathrm{Stop}}(r_1,r_2)}.
\end{equation*}
A necessary and sufficient condition for PDoS to remain effective is
$\omega_b < \Omega_{\mathcal{T}}^*$.
\end{definition}

\begin{theorem}[\textbf{Deterrence Resilience}]\label{thm:deterrence_resilience}
Under the same network environment, the profitability upper bound under which PDoS remains an effective deterrent is never lower than that of classical BDoS. If $\beta>0$ and some $r_1>0$ satisfies $\Gamma(r_1)>0$ at the BDoS validity boundary, the inequality is strict, i.e.,
$\Omega_{\mathcal{T}}^* \big|_{\mathrm{PDoS}}
    \;\ge\;
    \Omega_{\mathcal{T}}^* \big|_{\mathrm{BDoS}}$.
\end{theorem}

This theorem identifies the upper boundary of attack validity. By raising the profitability threshold at which the target refuses to shut down, PDoS remains effective in market regimes where BDoS already fails because the target keeps mining. Even in the ineffective region $\omega_b \ge \Omega_{\mathcal{T}}^*$, PDoS still earns more through infiltration; once profitability falls back into the effective region, PDoS enters the subsequent war of attrition with a larger budget. This establishes a theoretical advantage of PDoS across the full market cycle.

\section{Game-Theoretic Analysis}\label{sec:game-theoretic_analysis}

In this section, we extend the preceding homogeneous-strategy analysis to a heterogeneous micro-level evolutionary game. We first establish the existence of a unique strict Nash equilibrium, then characterize the liveness-degradation process that ultimately drives the system to collapse, and finally analyze the robustness of PDoS in the presence of realistic frictional costs.

\subsection{Nash Equilibrium}

\subsubsection{Micro-Miner Strategy Space}
Under a non-atomic game formulation, the independent decisions of micro-miners $\mathcal{T}_i$ (with $\eta_i \to 0$) jointly determine the aggregate active-hash-power ratio of the target population. Because SPV mining is a strictly dominated strategy (Appendix~\ref{app:spv_dominated}), and this dominance is preserved under the linear payoff scaling of micro-miners, the effective strategy space of each $\mathcal{T}_i$ reduces to
$
S_{\mathcal{T}_i} \in \{\mathrm{Mine}, \mathrm{Stop}\}$.
We define a continuous variable $x \in [0,1]$ as the fraction of target hash power that currently chooses normal mining, i.e.,
$
x = \frac{1}{\eta}\sum_{S_{\mathcal{T}_i}=\mathrm{Mine}} \eta_i$.

\subsubsection{Generalized Steady State}
Under partial shutdown, when the active fraction of the target population is $x$, the effective hash power available to advance the blockchain becomes
$
\alpha_{\mathrm{act}}(x) = \alpha + \beta + x\eta + \delta = 1-(1-x)\eta$.
As $x$ decreases, the growing hash-power vacuum continuously shifts the CTMC transition rates. Consequently, the steady-state distribution of the system becomes fully parameterized by $x$; the full derivation is given in Appendix~\ref{app:scenario4_derivation}.

\subsubsection{Uniqueness of Nash Equilibrium}
Under the non-atomic game assumption, all micro-miners are payoff-symmetric. We therefore characterize the stability of the game by the shutdown utility gap of an individual micro-miner,
$
\Delta U_{\mathcal{T}_i}^{\mathrm{Stop}}(x)
\triangleq
U_{\mathcal{T}_i}^{\mathrm{Mine}}(x) - U_{\mathcal{T}_i}^{\mathrm{Stop}}(x)$.

\begin{lemma}[\textbf{Monotonicity of the Micro-Miner Shutdown Utility Gap}]
\label{lemma:mono-utility}
Once the attacker hash power exceeds the critical threshold, i.e., $\alpha > \alpha_{\mathcal{A}}^*$, the shutdown utility gap of a micro-miner is strictly increasing in the active ratio $x$, namely,
$
\frac{d \Delta U_{\mathcal{T}_i}^{\mathrm{Stop}}(x)}{dx} > 0$.
\end{lemma}

This lemma reveals a positive feedback effect: as more miners shut down due to losses (i.e., as $x$ decreases), the enlarged hash-power vacuum prolongs the time spent in deterring states, which in turn further worsens the relative payoff of the miners that remain active. The proof is deferred to Appendix~\ref{app:mono-utility-proof}. We can therefore state the equilibrium characterization of the game; the full proof is given in Appendix~\ref{app:unique-ne-proof}.

\begin{theorem}[\textbf{Unique Strict Nash Equilibrium}]
\label{thm:unique-ne}
If the attacker hash power exceeds the critical threshold, i.e., $\alpha > \alpha_{\mathcal{A}}^*$, then $\mathrm{Stop}$ is a strictly dominant strategy for every micro-miner $\mathcal{T}_i$. Consequently, the strategy profile
$
\mathbf{S}^* = (\mathrm{Stop},\dots,\mathrm{Stop})$
is the unique strict Nash equilibrium of the game.
\end{theorem}

\subsection{Evolutionary Dynamics}

Although the system ultimately converges to the all-stop Nash equilibrium, real decentralized networks exhibit delayed collective responses due to information delay, strategic inertia, and bounded rationality.

To capture the dynamic trajectory of liveness degradation, we promote the active target-hash-power ratio to a continuous-time function $x(t)$. Under an evolutionary-game formulation, each $\mathcal{T}_i$ adjusts its hash-power deployment according to the current utility gap, and the population dynamics are described by the replicator equation. Let
$
\bar{U}(x)=x(t)U_{\mathcal{T}_i}^{\mathrm{Mine}}(x)+\bigl(1-x(t)\bigr)U_{\mathcal{T}_i}^{\mathrm{Stop}}(x)$
denote the population-average utility. Then the evolution of the fraction choosing normal mining satisfies
\begin{align}\label{eq:replicator}
    \dot{x}(t)
    &\triangleq
    x(t)\Big[U_{\mathcal{T}_i}^{\mathrm{Mine}}\bigl(x(t)\bigr)-\bar{U}\bigl(x(t)\bigr)\Big] \notag\\
    &=
    x(t)\bigl(1-x(t)\bigr)\Delta U_{\mathcal{T}_i}^{\mathrm{Stop}}\bigl(x(t)\bigr).
\end{align}

The following theorem establishes global convergence and the accelerated-collapse property; the full proof is deferred to Appendix~\ref{app:death-spiral-proof}.

\begin{theorem}[\textbf{Global Convergence and Accelerated Collapse}]
\label{thm:death-spiral}
If the attacker hash power exceeds the critical threshold, i.e., $\alpha > \alpha_{\mathcal{A}}^*$, then for any initial active ratio $x(0)\in(0,1)$, the trajectory $x(t)$ is strictly decreasing and converges to the all-stop terminal state:
$
\lim_{t\to\infty} x(t)=0$.
Moreover, the hash-power vacuum effect induced by PDoS introduces an additional negative acceleration into the evolutionary process.
\end{theorem}

\subsection{Robustness to Frictional Costs}\label{sec:epsilon_equilibrium}

The preceding results are derived under a frictionless assumption. In real PoW mining, however, electricity-contract penalties, hardware depreciation, and switching overhead create substantial frictional costs. These practical barriers dampen strategy switching, causing miners to tolerate mild losses rather than immediately shut down.

To evaluate the robustness of PDoS under realistic conditions, we introduce an $\epsilon$-equilibrium framework. Suppose that switching to $\mathrm{Stop}$ requires a micro-miner $\mathcal{T}_i$ to overcome an equivalent normalized frictional cost $\epsilon>0$. Then $\mathrm{Stop}$ becomes a strictly dominant strategy only when the relative loss from continued mining exceeds this frictional cost, i.e.,
$
\Delta U_{\mathcal{T}_i}^{\mathrm{Stop}}(x) < -\epsilon$.
Assuming that the single-block MEV inflation premium is economically bounded, we obtain the following theorem on friction penetration and comparative robustness; the full proof is deferred to Appendix~\ref{app:epsilon-robustness-proof}.

\begin{theorem}[\textbf{Friction Penetration and Robustness}]
\label{thm:epsilon-robustness}
For any realistic frictional cost $\epsilon \in (0,c)$, there exists a critical attacker hash-power threshold $\alpha_{\mathcal{A}}^\epsilon$ above which PDoS can penetrate the friction barrier and trigger accelerated liveness collapse. Moreover, under the same network parameters, the absolute shutdown utility gap induced by PDoS is strictly larger than that induced by BDoS:
\begin{equation*}
\left|\Delta U_{\mathcal{T}_i}^{\mathrm{Stop}}(\alpha)\right|_{\mathrm{PDoS}}
>
\left|\Delta U_{\mathcal{T}_i}^{\mathrm{Stop}}(\alpha)\right|_{\mathrm{BDoS}}.
\end{equation*}
\end{theorem}

This result shows that, due to the revenue subsidization provided by the underlying infiltration mechanism, PDoS can penetrate higher frictional barriers and exhibits stronger deterrence robustness in realistic environments.

\begin{table}[t]
  \centering
  \caption{Summary of dataset statistics and estimated parameters.}
  \label{tab:combined_stats}
  \resizebox{\linewidth}{!}{
  \begin{tabular}{lrlr}
    \toprule
    \multicolumn{2}{c}{\textbf{Dataset Statistics}} & \multicolumn{2}{c}{\textbf{Parameter Estimation}} \\
    \cmidrule(lr){1-2} \cmidrule(lr){3-4}
    \textbf{Metric} & \textbf{Value} & \textbf{Parameter} & \textbf{Value} \\
    \midrule
    Total Blocks & 106,422 & Tx Fee Accum. Rate ($\lambda_t$) & $4.0886 \times 10^{-5}$ BTC/s \\
    Block Interval Avg. & 593.11 s & Whale Arrival Rate ($\lambda_w$) & $1.9248 \times 10^{-6}$ /s \\
    Block Interval Med. & 414 s & Exp. Whale Reward ($\mathbb{E}[\mathcal{R}_w]$) & 0.1019 BTC \\
    Block Interval P95 & 1766 s & MEV Ratio ($M$) & 0.0078 \\
    Tx Fee Avg. & 0.1589 BTC & Block Mining Cost ($c_{\mathrm{block}}$) & 2.0097 BTC \\
    Tx Fee P99.9 & 6.1936 BTC & Profitability Factor ($\omega_b$) & 1.6149 \\
    \bottomrule
  \end{tabular}
  }
\end{table}

\begin{figure}[t]
  \centering
  \includegraphics[width=\linewidth]{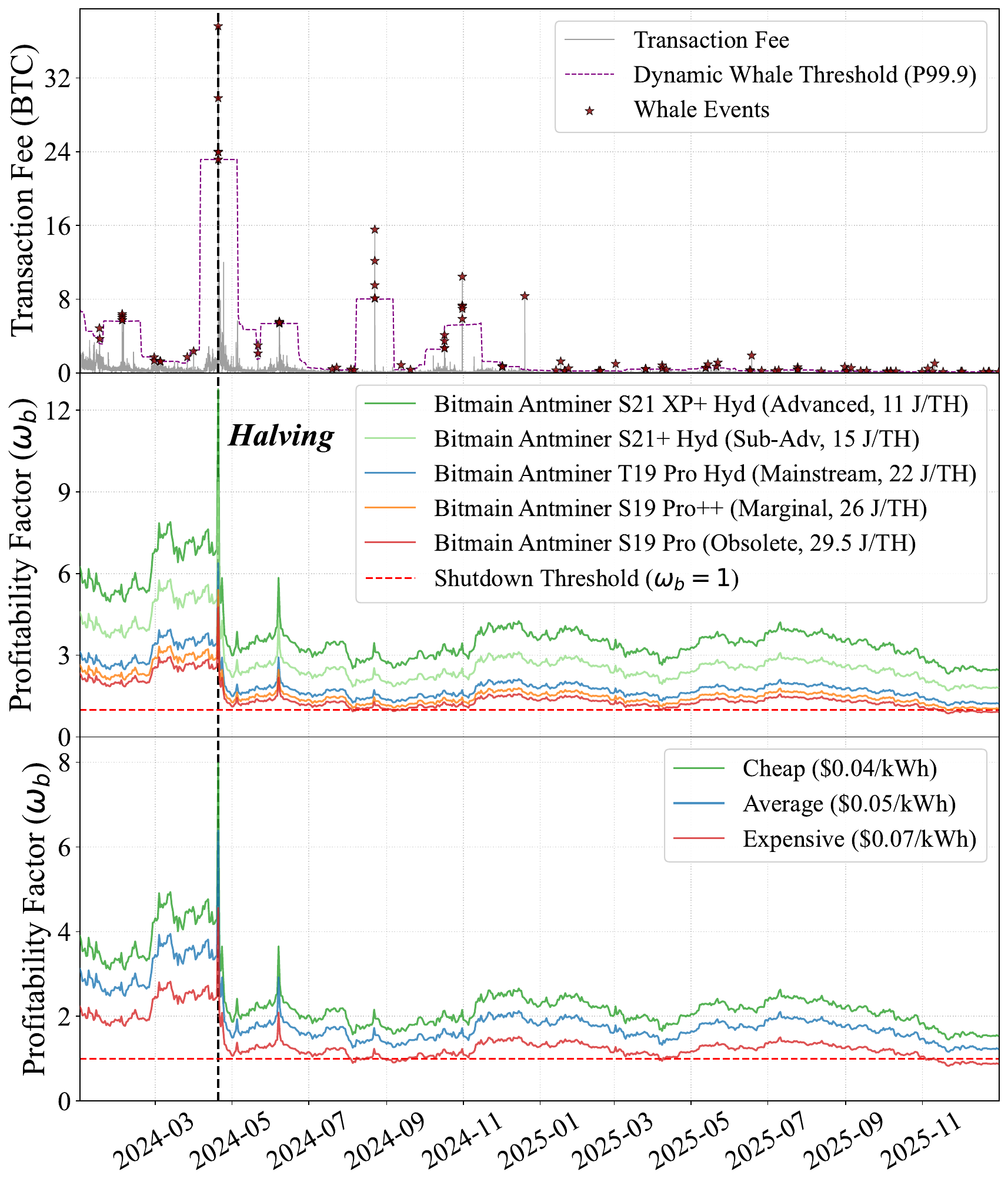}
  \caption{Real-world incentive environment and evolution of the baseline profitability factor. The black dashed line marks Bitcoin's fourth halving. (a) Transaction-fee distribution and high-value block events. (b) Comparison of $\omega_b$ across ASIC generations under the baseline electricity price of $0.05$ USD/kWh. (c) Comparison of $\omega_b$ across industrial electricity prices under the baseline hardware efficiency of $22$ J/TH.}
  \Description{Figure 1.}
  \label{fig:fig1}
\end{figure}

\section{Evaluation}

\subsection{Dataset and Parameter Estimation}\label{sec:dataset}

\heading{Dataset.}
To evaluate PDoS and the BDoS baseline under realistic constraints, we collect on-chain and market data covering Bitcoin's fourth halving cycle, from January 1, 2024 to December 30, 2025. Bitcoin block data is obtained from the public Google Cloud BigQuery dataset~\cite{google_bigquery_btc}; the BTC--USD historical price series is obtained from Investing.com~\cite{investing_btc_usd}; and the real-world mining-pool hash-power distribution is derived from mempool.space~\cite{mempool_pools_2y}. For hardware and energy costs, we use the industrial electricity-price benchmark from the Cambridge Bitcoin Electricity Consumption Index (CCAF)~\cite{ccaf_electricity}, adopt mainstream ASIC efficiency parameters (J/TH) from ASIC Miner Value~\cite{asic_miner_value}, and incorporate an additional 15\% operating expenditure (OPEX) following CoinShares industry reports~\cite{coinshares_q4_2024,coinshares_q4_2025}. This additional cost accounts for cooling, maintenance, and administrative overhead.

\heading{Real-World Profitability.}
The resilience of miners against liveness attacks is fundamentally governed by the baseline profitability factor $\omega_b$. Figure~\ref{fig:fig1} shows the temporal variation in transaction fees during our evaluation window and highlights the decisive impact of hardware efficiency and electricity prices on $\omega_b$. In particular, after Bitcoin's fourth halving, the abrupt reduction in block subsidy persistently compresses the profitability of miners using older hardware or facing high electricity prices, often pushing their $\omega_b$ close to the shutdown threshold $\omega_b=1$. This vulnerable hash power is therefore the first to be forced offline under attack.

\heading{Statistics and Parameters.}
Table~\ref{tab:combined_stats} summarizes the dataset statistics and the key parameter estimates used in our Markov model. The on-chain statistics on the left highlight the high volatility of Bitcoin's reward environment, which creates asymmetric profit opportunities for PDoS. The parameter estimates on the right provide the empirical economic baseline, which is specifically calibrated using the stabilized period in the second half of 2025 (where $K=3.125$ BTC), including the MEV ratio $M$ and the per-block economic mining cost $c_{\mathrm{block}}$. Unless otherwise specified, we use the following default parameters in the subsequent experiments: attacker hash power $\alpha=0.15$, victim-pool hash power $\beta=0.2$, target-population hash power $\eta=0.1$, propagation advantage $\gamma=0.5$, and baseline profitability factor $\omega_b=1.6$. For visual consistency, we use \textit{blue dashed curves} for the BDoS baseline and \textit{red solid curves} for PDoS. The color saturation of each curve indicates the attacker hash-power level $\alpha$, with darker colors representing larger $\alpha$.

\subsection{Attack Critical Threshold}

We first derive the analytical critical attacker hash power $\alpha_{\mathcal{A}}^*$ from the CTMC model, and then validate it using local checkpoint simulations. For each checkpoint, we run a Monte Carlo discrete-event simulation with $10^6$ block events and extract the stabilized output using polynomial fitting. The simulation checkpoints closely match the analytical boundary.

As shown in Figure~\ref{fig:fig2}, $\alpha_{\mathcal{A}}^*$ increases monotonically with the profitability factor $\omega_b$, while a larger propagation advantage $\gamma$ substantially lowers the required attack threshold. Across the entire boundary distribution, the PDoS threshold consistently lies below that of BDoS, empirically supporting Theorem~\ref{thm:threshold_dominance}.

\begin{figure}[t]
  \centering
  \includegraphics[width=\linewidth]{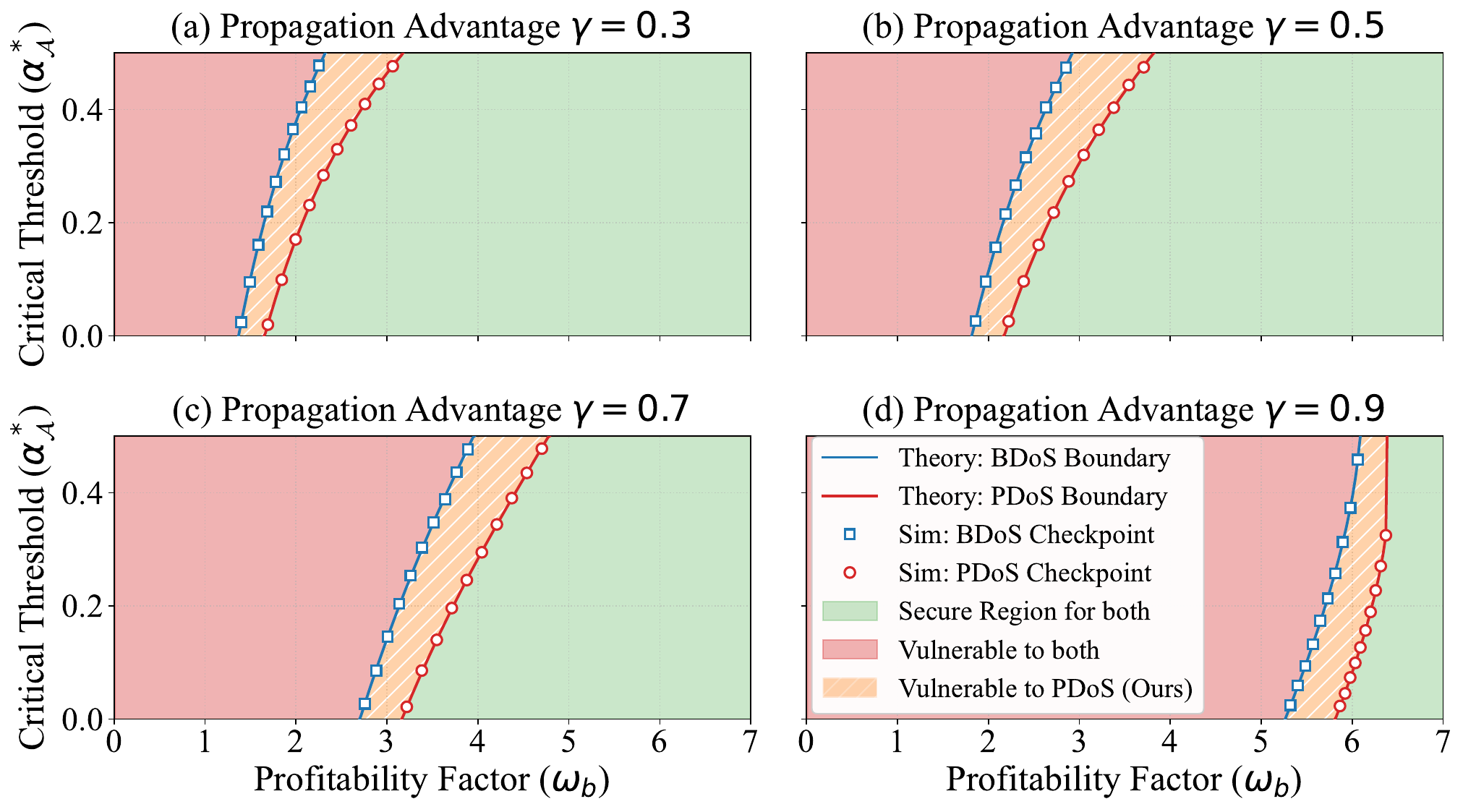}
  \caption{Phase diagram of the critical attacker hash power $\alpha_{\mathcal{A}}^*$ as a function of the profitability factor $\omega_b$ under different propagation advantages $\gamma$. The PDoS threshold consistently lies below the BDoS threshold, and the orange shaded region visualizes the enlarged feasible attack space enabled by PDoS.}
  \Description{Figure 2.}
  \label{fig:fig2}
\end{figure}

\begin{figure}[t]
  \centering
  \includegraphics[width=\linewidth]{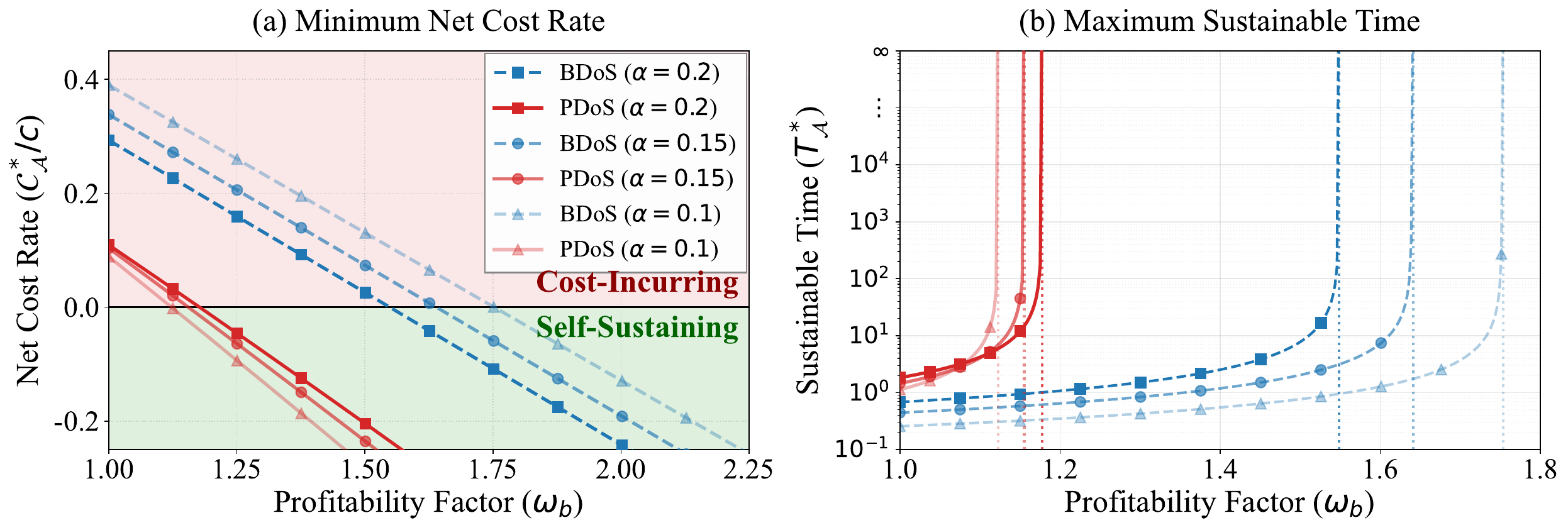}
  \caption{Comparison of attack cost and sustainable duration. (a) Normalized net cost rate $\mathcal{C}_{\mathcal{A}}^*/c$ as a function of $\omega_b$. (b) Maximum sustainable time under the same budget.}
  \Description{Figure 3.}
  \label{fig:cost_time}
\end{figure}

\subsection{Net Cost and Attack Sustainability}

To quantify the attacker's capital consumption rate and the sustainable attack duration, we compare three attacker hash-power levels, $\alpha \in \{0.10,0.15,0.20\}$.

\heading{Self-Sustaining Zone.}
As shown in Figure~\ref{fig:cost_time}(a), the horizontal boundary $\mathcal{C}_{\mathcal{A}}^*=0$ separates the cost-incurring zone from the self-sustaining zone. Under the same conditions, the net cost rate of PDoS is consistently lower than that of BDoS. PDoS already enters the self-sustaining region under a relatively low-profitability environment around $\omega_b \approx 1.15$, whereas BDoS requires $\omega_b>1.5$ to break even. This result corroborates the cost-dominance result in Theorem~\ref{thm:net_cost_dominance}.

\heading{Maximum Sustainable Time.}
Figure~\ref{fig:cost_time}(b) shows that the expected sustainable time of PDoS grows by orders of magnitude relative to BDoS on a logarithmic scale. PDoS also lowers the market-profitability threshold required for indefinite attack persistence, validating the sustainability result in Theorem~\ref{thm:self-sustain}.

\heading{hash-power Inversion.}
As $\alpha$ decreases, the self-sustaining threshold of BDoS shifts to the right, whereas that of PDoS shifts to the left. This counterintuitive hash-power inversion arises because the infiltration-based share-reward mechanism provides small attackers with higher marginal compensation, allowing them to cross the self-sustaining boundary at a lower $\omega_b$.

\subsection{Reward and Scheduling Ablations}

Figure~\ref{fig:ablations}(a) holds economic mining cost and the 3.125 BTC coinbase reward fixed, then sweeps the non-coinbase share of total block reward from 0 to 5\%. The requested coinbase-only and empirical coinbase-plus-fee/whale cases are marked on the same continuous sensitivity curve. The empirical non-coinbase share is only $M=0.783\%$, so removing it changes the optimized net cost from $-0.343c$ to $-0.333c$ (2.8\% of the empirical cost magnitude). Bursty fees strengthen PDoS at the margin, but the victim-funded share of the coinbase reward already creates most of the subsidy channel.

Figures~\ref{fig:ablations}(b)--(c) enumerate all 36 fixed schedules on the finer grid $r_1,r_2\in\{0,0.2,\ldots,1\}$ at the default $\omega_b=1.6$. Panel~(b) reports total net cost. Positive $r_2$ is wasteful at this baseline because the exact settlement-weighted marginal utility satisfies $G^{\mathrm{Stop}}(r_1,r_2)<0$ throughout the sampled range; additional payout does not offset operating cost, so the optimizer powers down after a lead. Panel~(c) directly tests whether state-dependent scheduling matters. Each cell reports the cost saving relative to the corresponding static policy $r_2=r_1$, i.e., $\mathcal C_{\mathcal A}(r_1,r_1)-\mathcal C_{\mathcal A}(r_1,r_2)$. The optimized $(r_1,r_2)=(1,0)$ policy saves $0.112c$ over the static $(1,1)$ policy; conversely, continuing to infiltrate after a lead produces negative savings below the zero-saving diagonal. The gold outline marks the one deployed joint optimizer.

\begin{figure}[t]
  \centering
  \includegraphics[width=\linewidth]{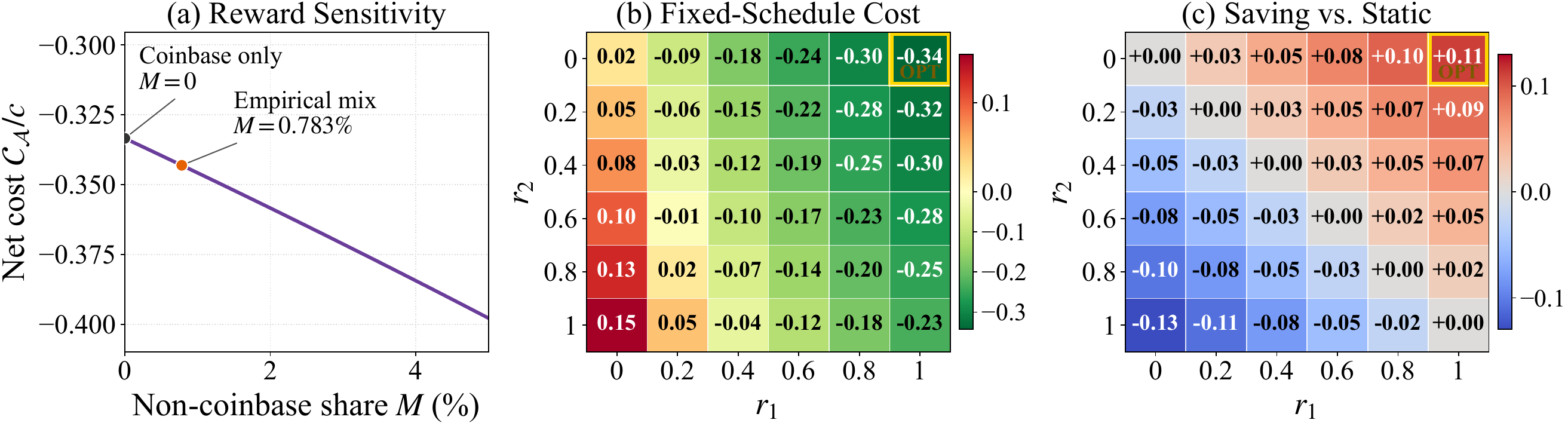}
  \caption{Reward and scheduling ablations. (a) Optimized net cost as the non-coinbase reward share varies; the coinbase-only and empirical $M=0.783\%$ points are marked. (b) Net cost for the 36 fixed schedules on the $0.2$-spaced $r_1,r_2$ grid at $\omega_b=1.6$. (c) Cost saving of each state-dependent schedule relative to its static counterpart $r_2=r_1$; positive values favor dynamic scheduling. The gold outline is the deployed joint optimizer.}
  \Description{Continuous reward-source sensitivity, fixed-schedule net cost, and dynamic-scheduling saving.}
  \label{fig:ablations}
\end{figure}

\subsection{Evolutionary and Equilibrium Analysis}

\heading{Evolutionary Trajectories.}
Figure~\ref{fig:fig4}(a) uses the replicator dynamics equation to characterize the decay trajectory of the target survival ratio $x$ over evolutionary time $t$. The results show that PDoS causes a substantially faster loss of active target hash power than BDoS. Using the explicit numerical convergence criterion $x\le10^{-3}$, PDoS reaches the terminal neighborhood at $t=64.1$, compared with $t=97.4$ for BDoS, a $1.5\times$ acceleration that validates Theorem~\ref{thm:death-spiral}.

\heading{Trace-Driven Nash Equilibrium.}
Figure~\ref{fig:fig4}(b) shows the final target survival ratio after the system reaches Nash equilibrium under different attacker hash-power levels and the real-world hash-power distribution. We instantiate a discrete pool-level model and replay the on-chain data to compute the instantaneous shutdown boundary of each independent mining pool. The results show that PDoS consistently yields a lower terminal survival ratio than BDoS.

\subsection{Trace-Driven Dynamic Evaluation}

To evaluate the long-term behavior of PDoS under realistic economic fluctuations, we replay the historical block sequence using a trace-driven simulation parameterized by the real-world hash-power distribution. We aggregate the trace into 730 daily locally stationary slices: each slice uses the observed mean profitability and reward-weighted fee fraction, and its utility contribution is multiplied by the exact number of blocks observed that day. This captures difficulty and reward evolution without presenting block-level flicker as a persistent equilibrium change. For each slice, we compute the Nash-equilibrium survival ratio $x$ under attacker hash-power levels $\alpha\in\{0.1,0.15,0.2\}$ using the preceding slice's decision environment. The policy selected from that preceding slice is then held fixed while the current slice's realized rewards are settled, preventing one-day look-ahead.

\heading{Network Liveness Evolution.}
The upper panel of Figure~\ref{fig:trace_dynamic} shows that PDoS consistently induces a larger loss of active hash power than BDoS. After the halving, the survival ratio of the target population drops substantially, indicating that PDoS has stronger deterrence power under post-halving profitability compression.

\heading{Divergence of Attacker Cumulative Utility.}
As shown in the lower panel of Figure~\ref{fig:trace_dynamic}, both PDoS and BDoS accumulate utility during the high-profit period before the halving. After the halving, the mean daily BDoS increment becomes negative for every evaluated $\alpha$, so its cumulative curve trends downward. PDoS retains a positive post-halving mean increment; all 730 daily increments are nonnegative for $\alpha\in\{0.10,0.15\}$ and 728 of 730 are nonnegative for $\alpha=0.20$. The orange region reports the cumulative-utility advantage of PDoS over BDoS at the baseline $\alpha=0.15$.

\begin{figure}[t]
  \centering
  \includegraphics[width=\linewidth]{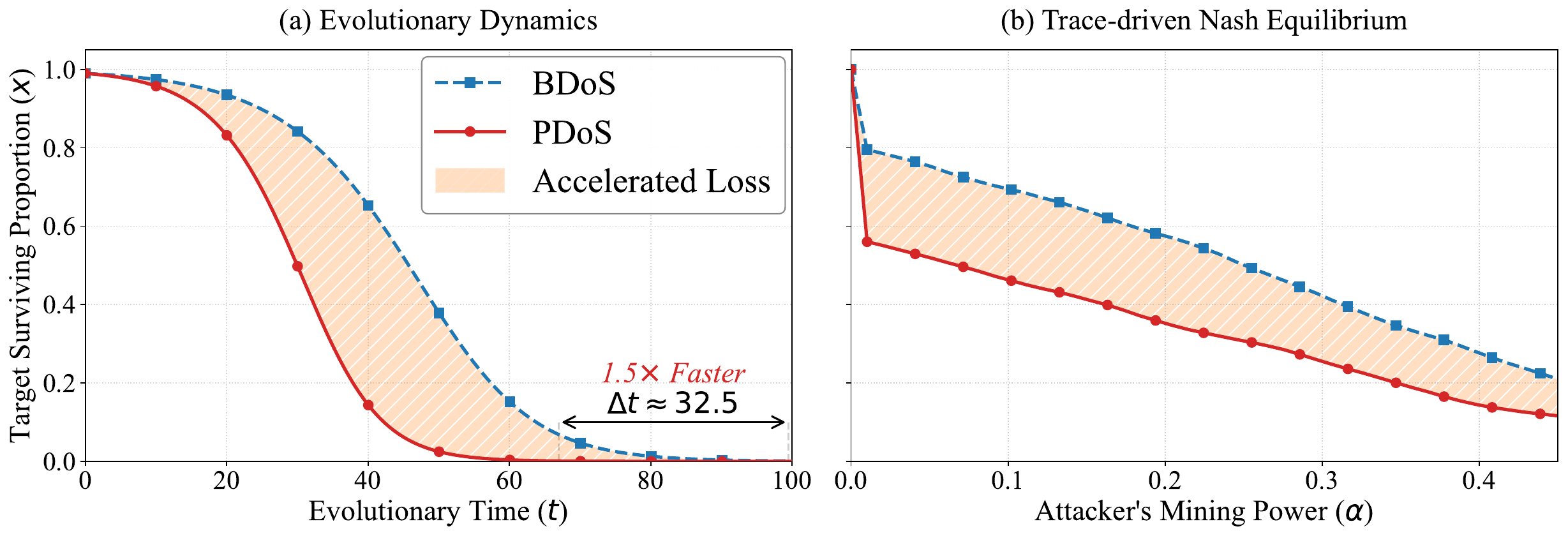}
  \caption{Comparison of system-liveness degradation. The orange shaded region quantifies the accelerated loss induced by PDoS. (a) Liveness-degradation trajectories based on evolutionary game theory. (b) Trace-driven terminal distribution of Nash equilibrium states.}
  \Description{Figure 4.}
  \label{fig:fig4}
\end{figure}

\begin{figure}[t]
    \centering
    \includegraphics[width=\linewidth]{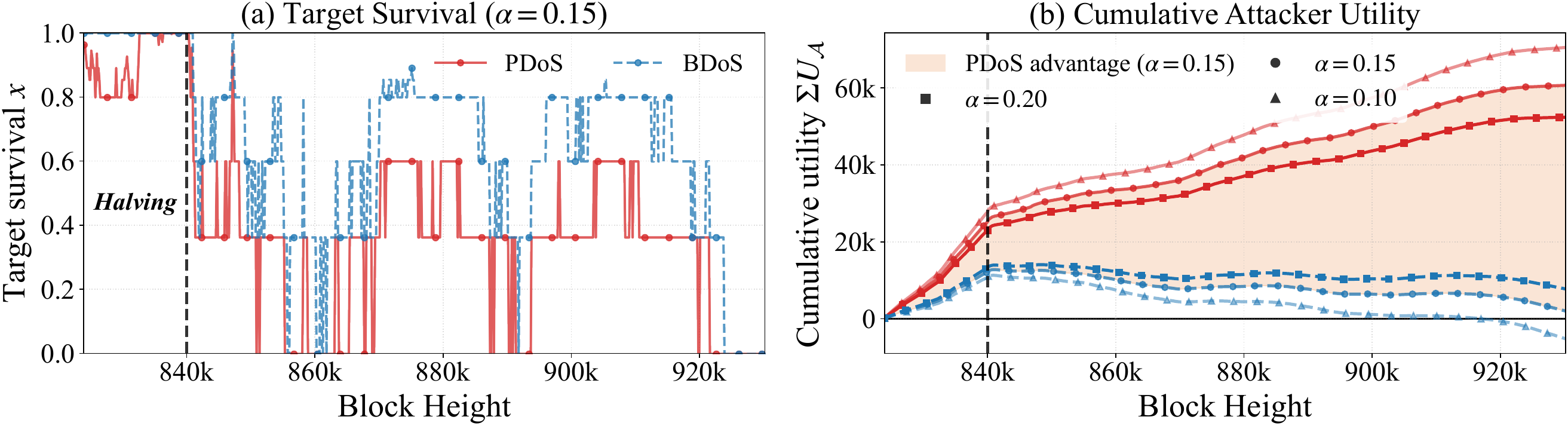}
    \caption{Trace-driven dynamic game evaluation over the 2024--2025 historical period. The black dashed line marks Bitcoin's fourth halving. The upper panel shows the time-varying target survival ratio $x$ under the baseline attacker hash power $\alpha=0.15$. The lower panel compares the attacker's cumulative utility $\sum U_{\mathcal{A}}$ under different strategies. The green shaded region indicates the profitable regime of PDoS, the gray region shows the declining utility trend of BDoS, and the red region marks the net-negative cumulative-utility regime of BDoS.}
    \Description{Figure 5.}
    \label{fig:trace_dynamic}
\end{figure}

\subsection{Withholding-Defense Evaluation}\label{sec:defense-evaluation}
We parameterize pool hardening by a \emph{released-share fraction} $q\in[0,1]$: if an account is flagged before its escrow matures, only fraction $q$ of its accrued share reward is paid. This replaces $s_{\mathcal{A}}^S$ by $q s_{\mathcal{A}}^S$ while leaving direct race rewards unchanged. Figure~\ref{fig:defenses}(a) first reports the no-defense policy at $q=1$ and $\omega_b=1.6$. Increasing $r_1$ from 0 to 1 lowers normalized net cost from $0.024$ to $-0.343$ and amplifies the target's loss gap by 67.5\%.

Figure~\ref{fig:defenses}(b) evaluates the impact of PPLNS on short-term payout risk and long-term subsidy. Under the standard independent equal-difficulty-share approximation, let $K\sim\mathrm{Binomial}(N,F)$ be the number of attacker shares in an effective window of $N$ shares, where $F(r)=r\alpha/(\beta+r\alpha)$. Then
\(
\mathbb E[K/N]=F~\text{and}~\operatorname{CV}(K/N)=\sqrt{\frac{1-F}{NF}}.
\)
PPLNS therefore reduces finite-window payout variance as $N$ grows but does not reduce the persistent infiltrator's mean subsidy.

For conditional payout, a useful defense metric must require profitability and liveness effectiveness simultaneously. Let $\Omega_{\mathcal A}(r_1,q)$ be the lower profitability boundary at which the attack is self-sustaining (with the exact inner $r_2$ optimizer), and let $\Omega_{\mathcal T}(r_1)$ be the upper boundary below which the same policy deters. We define the profitable-liveness window
\[
W_P(q)=\max_{r_1\in[0,1]}
       [\Omega_{\mathcal T}(r_1)-\Omega_{\mathcal A}(r_1,q)]_+,
\]
and the BDoS reference $W_B=[\Omega_{\mathcal T}(0)-\Omega_{\mathcal A}(0,q)]_+=0.419$, which is independent of $q$ because BDoS submits no shares. Figure~\ref{fig:defenses}(c) plots the surplus $W_P(q)-W_B$ and the maximizing $r_1$. Conditional payout removes PDoS's additional feasible window for $q\lesssim0.429$; above this transition, the optimizer moves continuously away from BDoS and reaches $r_1=1$ near $q=0.62$. With immediate payout, the surplus window reaches $0.967$. Thus, delayed forfeiture can eliminate PDoS's subsidy advantage, whereas PPLNS alone changes only payout risk. Effective release should be conditioned on statistically credible FPoW contribution or auditable work~\cite{lerner2026apow}.

\begin{figure}[t]
  \centering
  \includegraphics[width=\linewidth]{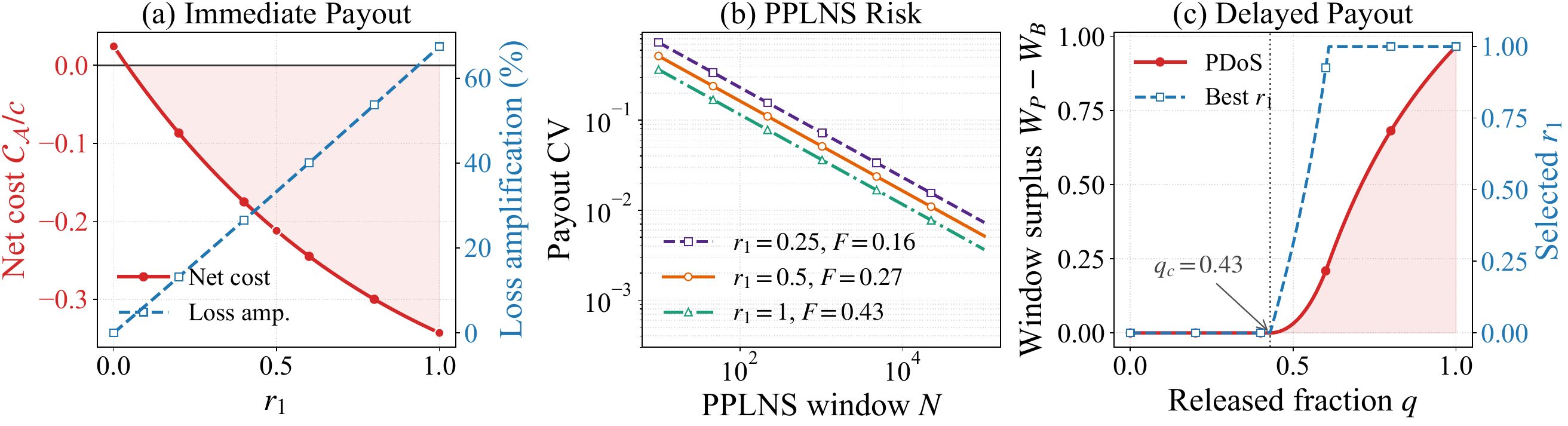}
  \caption{Mechanism-specific withholding-defense evaluation. (a) Under immediate payout, the same $r_1$ policy is evaluated by attacker net cost and target-loss amplification. (b) Stationary PPLNS leaves the mean share fraction unchanged but changes its finite-window coefficient of variation. (c) Delayed/conditional payout: surplus profitable-liveness window over BDoS and the corresponding maximizing $r_1$ versus released fraction $q$.}
  \Description{No-defense policy, PPLNS finite-window risk, and conditional-payout feasible-window reduction.}
  \label{fig:defenses}
\end{figure}

\section{Discussion and Countermeasures}\label{sec:defenses}

\heading{The Impact of the SPV Trap.}
Viewed through the lens of an incomplete-information game, a target cannot immediately distinguish malicious body withholding from benign relay delay. The protocol-level separation is concrete: mining jobs expose header search space and accept later proof submissions~\cite{stratumv2mining}, compact relay can initially leave transaction data missing~\cite{bip152}, and both BDoS and perishing mining analyze incentives created by an available proof-bearing header without a validated body~\cite{mirkin2020bdos,cao2023verifier}. In real-world dataset (see \footnote{\url{https://anonymous.4open.science/r/PDoS-23C3/dataset/block_metrics.csv}}), we observed 278 empty blocks with zero transaction fees, corresponding to an empty-block rate of 0.26\%, with the last such instance at height 929546. Empty blocks corroborate the operational incentive to begin hashing before a profitable template is fully available.

\heading{Remembering the Header and Concurrent Mining.}
A correct miner can remember the isolated header while continuing on a fully validated tip, or split power between both choices. Let $\kappa$ be the fraction kept on normal validated mining and $1-\kappa$ the fraction exposed to the header-only branch. Under linear per-hash rewards and no switching overhead, its mixed payoff is
\(
U_{\mathcal{T}}^{\mathrm{mix}}(\kappa)
=\kappa U_{\mathcal{T}}^{\mathrm{Mine}}
+(1-\kappa)U_{\mathcal{T}}^{\mathrm{SPV}}
\le U_{\mathcal{T}}^{\mathrm{Mine}},
\)
because Appendix~\ref{app:spv_dominated} proves $U^{\mathrm{SPV}}<U^{\mathrm{Mine}}$. Thus, remembering and hedging weakly reduces the transient SPV loss but does not create a response better than normal mining, and it does not remove the Mine-versus-Stop deterrence boundary. The generalized survival fraction $x$ allows different miners to update at different times; within-miner switching costs would make the inequality stricter.

We next discuss three classes of countermeasures against PDoS: hardening pool payout rules, shortening the reward-inflation window after liveness degradation, and reducing race-time propagation asymmetry.

\heading{Modifying the Share-Payout.}
The key advantage of PDoS over BDoS is that a public victim pool funds the attacker. Mining in an attacker-owned private pool would avoid detection by an external operator, but it would also eliminate victim-funded share transfers and therefore does not instantiate the PDoS subsidy. In a public pool, PPLNS or time-decay scoring raises variance and discourages pool hopping~\cite{rosenfeld2011pooled}; as Figure~\ref{fig:defenses}(b) shows, changing $N$ scales the finite-window payout risk as $N^{-1/2}$ but does not lower a persistent infiltrator's stationary expected share fraction. Effective hardening must instead condition release on evidence correlated with useful FPoWs: probation, delayed escrow, cross-account aggregation, and statistical tests of the share-to-FPoW ratio. Natural Poisson variance prevents immediate proof of withholding and Sybil rotation extends the detection delay. Auditable work constructions can strengthen this evidence at the cost of mining-protocol or ASIC changes~\cite{lerner2026apow}.

\heading{Shortening the Reward-Inflation Window.}
PDoS exploits a positive feedback loop:
target shutdown
$\rightarrow
\alpha_{\mathrm{act}}$ 
decreases
$\rightarrow$
longer block interval
$\rightarrow$
larger per-block reward
$\rightarrow$
lower attacker cost.
A second class of defenses aims to weaken this feedback loop. Faster or smoother difficulty adjustment can shorten the abnormal high-reward window under reduced active hash power. Protocols may also consider GHOST-style or uncle-block reward mechanisms that partially compensate losing blocks in races, thereby narrowing the utility gap between continuing to mine and shutting down. These mechanisms, however, require protocol-level redesign and may introduce nontrivial side effects. Overly responsive difficulty adjustment can create new manipulation opportunities, while rewards for losing branches may incentivize additional strategic forking behavior.

\section{Conclusion}\label{Conclusion}
We present PDoS, a hybrid attack against PoW blockchain liveness. Through both theoretical analysis and experimental evaluation, we show that PDoS substantially lowers the hash-power threshold required to induce rational miner shutdown compared with BDoS, and that in high-MEV environments, it can cross the break-even point to become self-sustaining or even profitable. More broadly, our results show that increasingly lucrative reward environments do not necessarily strengthen PoW security; instead, they can also subsidize adversarial strategies that amplify liveness fragility. Our findings have been disclosed to the Bitcoin Core Security Team.

\cleardoublepage
\appendix
\section*{Ethical Considerations}
\textbf{We are committed to complying with all relevant research ethics guidelines, and this study involves no ethical risks. The research content is not associated with animal experimentation, human subjects, environmental protection, medical health, or military applications. We confirm full adherence to all ethical principles outlined in the Call for Papers (CFPs). We pledge to proactively address any ethical concerns that may arise during the review process or post-publication, welcome academic feedback regarding research ethics norms, and will refine the research protocol based on such input.}

\section*{Open Science}
\textbf{We fully endorse and strictly adhere to the Open Science Policy outlined in the Call for Papers (CFP). The formal proofs for all theorems and lemmas presented in this paper are provided in the Appendix. To facilitate reproducibility and independent validation, we provide an anonymized implementation of our analysis and evaluation framework at:
\url{https://anonymous.4open.science/r/PDoS-23C3}. We hereby grant the scientific community unrestricted rights to review, validate, and build upon this research work.}
\cleardoublepage

\bibliographystyle{plainurl}
\bibliography{Ref}

\section{Related Work}

\heading{Nakamoto Consensus Security.}
A large body of work has analyzed the security of PoW longest-chain protocols. The Bitcoin backbone line of work formalizes the consistency and liveness properties of Nakamoto consensus~\cite{garay2015backbone,garay2020fullanalysis,pass2017analysis,pass2017sleepy,pass2017rethinking,thehaltgame}. Subsequent studies refine PoW security bounds under tighter consistency, settlement latency, bounded capacity, and optimal attack models~\cite{kiffer2018better,gazi2020tight,li2021close,guo2022latency,gazi2022practical,gazi2023practical,dembo2020race,kiffer2024bounded,cao2025beat,cao2025security,doger2025refined}. These works mainly ask when blocks can be safely confirmed and how security trades off with latency and throughput. In contrast, we study whether an adversary can degrade liveness through miner-incentive manipulation even when the protocol remains formally correct.

\heading{Profit-Driven Mining Deviations.}
Selfish mining and its variants show that strategic block withholding and timed release can yield excess revenue with minority hash power~\cite{eyal2014majority,nayak2016stubborn,sapirshtein2016optimal,barzur2020efficient,marmolejo2019competing,li2021semiselfish,negy2020selfishreexamined,bai2023multiple,BM-PAW,han2025collapsing,timestampmanipulationtimestampbasednakamotostyle}. Related work further studies gap games, Ethereum-specific fork-choice rules, and time-averaged revenue effects~\cite{tsabary2018gap,feng2019selfishethereum,sarenche2024timeaveraged}. PDoS differs in objective: it does not use strategic mining primarily to maximize relative revenue, but uses revenue extraction as a funding source for liveness disruption.

\heading{Pool Infiltration and Withholding.}
Mining-pool reward mechanisms create a separate attack surface. Block withholding lets an attacker submit PPoWs to earn shares while discarding FPoWs that would otherwise solve blocks~\cite{rosenfeld2011pooled,eyal2015miners,luu2015powersplitting}. FAW, PAW, and subsequent variants combine infiltration with fork competition, bribery, extortion, and power adjustment~\cite{kwon2017faw,gao2019paw,yang2022fwap,hu2023bribery,luu2017smartpool}. Defenses range from payout-window changes and statistical monitoring to auditable work constructions~\cite{rosenfeld2011pooled,lerner2026apow}. Unlike the revenue-oriented attacks, PDoS treats infiltration as an \emph{attack subsidy and race control}: the victim pool funds liveness disruption and, on an infiltration-originated fork, supports its own compromised template.

\heading{DoS Attacks on Blockchains.}
BDoS shows that an attacker can depress rational miners' expected payoff by proving the existence of a block without revealing its body, thereby inducing miner shutdown~\cite{mirkin2020bdos}. Perishing-mining attack also releases headers from a private chain to lure validationless miners away from useful work, and uses that diversion to improve a double-spend attack~\cite{cao2023verifier}. PDoS differs in both objective and feedback: its endpoint is a self-funded liveness attack, and an infiltration-originated lead recruits the victim pool during the race. Other DoS studies target mempools, execution or validation load, state storage, network resources, and L2 sequencing~\cite{li2021deter,yaish2024speculative,he2024nurgle,tsuchiya2025amplification,li2025denialsequencing}. These attacks primarily exploit resource bottlenecks or transaction-processing paths.

\heading{MEV, Reward Volatility, and Incentive Attacks.}
Prior work shows that diminishing block subsidies and volatile transaction fees can destabilize PoW incentives~\cite{carlsten2016instability}. Studies on whale transactions, undercutting, and empirical mempool dynamics further show that reward volatility can lower the threshold for profitable deviations~\cite{barzur2023werlman,gong2022undercutting,sarenche2024deepselfish,sarenche2025volatile}. MEV research similarly demonstrates that high-value on-chain opportunities can feed back into consensus-layer stability~\cite{daian2020flashboys,qin2022darkforest}. Incentive attacks also appear in PoS settings, including longest-chain PoS barriers, Ouroboros-style security analyses, and Ethereum attestation-incentive attacks~\cite{browncohen2019barriers,kiayias2017ouroboros,kiayias2018praos,schwarzschilling2022threeattacks,zhang2024maxattestation}. Our work combines reward volatility with victim-funded infiltration, showing that high-reward environments can lower, rather than raise, the economic barrier for liveness attacks.

\heading{Network Propagation.}
Network-layer advantages can amplify consensus-layer attacks. Eclipse attacks and BGP hijacking show that adversaries can distort node views or routing paths~\cite{heilman2015eclipse,apostolaki2017hijacking}, while studies of high-throughput and bounded-capacity blockchains show that propagation and processing bottlenecks affect security--performance trade-offs~\cite{sompolinsky2015highrate,lewenberg2015inclusive,gervais2016security,bagaria2019prism,kiffer2024bounded}. PDoS does not rely on a specific network-layer exploit; instead, we abstract propagation asymmetry through the rushing parameter $\gamma$, which determines branch-winning probabilities in racing states.

\heading{Summary.}
Existing work studies revenue-maximizing deviations, pool infiltration, incentive-driven DoS, and reward volatility largely in isolation. PDoS connects these lines by showing that header-only deterrence and victim-pool subsidies can jointly transform liveness disruption from a costly attack into a victim-funded, self-sustaining, and potentially profitable strategy.

\section{Background and Preliminaries}

\heading{Denial-of-Service (DoS) Attacks.}
In traditional distributed systems, DoS attacks typically target network-layer bottlenecks such as bandwidth, connectivity, or computation. In PoW blockchains, however, a more representative class of DoS attacks operates directly on the incentive structure of the consensus layer. Rather than exhausting raw network resources, such attacks manipulate the expected utility of rational participants, causing them to conclude that continued mining is no longer economically worthwhile under certain states.

A representative example is \textit{BDoS}~\cite{mirkin2020bdos}. In BDoS, the attacker broadcasts only a block header while withholding the corresponding block body, thereby releasing a credible signal that a competing branch may already exist. This signal increases the perceived risk that the honest miners' current work will eventually end up on an orphaned branch. Once this expected risk becomes sufficiently large, some rational miners may choose to shut down to avoid further loss. Moreover, if rational miners attempt to continue mining based only on the attacker's published header, their actions become contingent on whether and when the attacker later reveals the block body.

\heading{Economic Limitation of BDoS.}
Although BDoS can, in principle, exert substantial pressure on chain liveness~\cite{mirkin2020bdos}, its practical sustainability is severely hindered by its unfavorable accounting profitability. To maintain a credible threat, the adversary must bear continuous explicit operational expenses, which can only be offset by winning substantive branch races to capture block rewards. However, following reward halving events, this sole revenue stream is substantially reduced. As operating margins become razor-thin, these diminished rewards fail to cover the ongoing mining costs, inevitably resulting in a negative accounting profit. Such an attritional attack becomes operationally unsustainable, as the adversary faces a rapid depletion of capital.

\heading{Withholding Attacks.}
Block withholding (BWH) attacks allow an adversary to infiltrate a victim mining pool, earn a share of the pool reward by submitting partial proofs of work (PPoWs), and discard any full proof of work (FPoW) that would otherwise produce a valid block~\cite{rosenfeld2011pooled,eyal2015miners}. This behavior dilutes the victim pool's realized revenue and transfers part of the reward to the attacker.

Fork-after-withholding (FAW) extends BWH by strategically releasing a withheld full block upon a favorable fork opportunity, thereby coupling pool infiltration with opportunistic branch competition~\cite{kwon2017faw}.

PAW can be viewed as a further refinement that combines FAW-style withholding with state-dependent hash-power reallocation~\cite{gao2019paw}. The attacker dynamically reallocates its hash power between private mining and infiltration mining according to the current chain state. When the attacker does not yet hold a private lead block, it primarily infiltrates the victim pool to extract share rewards. Once it finds a private block, it shifts more hash power to the private branch in order to enlarge its lead or improve its probability of winning the ensuing race, thereby optimizing its overall expected utility. The key feature of PAW is that it leverages victim resources to provide a continuous subsidy to the attacker. However, PAW is designed to maximize adversarial revenue rather than to directly disrupt chain liveness. In fact, a PAW attacker typically prefers the blockchain to continue operating, so that it can continue extracting excess rewards from the victim pool over time.

\heading{Simplified Payment Verification (SPV).}
A blockchain can be formalized as a sequence of blocks linked by hash pointers, denoted by $\mathcal{C}=\{B_0,B_1,\dots\}$, where $B_0$ is the genesis block~\cite{nakamoto2008bitcoin}. Each block $B$ is represented as a tuple $B=(h,b)$, where $h$ denotes the block header, containing metadata such as the previous-block hash, the Merkle root, and the PoW nonce, and $b$ denotes the block body, which contains the transaction list. Consensus is determined by the longest-chain rule.

Unlike many conventional models that implicitly assume $(h,b)$ is propagated atomically, we explicitly consider an \emph{asymmetric publication} capability: a participant may first broadcast the block header $h$ to prove possession of a valid PoW solution while temporarily withholding the block body $b$. We refer to this capability as \emph{SPV mining}. SPV mining emits a credible signal that a competing branch may exist, thereby altering the strategic expectations of rational miners~\cite{mirkin2020bdos}. PDoS is built precisely on this signaling capability.

Moreover, SPV mining, also referred to as validationless mining, is not uncommon in real-world systems. Coin Metrics~\cite{coinmetrics_empty_blocks_2024} reports that, upon receiving an initial signal that a new block has been found, mining pools often begin mining immediately without waiting for the full transaction template in order to gain a timing advantage; this practice can produce empty blocks containing no user transactions.

\heading{Blockchain Race.}
A \emph{race} refers to the competition among multiple conflicting candidate branches for inclusion in the main chain under strategic block release. In our model, a race is triggered when the attacker ends a withholding phase and releases its private or infiltration block, which happens to conflict with a newly generated honest block. Because the target miners may adopt different response strategies, the dynamics of different race scenarios are not identical.

It is important to distinguish \emph{race} from mere signaling. In this paper, a \emph{race} specifically refers to the case in which the attacker releases a \emph{full block} $(h,b)$ and engages in substantive consensus competition against an honest block. By contrast, if the attacker releases only the block header $h$ while withholding the body, we treat that state as a signaling-based interference mechanism rather than a genuine consensus race.

\section{Markov State Dynamics}

\subsection{Definitions and Rules}
In PoW blockchains, consensus strictly adheres to the longest-chain rule. Consequently, on-chain confirmation and reward settlement for the current round of the game are triggered only when the system transitions back to the initial state.

We define a block that breaks the equilibrium and transitions the system into a deterring or racing state as a \emph{race-triggering block} (or simply triggering block), whose reward remains pending. Conversely, a block that concludes the race and returns the system to State 0 is termed a \emph{race-resolving block}. Upon the generation of a race-resolving block, all pending rewards are settled simultaneously based on the ownership of the resulting longest chain. To formalize branch ownership during the race phase, we classify competing branches into three categories:
\begin{itemize}
    \item \textbf{Private Branch}: The branch containing $B_{\mathrm{pri}}$, which is discovered and released by the attacker's private hash power.
    \item \textbf{Infiltration Branch}: The branch containing $B_{\mathrm{inf}}$, which is discovered within the victim pool but selectively submitted by the attacker. At the consensus layer, this branch is recognized as the victim's output.
    \item \textbf{Honest Branch}: The branch containing any competing block other than the two aforementioned types (i.e., $B_{\mathrm{vic}}, B_{\mathrm{tar}}, \text{ or } B_{\mathrm{oth}}$).
\end{itemize}

Because fundamental mining events follow a memoryless Poisson process, the transition rate between any two system states is determined by the product of the base block generation rate $\lambda$ and the active hash power triggering the respective event.

\subsection{State Transition and Reward Settlement}\label{app:state_transitions}

Based on the aforementioned rules, we first detail the baseline full transition map when target miners mine normally. Subsequently, we elucidate the variations in state transitions when the target adopts alternative strategies. For each scenario, we list the state transitions, their transition rates, the corresponding mining events, and the resulting reward settlements.

\subsubsection{Scenario 1: Target Miners Mine Normally.}
\begin{itemize}
    \item \textbf{State 0 (Initial):}
    \begin{itemize}
        \item $0 \to 0$: $\lambda \eta$, $\mathcal{T}$ mines a race-resolving block $B_{\mathrm{new}}$. $\mathcal{T}$ receives the block reward for $B_{\mathrm{new}}$.
        \item $0 \to 0$: $\lambda \delta$, $\mathcal{O}$ mines a race-resolving block $B_{\mathrm{new}}$. $\mathcal{O}$ receives the block reward for $B_{\mathrm{new}}$.
        \item $0 \to 0$: $\lambda \beta$, $\mathcal{V}$ mines a race-resolving block $B_{\mathrm{new}}$. $\mathcal{V}$ receives the block reward for $B_{\mathrm{new}}$, and $\mathcal{A}$ receives a share reward from $B_{\mathrm{new}}$ based on $r_1$.
        \item $0 \to 1$: $\lambda (1-r_1)\alpha$, $\mathcal{A}$'s private hash power mines a race-triggering block $B_{\mathrm{pri}}$.
        \item $0 \to 2$: $\lambda r_1\alpha$, $\mathcal{A}$'s infiltrating hash power mines a race-triggering block $B_{\mathrm{inf}}$.
    \end{itemize}
    
    \item \textbf{State 1 (Private Lead):}
    \begin{itemize}
        \item $1 \to 3$: $\lambda \beta$, $\mathcal{V}$ mines a race-triggering block $B_{\mathrm{vic}}$.
        \item $1 \to 4$: $\lambda \eta$, $\mathcal{T}$ mines a race-triggering block $B_{\mathrm{tar}}$.
        \item $1 \to 4$: $\lambda \delta$, $\mathcal{O}$ mines a race-triggering block $B_{\mathrm{oth}}$.
    \end{itemize}
    
    \item \textbf{State 2 (Infiltration Lead):}
    \begin{itemize}
        \item $2 \to 0$: $\lambda \beta$, $\mathcal{V}$ mines a race-resolving block $B_{\mathrm{vic}}$, and $\mathcal{A}$ discards the stale $B_{\mathrm{inf}}$.\footnote{Upon mining $B_{\mathrm{vic}}$, the pool server updates the mining task. The attacker's leading infiltration block loses its validity, becomes a stale block, and must be discarded.} $\mathcal{V}$ receives the block reward for $B_{\mathrm{vic}}$, and $\mathcal{A}$ receives a share reward from $B_{\mathrm{vic}}$ weighted by $r_1$ and $r_2$.
        \item $2 \to 5$: $\lambda \eta$, $\mathcal{T}$ mines a race-triggering block $B_{\mathrm{tar}}$.
        \item $2 \to 5$: $\lambda \delta$, $\mathcal{O}$ mines a race-triggering block $B_{\mathrm{oth}}$.
    \end{itemize}
    
    \item \textbf{State 3 (Race PvV):} $B_{\mathrm{pri}}$ vs. $B_{\mathrm{vic}}$
    \begin{itemize}
        \item $3 \to 0$: $\lambda \alpha$, $\mathcal{A}$ mines a race-resolving block $B_{\mathrm{new}}$. $\mathcal{A}$ receives the block rewards for both $B_{\mathrm{pri}}$ and $B_{\mathrm{new}}$.
        \item $3 \to 0$: $\lambda \beta$, $\mathcal{V}$ mines a race-resolving block $B_{\mathrm{new}}$. $\mathcal{V}$ receives the block rewards for both $B_{\mathrm{vic}}$ and $B_{\mathrm{new}}$, and $\mathcal{A}$ receives a share reward from $B_{\mathrm{vic}}$ weighted by $r_1$ and $r_2$.\footnote{Because this race-resolving block $B_{\mathrm{new}}$ is generated during the racing state, the attacker has fully shifted its hash power back to the private branch (i.e., $r=0$) and thus does not participate in the share payout for this new block.}
        \item $3 \to 0$: $\lambda \gamma \eta$, $\mathcal{T}$ mines a race-resolving block $B_{\mathrm{new}}$ on the private branch. $\mathcal{A}$ receives the block reward for $B_{\mathrm{pri}}$, and $\mathcal{T}$ receives the block reward for $B_{\mathrm{new}}$.
        \item $3 \to 0$: $\lambda \gamma \delta$, $\mathcal{O}$ mines a race-resolving block $B_{\mathrm{new}}$ on the private branch. $\mathcal{A}$ receives the block reward for $B_{\mathrm{pri}}$, and $\mathcal{O}$ receives the block reward for $B_{\mathrm{new}}$.
        \item $3 \to 0$: $\lambda (1-\gamma) \eta$, $\mathcal{T}$ mines a race-resolving block $B_{\mathrm{new}}$ on the honest branch. $\mathcal{V}$ receives the block reward for $B_{\mathrm{vic}}$, $\mathcal{A}$ receives a share reward from $B_{\mathrm{vic}}$ weighted by $r_1$ and $r_2$, and $\mathcal{T}$ receives the block reward for $B_{\mathrm{new}}$.
        \item $3 \to 0$: $\lambda (1-\gamma) \delta$, $\mathcal{O}$ mines a race-resolving block $B_{\mathrm{new}}$ on the honest branch. $\mathcal{V}$ receives the block reward for $B_{\mathrm{vic}}$, $\mathcal{A}$ receives a share reward from $B_{\mathrm{vic}}$ weighted by $r_1$ and $r_2$, and $\mathcal{O}$ receives the block reward for $B_{\mathrm{new}}$.
    \end{itemize}
    
    \item \textbf{State 4 (Race PvE):} $B_{\mathrm{pri}}$ vs. $B_{\mathrm{ext}} \in \{B_{\mathrm{tar}}, B_{\mathrm{oth}}\}$
    \begin{itemize}
        \item $4 \to 0$: $\lambda \alpha$, $\mathcal{A}$ mines a race-resolving block $B_{\mathrm{new}}$. $\mathcal{A}$ receives the block rewards for both $B_{\mathrm{pri}}$ and $B_{\mathrm{new}}$.
        \item $4 \to 0$: $\lambda \gamma \beta$, $\mathcal{V}$ mines a race-resolving block $B_{\mathrm{new}}$ on the private branch. $\mathcal{A}$ receives the block reward for $B_{\mathrm{pri}}$, and $\mathcal{V}$ receives the block reward for $B_{\mathrm{new}}$.
        \item $4 \to 0$: $\lambda \gamma \eta$, $\mathcal{T}$ mines a race-resolving block $B_{\mathrm{new}}$ on the private branch. $\mathcal{A}$ receives the block reward for $B_{\mathrm{pri}}$, and $\mathcal{T}$ receives the block reward for $B_{\mathrm{new}}$.
        \item $4 \to 0$: $\lambda \gamma \delta$, $\mathcal{O}$ mines a race-resolving block $B_{\mathrm{new}}$ on the private branch. $\mathcal{A}$ receives the block reward for $B_{\mathrm{pri}}$, and $\mathcal{O}$ receives the block reward for $B_{\mathrm{new}}$.
        \item $4 \to 0$: $\lambda (1-\gamma) \beta$, $\mathcal{V}$ mines a race-resolving block $B_{\mathrm{new}}$ on the honest branch. The discoverer of the triggering block receives the block reward for $B_{\mathrm{ext}}$, and $\mathcal{V}$ receives the block reward for $B_{\mathrm{new}}$.
        \item $4 \to 0$: $\lambda (1-\gamma) \eta$, $\mathcal{T}$ mines a race-resolving block $B_{\mathrm{new}}$ on the honest branch. The discoverer of the triggering block receives the block reward for $B_{\mathrm{ext}}$, and $\mathcal{T}$ receives the block reward for $B_{\mathrm{new}}$.
        \item $4 \to 0$: $\lambda (1-\gamma) \delta$, $\mathcal{O}$ mines a race-resolving block $B_{\mathrm{new}}$ on the honest branch. The discoverer of the triggering block receives the block reward for $B_{\mathrm{ext}}$, and $\mathcal{O}$ receives the block reward for $B_{\mathrm{new}}$.
    \end{itemize}
    
    \item \textbf{State 5 (Race IvE):} $B_{\mathrm{inf}}$ vs. $B_{\mathrm{ext}} \in \{B_{\mathrm{tar}}, B_{\mathrm{oth}}\}$
    \begin{itemize}
        \item $5 \to 0$: $\lambda \alpha$, $\mathcal{A}$ mines a race-resolving block $B_{\mathrm{new}}$. $\mathcal{V}$ receives the block rewards for both $B_{\mathrm{inf}}$ and $B_{\mathrm{new}}$, while $\mathcal{A}$ receives a share reward from $B_{\mathrm{inf}}$ weighted by $r_1$ and $r_2$, as well as a share reward from $B_{\mathrm{new}}$ based on $r=1$.\footnote{Since $B_{\mathrm{new}}$ is generated in State 5, $\mathcal{A}$ participates in the payout with an infiltration ratio of $r=1$.}
        \item $5 \to 0$: $\lambda \beta$, $\mathcal{V}$ mines a race-resolving block $B_{\mathrm{new}}$. $\mathcal{V}$ receives the block rewards for both $B_{\mathrm{inf}}$ and $B_{\mathrm{new}}$, while $\mathcal{A}$ receives a share reward from $B_{\mathrm{inf}}$ weighted by $r_1$ and $r_2$, as well as a share reward from $B_{\mathrm{new}}$ based on $r=1$.
        \item $5 \to 0$: $\lambda \gamma \eta$, $\mathcal{T}$ mines a race-resolving block $B_{\mathrm{new}}$ on the infiltration branch. $\mathcal{V}$ receives the block reward for $B_{\mathrm{inf}}$, $\mathcal{A}$ receives a share reward from $B_{\mathrm{inf}}$ weighted by $r_1$ and $r_2$, and $\mathcal{T}$ receives the block reward for $B_{\mathrm{new}}$.
        \item $5 \to 0$: $\lambda \gamma \delta$, $\mathcal{O}$ mines a race-resolving block $B_{\mathrm{new}}$ on the infiltration branch. $\mathcal{V}$ receives the block reward for $B_{\mathrm{inf}}$, $\mathcal{A}$ receives a share reward from $B_{\mathrm{inf}}$ weighted by $r_1$ and $r_2$, and $\mathcal{O}$ receives the block reward for $B_{\mathrm{new}}$.
        \item $5 \to 0$: $\lambda (1-\gamma) \eta$, $\mathcal{T}$ mines a race-resolving block $B_{\mathrm{new}}$ on the honest branch. The discoverer of the triggering block receives the block reward for $B_{\mathrm{ext}}$, and $\mathcal{T}$ receives the block reward for $B_{\mathrm{new}}$.
        \item $5 \to 0$: $\lambda (1-\gamma) \delta$, $\mathcal{O}$ mines a race-resolving block $B_{\mathrm{new}}$ on the honest branch. The discoverer of the triggering block receives the block reward for $B_{\mathrm{ext}}$, and $\mathcal{O}$ receives the block reward for $B_{\mathrm{new}}$.
    \end{itemize}
\end{itemize}

\subsubsection{Scenario 2: Target Miners Fall into the SPV Trap.}
In this scenario, $\mathcal{T}$ falls into the trap, causing their output to become orphaned blocks. Consequently, there is no reward settlement, and their original transition paths to the racing states are redirected to a reset. Unlisted transition and settlement rules are identical to those in Scenario 1.
\begin{itemize}
    \item \textbf{State 1 Changes:}
    \begin{itemize}
        \item $1 \to 0$: $\lambda \eta$, $\mathcal{T}$ falls into the SPV trap.
        \item $1 \to 4$: $\lambda \delta$, rate attenuation only.
    \end{itemize}
    
    \item \textbf{State 2 Changes:}
    \begin{itemize}
        \item $2 \to 0$: $\lambda \eta$, $\mathcal{T}$ falls into the SPV trap.
        \item $2 \to 5$: $\lambda \delta$, rate attenuation only.
    \end{itemize}
\end{itemize}

\subsubsection{Scenario 3: Target Miners Stop to Cut Losses.}
In this scenario, the hash power of $\mathcal{T}$ is entirely withdrawn, directly reducing the corresponding state transition rates. Unlisted transition and settlement rules are identical to those in Scenario 1.
\begin{itemize}
    \item \textbf{State 1 Changes:} 
    \begin{itemize}
        \item $1 \to 4$: $\lambda \delta$, rate attenuation only.
    \end{itemize}
    
    \item \textbf{State 2 Changes:}
    \begin{itemize}
        \item $2 \to 5$: $\lambda \delta$, rate attenuation only.
    \end{itemize}
\end{itemize}

\subsection{Derivation of Race Winning Probabilities}\label{app:race_win_rate}

Regardless of the strategy $\mathcal{T}$ adopts in the deterring state, the block-header deterrence signal is broken once a race is triggered. Furthermore, because $\mathcal{T}$ is modeled as a continuum, when a micro-miner within the population or an external node discovers a race-triggering block, the remaining miners in the population are still subject to underlying P2P network propagation delays during the subsequent branch competition. Consequently, in the macro-level race, the entire hash power population $\eta$ behaves as neutral hash power influenced by the rushing ability $\gamma$. Thus, the branch winning probabilities in different racing states are global constants strictly independent of the target's response strategy $S$.

The hash power allocation and winning probabilities for each racing state are derived as follows:
\begin{itemize}
    \item \textbf{PvV Winning Probability:} $\mathcal{A}$ supports the private block $B_{\mathrm{pri}}$; $\mathcal{V}$ supports its own $B_{\mathrm{vic}}$; $\mathcal{E}$ supports $B_{\mathrm{pri}}$ in proportion to $\gamma$. Thus, the private branch winning probability is $p_{\mathcal{A}}^3 = \alpha + \gamma(\eta + \delta)$, and the corresponding honest branch winning probability is $p_{\mathcal{V}}^3 = \beta + (1-\gamma)(\eta + \delta)$.
    
    \item \textbf{PvE Winning Probability:} $\mathcal{A}$ supports the private block $B_{\mathrm{pri}}$; $\mathcal{H}$ supports $B_{\mathrm{pri}}$ in proportion to $\gamma$. Thus, the private branch winning probability is $p_{\mathcal{A}}^4 = \alpha + \gamma(\beta + \eta + \delta)$, and the corresponding honest branch winning probability is $p_{\mathcal{E}}^4 = (1-\gamma)(\beta + \eta + \delta)$.
    
    \item \textbf{IvE Winning Probability:} $\mathcal{A}$ supports the infiltration block $B_{\mathrm{inf}}$; $\mathcal{V}$ is forced to support $B_{\mathrm{inf}}$; $\mathcal{E}$ supports $B_{\mathrm{inf}}$ in proportion to $\gamma$. Thus, the infiltration branch winning probability is $p_{\mathcal{A}}^5 = \alpha + \beta + \gamma(\eta + \delta)$, and the corresponding honest branch winning probability is $p_{\mathcal{E}}^5 = (1-\gamma)(\eta + \delta)$.
\end{itemize}

\section{Derivation of Steady-State Distributions}
\label{app:steady_state}

This section details the methodology and algebraic derivations for solving the system's steady-state distributions under various strategies. For any given strategy $S$, the steady-state distribution vector $\boldsymbol{\pi}^S$ must satisfy the global balance equations of the CTMC:
\begin{equation*}
    \boldsymbol{\pi}^S \mathbf{Q}^S = \mathbf{0}, \quad \text{s.t.} \quad \sum_{i=0}^{5} \pi_i^S = 1,
\end{equation*}
where $\mathbf{Q}^S$ denotes the generator matrix of the state transitions.

Based on the exhaustive state transition events and their respective rates detailed in Appendix~\ref{app:state_transitions}, we construct the generator matrix $\mathbf{Q}^S$ for each target response strategy. Because the rows of this matrix sum to zero (i.e., $Q_{ii} = -\sum_{j \neq i} Q_{ij}$), the system of global balance equations inherently contains one linearly dependent equation. To simplify the derivation across the following scenarios, we uniformly omit the relatively complex inflow equation for the initial state (State 0) and utilize the equations for the remaining non-initial states. By applying substitution, we express the stationary probabilities of all other states as functions of the initial state probability $\pi_0^S$. Finally, we substitute these expressions into the normalization condition $\sum_{i=0}^{5} \pi_i^S = 1$ to obtain the unique closed-form solution.

\subsection{Scenario 1: Target Chooses to Mine}
Extracting the base block generation rate $\lambda$ as a common factor, the generator matrix is given by:
\begin{equation*}
\mathbf{Q}^{\mathrm{Mine}} = \lambda
\begin{bmatrix} 
-\alpha & (1-r_1)\alpha & r_1\alpha & 0 & 0 & 0 \\ 
0 & -(1-\alpha) & 0 & \beta & \eta+\delta & 0 \\ 
\beta & 0 & -(1-\alpha) & 0 & 0 & \eta+\delta \\ 
1 & 0 & 0 & -1 & 0 & 0 \\ 
1 & 0 & 0 & 0 & -1 & 0 \\ 
1 & 0 & 0 & 0 & 0 & -1 
\end{bmatrix}
\end{equation*}
The balance equations for the non-initial states are as follows:
\begin{equation*}
\begin{cases} 
\lambda (1-\alpha) \cdot \pi_1^{\mathrm{Mine}} = \lambda (1-r_1)\alpha \cdot \pi_0^{\mathrm{Mine}} \\ 
\lambda (1-\alpha) \cdot \pi_2^{\mathrm{Mine}} = \lambda r_1\alpha \cdot \pi_0^{\mathrm{Mine}} \\ 
\lambda \cdot \pi_3^{\mathrm{Mine}} = \lambda \beta \cdot \pi_1^{\mathrm{Mine}} \\ 
\lambda \cdot \pi_4^{\mathrm{Mine}} = \lambda (\eta + \delta) \cdot \pi_1^{\mathrm{Mine}} \\ 
\lambda \cdot \pi_5^{\mathrm{Mine}} = \lambda (\eta + \delta) \cdot \pi_2^{\mathrm{Mine}} 
\end{cases}
\end{equation*}
By substitution, the steady-state probabilities of all other states can be expressed as functions of the initial state probability $\pi_0^{\mathrm{Mine}}$:
\begin{align*} 
\pi_1^{\mathrm{Mine}} &= \frac{(1-r_1)\alpha}{1-\alpha} \pi_0^{\mathrm{Mine}}, \quad
\pi_2^{\mathrm{Mine}} = \frac{r_1\alpha}{1-\alpha} \pi_0^{\mathrm{Mine}} \\ 
\pi_3^{\mathrm{Mine}} &= \beta \pi_1^{\mathrm{Mine}} = \frac{\beta(1-r_1)\alpha}{1-\alpha} \pi_0^{\mathrm{Mine}} \\ 
\pi_4^{\mathrm{Mine}} &= (\eta+\delta) \pi_1^{\mathrm{Mine}} = \frac{(\eta+\delta)(1-r_1)\alpha}{1-\alpha} \pi_0^{\mathrm{Mine}} \\ 
\pi_5^{\mathrm{Mine}} &= (\eta+\delta) \pi_2^{\mathrm{Mine}} = \frac{(\eta+\delta)r_1\alpha}{1-\alpha} \pi_0^{\mathrm{Mine}} 
\end{align*}
Substituting these into the normalization condition yields:
\begin{align*}
    \pi_0^{\mathrm{Mine}} \Bigg[ 1 &+ \frac{(1-r_1)\alpha + r_1\alpha}{1-\alpha} \\
    &+ \frac{\beta(1-r_1)\alpha + (\eta+\delta)(1-r_1)\alpha + (\eta+\delta)r_1\alpha}{1-\alpha} \Bigg] = 1
\end{align*}
Noting that $(1-r_1)\alpha + r_1\alpha = \alpha$ and applying the hash power conservation relation $\alpha+\beta+\eta+\delta=1$, the summation term above can be simplified to:
\begin{equation*}
    \frac{1 + \alpha ( 1-\alpha-r_1\beta )}{1-\alpha}
\end{equation*}
We define the normalized partition function for this scenario as:
\begin{equation*}
    D_{\mathrm{Mine}} = 1 + \alpha ( 1-\alpha-r_1\beta),
\end{equation*}
This allows us to derive the complete steady-state distribution vector:
\begin{align} \label{eq:pi_mine}
    \boldsymbol{\pi}^{\mathrm{Mine}} = \frac{1}{D_{\mathrm{Mine}}} \Big( &1-\alpha,\ (1-r_1)\alpha,\ r_1\alpha,\ \beta(1-r_1)\alpha, \notag \\
    &(1-\alpha-\beta)(1-r_1)\alpha,\ (1-\alpha-\beta)r_1\alpha \Big).
\end{align}

\subsection{Scenario 2: Target Chooses SPV Mining}
After extracting the base block generation rate $\lambda$ as a common factor, the generator matrix is given by:
\begin{equation*}
\mathbf{Q}^{\mathrm{SPV}} = \lambda
\begin{bmatrix} 
-\alpha & (1-r_1)\alpha & r_1\alpha & 0 & 0 & 0 \\ 
\eta & -(1-\alpha) & 0 & \beta & \delta & 0 \\ 
\beta+\eta & 0 & -(1-\alpha) & 0 & 0 & \delta \\ 
1 & 0 & 0 & -1 & 0 & 0 \\ 
1 & 0 & 0 & 0 & -1 & 0 \\ 
1 & 0 & 0 & 0 & 0 & -1 
\end{bmatrix}
\end{equation*}
The balance equations for the non-initial states are as follows:
\begin{equation*}
\begin{cases} 
\lambda (1-\alpha) \cdot \pi_1^{\mathrm{SPV}} = \lambda (1-r_1)\alpha \cdot \pi_0^{\mathrm{SPV}} \\ 
\lambda (1-\alpha) \cdot \pi_2^{\mathrm{SPV}} = \lambda r_1\alpha \cdot \pi_0^{\mathrm{SPV}} \\ 
\lambda \cdot \pi_3^{\mathrm{SPV}} = \lambda \beta \cdot \pi_1^{\mathrm{SPV}} \\ 
\lambda \cdot \pi_4^{\mathrm{SPV}} = \lambda \delta \cdot \pi_1^{\mathrm{SPV}} \\ 
\lambda \cdot \pi_5^{\mathrm{SPV}} = \lambda \delta \cdot \pi_2^{\mathrm{SPV}} 
\end{cases}
\end{equation*}
Solving by substitution yields:
\begin{align*} 
\pi_1^{\mathrm{SPV}} &= \frac{(1-r_1)\alpha}{1-\alpha} \pi_0^{\mathrm{SPV}}, \quad
\pi_2^{\mathrm{SPV}} = \frac{r_1\alpha}{1-\alpha} \pi_0^{\mathrm{SPV}} \\ 
\pi_3^{\mathrm{SPV}} &= \beta \pi_1^{\mathrm{SPV}} = \frac{\beta(1-r_1)\alpha}{1-\alpha} \pi_0^{\mathrm{SPV}} \\ 
\pi_4^{\mathrm{SPV}} &= \delta \pi_1^{\mathrm{SPV}} = \frac{\delta(1-r_1)\alpha}{1-\alpha} \pi_0^{\mathrm{SPV}} \\ 
\pi_5^{\mathrm{SPV}} &= \delta \pi_2^{\mathrm{SPV}} = \frac{\delta r_1\alpha}{1-\alpha} \pi_0^{\mathrm{SPV}} 
\end{align*}
Substituting these into the normalization condition yields:
\begin{align*}
    \pi_0^{\mathrm{SPV}} \Bigg[ 1 &+ \frac{(1-r_1)\alpha + r_1\alpha}{1-\alpha} \\
    &+ \frac{\beta(1-r_1)\alpha + \delta(1-r_1)\alpha + \delta r_1\alpha}{1-\alpha} \Bigg] = 1
\end{align*}
Similarly, the summation term can be simplified to:
\begin{equation*}
    \frac{1 + \alpha ( 1-\alpha-\eta-r_1\beta )}{1-\alpha}
\end{equation*}
We define the normalized partition function for this scenario as:
\begin{equation*}
    D_{\mathrm{SPV}} = 1 + \alpha ( 1-\alpha-\eta-r_1\beta),
\end{equation*}
Consequently, we derive the complete steady-state distribution vector:
\begin{align*}
    \boldsymbol{\pi}^{\mathrm{SPV}} = \frac{1}{D_{\mathrm{SPV}}} \Big( &1-\alpha,\ (1-r_1)\alpha,\ r_1\alpha, \notag \\
    &\beta(1-r_1)\alpha,\ \delta(1-r_1)\alpha,\ \delta r_1\alpha \Big).
\end{align*}

\subsection{Scenario 3: Target Chooses to Stop}
Extracting the common factor $\lambda$, the generator matrix is:
\begin{equation*}
\mathbf{Q}^{\mathrm{Stop}} = \lambda
\begin{bmatrix} 
-\alpha & (1-r_1)\alpha & r_1\alpha & 0 & 0 & 0 \\ 
0 & -(\beta+\delta) & 0 & \beta & \delta & 0 \\ 
\beta & 0 & -(\beta+\delta) & 0 & 0 & \delta \\ 
1 & 0 & 0 & -1 & 0 & 0 \\ 
1 & 0 & 0 & 0 & -1 & 0 \\ 
1 & 0 & 0 & 0 & 0 & -1 
\end{bmatrix}
\end{equation*}
The balance equations for the non-initial states are as follows:
\begin{equation*}
\begin{cases} 
\lambda (\beta+\delta) \cdot \pi_1^{\mathrm{Stop}} = \lambda (1-r_1)\alpha \cdot \pi_0^{\mathrm{Stop}} \\ 
\lambda (\beta+\delta) \cdot \pi_2^{\mathrm{Stop}} = \lambda r_1\alpha \cdot \pi_0^{\mathrm{Stop}} \\ 
\lambda \cdot \pi_3^{\mathrm{Stop}} = \lambda \beta \cdot \pi_1^{\mathrm{Stop}} \\ 
\lambda \cdot \pi_4^{\mathrm{Stop}} = \lambda \delta \cdot \pi_1^{\mathrm{Stop}} \\ 
\lambda \cdot \pi_5^{\mathrm{Stop}} = \lambda \delta \cdot \pi_2^{\mathrm{Stop}} 
\end{cases}
\end{equation*}
Solving by substitution yields:
\begin{align*} 
\pi_1^{\mathrm{Stop}} &= \frac{(1-r_1)\alpha}{\beta+\delta} \pi_0^{\mathrm{Stop}}, \quad
\pi_2^{\mathrm{Stop}} = \frac{r_1\alpha}{\beta+\delta} \pi_0^{\mathrm{Stop}}\\ 
\pi_3^{\mathrm{Stop}} &= \beta \pi_1^{\mathrm{Stop}} = \frac{\beta(1-r_1)\alpha}{\beta+\delta} \pi_0^{\mathrm{Stop}} \\ 
\pi_4^{\mathrm{Stop}} &= \delta \pi_1^{\mathrm{Stop}} = \frac{\delta(1-r_1)\alpha}{\beta+\delta} \pi_0^{\mathrm{Stop}} \\ 
\pi_5^{\mathrm{Stop}} &= \delta \pi_2^{\mathrm{Stop}} = \frac{\delta r_1\alpha}{\beta+\delta} \pi_0^{\mathrm{Stop}} 
\end{align*}
Substituting these into the normalization condition yields:
\begin{align*}
    \pi_0^{\mathrm{Stop}} \Bigg[ 1 &+ \frac{(1-r_1)\alpha + r_1\alpha}{\beta+\delta} \\
    &+ \frac{\beta(1-r_1)\alpha + \delta(1-r_1)\alpha + \delta r_1\alpha}{\beta+\delta} \Bigg] = 1
\end{align*}
Similarly, the summation term can be simplified to:
\begin{equation*}
    \frac{1 - \eta + \alpha ( 1-\alpha-\eta-r_1\beta )}{\beta+\delta}
\end{equation*}
We define the normalized partition function for this scenario as:
\begin{equation*}
D_{\mathrm{Stop}} = 1 - \eta + \alpha ( 1-\alpha-\eta-r_1\beta),
\end{equation*}
Consequently, we derive the complete steady-state distribution vector:
\begin{align*}
    \boldsymbol{\pi}^{\mathrm{Stop}} = \frac{1}{D_{\mathrm{Stop}}} \Big( &1-\alpha-\eta,\ (1-r_1)\alpha,\ r_1\alpha, \notag \\
    &\beta(1-r_1)\alpha,\ \delta(1-r_1)\alpha,\ \delta r_1\alpha \Big).
\end{align*}

\subsection{Scenario 4: Partial Termination of Target Miners}
\label{app:scenario4_derivation}

When the active target hash power proportion is $x \in [0, 1]$, the transition rates exiting the deterring states (States 1 and 2) are affected by the hash-power vacuum. We define the total active competing hash power in the deterring states as $\sigma(x) = 1-\alpha-(1-x)\eta$, and the block generation rate of the honest branch as $\nu(x) = x\eta+\delta$. Extracting $\lambda$ as a common factor, the generator matrix is given by:
\begin{equation*}
\mathbf{Q}(x) = \lambda
\begin{bmatrix} 
-\alpha & (1-r_1)\alpha & r_1\alpha & 0 & 0 & 0 \\ 
0 & -\sigma(x) & 0 & \beta & \nu(x) & 0 \\ 
\beta & 0 & -\sigma(x) & 0 & 0 & \nu(x) \\ 
1 & 0 & 0 & -1 & 0 & 0 \\ 
1 & 0 & 0 & 0 & -1 & 0 \\ 
1 & 0 & 0 & 0 & 0 & -1 
\end{bmatrix}
\end{equation*}

The balance equations for the non-initial states are as follows:
\begin{equation*}
\begin{cases} 
\lambda \big(1-\alpha-(1-x)\eta\big) \cdot \pi_1(x) = \lambda (1-r_1)\alpha \cdot \pi_0(x) \\ 
\lambda \big(1-\alpha-(1-x)\eta\big) \cdot \pi_2(x) = \lambda r_1\alpha \cdot \pi_0(x) \\ 
\lambda \cdot \pi_3(x) = \lambda \beta \cdot \pi_1(x) \\ 
\lambda \cdot \pi_4(x) = \lambda (x\eta + \delta) \cdot \pi_1(x) \\ 
\lambda \cdot \pi_5(x) = \lambda (x\eta + \delta) \cdot \pi_2(x) 
\end{cases}
\end{equation*}

Solving by substitution yields:
\begin{align*} 
\pi_1(x) &= \frac{(1-r_1)\alpha}{1-\alpha-(1-x)\eta} \pi_0(x)\\
\pi_2(x) &= \frac{r_1\alpha}{1-\alpha-(1-x)\eta} \pi_0(x) \\ 
\pi_3(x) &= \beta \pi_1(x) = \frac{\beta(1-r_1)\alpha}{1-\alpha-(1-x)\eta} \pi_0(x) \\ 
\pi_4(x) &= (x\eta+\delta) \pi_1(x) = \frac{(x\eta+\delta)(1-r_1)\alpha}{1-\alpha-(1-x)\eta} \pi_0(x) \\ 
\pi_5(x) &= (x\eta+\delta) \pi_2(x) = \frac{(x\eta+\delta)r_1\alpha}{1-\alpha-(1-x)\eta} \pi_0(x) 
\end{align*}

Substituting these into the normalization condition:
\begin{align*} 
    \pi_0(x) \Bigg[ 1 &+ \frac{(1-r_1)\alpha + r_1\alpha}{1-\alpha-(1-x)\eta} \\ 
    &+ \frac{\beta(1-r_1)\alpha + (x\eta+\delta)(1-r_1)\alpha + (x\eta+\delta)r_1\alpha}{1-\alpha-(1-x)\eta} \Bigg] = 1 
\end{align*}

Similarly, the summation term can be simplified to:
\begin{equation*}
    \frac{1-(1-x)\eta + \alpha \big[ 1 - \alpha - (1-x)\eta - r_1\beta \big]}{1-\alpha-(1-x)\eta}
\end{equation*}

We define the normalized partition function for this scenario as:
\begin{equation} \label{eq:d_x}
    D(x) = 1 - (1-x)\eta + \alpha \left[ 1 - \alpha - (1-x)\eta - r_1\beta \right].
\end{equation}
Consequently, we derive the complete generalized steady-state distribution vector:
\begin{align} \label{eq:pi_x}
    \boldsymbol{\pi}(x) = \frac{1}{D(x)} \Big(&1-\alpha-(1-x)\eta,\quad (1-r_1)\alpha,\quad r_1\alpha, \notag \\
    &\beta(1-r_1)\alpha,\quad (x\eta+\delta)(1-r_1)\alpha,\quad (x\eta+\delta)r_1\alpha \Big).
\end{align}

\section{Derivation of Expected Utilities via Markov Reward Process}
\label{app:utility_derivation}

This section employs the first principles of the standard MRP. We derive the core dynamic inflation parameter, then demonstrate the rigorous extraction of the utility coefficients presented in the main text.

\subsection{State-Dependent MEV Inflation Coefficient}
\label{app:mev_inflation}

In realistic networks, when the system enters a deterring state, the withdrawal or invalidation of a portion of the hash power prolongs the actual block interval. This, in turn, causes the time-varying MEV accumulated within a single block to inflate. Here, we derive the analytical solution for this inflation coefficient $e^S$.

\heading{Absolute Time Difference and Extra MEV Accumulation.}
In the baseline state without hash power shutdown, the expected block generation time is $\mathbb{E}[\tau_b] = \frac{1}{\lambda}$. 
In the deterring state, let $\Delta \alpha_\mathrm{act}^S$ denote the hash power vacuum, defined as the proportion of hash power that withdraws from public effective competition. The expected effective block generation time of the main chain is then prolonged to $\mathbb{E}[\tau_{\mathrm{det}}^S] = \frac{1}{\lambda(1-\Delta \alpha_\mathrm{act}^S)}$. 
Compared to the baseline state, the additional waiting time required to produce an effective block in the deterring state is:
\begin{equation*}
    \Delta \tau^S = \mathbb{E}[\tau_{\mathrm{det}}^S] - \mathbb{E}[\tau_b] = \frac{1}{\lambda} \left( \frac{\Delta \alpha_\mathrm{act}^S}{1-\Delta \alpha_\mathrm{act}^S} \right).
\end{equation*}

During this additional time $\Delta \tau^S$, the transaction fees and whale rewards accumulating at a constant rate constitute the system's excess dividend, given by $(\lambda_t + \lambda_w \mathbb{E}[\mathcal{R}_w]) \cdot \Delta \tau^S$.

\heading{Closed-Form Solution for the Inflation Coefficient.}
The inflation coefficient is essentially the ratio of the block's total realized value to its baseline value:
\begin{align*}
    e^S &\triangleq \frac{\mathcal{R}_b + (\lambda_t + \lambda_w \mathbb{E}[\mathcal{R}_w]) \cdot \Delta \tau^S}{\mathcal{R}_b} \\
    &= 1 + \frac{(\lambda_t + \lambda_w \mathbb{E}[\mathcal{R}_w]) \cdot \Delta \tau^S}{\mathcal{R}_b}\\
    &=1 + \left( \frac{\lambda_t + \lambda_w \mathbb{E}[\mathcal{R}_w]}{\lambda \mathcal{R}_b} \right) \left( \frac{\Delta \alpha_\mathrm{act}^S}{1-\Delta \alpha_\mathrm{act}^S} \right).
\end{align*}
Defining $M \triangleq \frac{\lambda_t + \lambda_w \mathbb{E}[\mathcal{R}_w]}{\lambda \mathcal{R}_b} \in [0, 1)$ as the proportion of MEV rewards relative to the total reward, we obtain:
\begin{equation*}
    e^S = 1 + M \left( \frac{\Delta \alpha_\mathrm{act}^S}{1-\Delta \alpha_\mathrm{act}^S} \right)\ge 1.
\end{equation*}

When $S=\mathrm{Mine}$, only the attacker's hash power is effectively offline, yielding $\Delta \alpha_\mathrm{act}^{\mathrm{Mine}} = \alpha$. When $S=\mathrm{SPV}$, the target's hash power remains active and triggers state transitions (i.e., the SPV trap), which yields $\Delta \alpha_\mathrm{act}^{\mathrm{SPV}} = \alpha$. When $S=\mathrm{Stop}$, the target's hash power deterministically shuts down, yielding $\Delta \alpha_\mathrm{act}^{\mathrm{Stop}} = \alpha + \eta$. Clearly, the inequality $e^{\mathrm{Stop}} > e^{\mathrm{SPV}} = e^{\mathrm{Mine}}$ strictly holds.

\subsection{Formal Extraction of Utility Coefficients}
\label{app:formal_extraction}

In an MRP, the system's total expected reward stream consists of continuous rate rewards within states and impulse rewards during transitions. In PoW blockchains, transaction fees and MEV accumulated over time must rely on the one-time on-chain confirmation of blocks for settlement. Therefore, by introducing the dynamic inflation coefficient $e^S$, our model equivalently transforms the integral of continuous rewards during deterring states into amplified impulse rewards upon race-resolving state transitions.

Let $\Pi_{\mathcal{M}}^S = \mathcal{R}_{\mathcal{M}}^S - \mathcal{C}_{\mathcal{M}}^S$ denote the expected profit of any participant $\mathcal{M}$ (with hash power proportion $\alpha_{\mathcal{M}}$) under strategy $S$. The corresponding general formula for the normalized utility (per unit time, per unit hash power) is defined as:
\begin{equation*}
    U_{\mathcal{M}}^S = \frac{\Pi_{\mathcal{M}}^S}{\alpha_{\mathcal{M}}} = \frac{\mathcal{R}_{\mathcal{M}}^S - \mathcal{C}_{\mathcal{M}}^S}{\alpha_{\mathcal{M}}}.
\end{equation*}

The absolute expected total revenue $\mathcal{R}_{\mathcal{M}}^S$ in the numerator is formalized under the MRP framework as the sum of impulse rewards upon state transitions:
\begin{equation} \label{eq:mrp_base}
    \mathcal{R}_{\mathcal{M}}^S = \sum_{i \in \mathcal{S}} \pi_i^S \sum_{j \in \mathcal{S}} q_{ij}^S R_{ij}^{\mathcal{M}}
\end{equation}
where $q_{ij}^S$ is the state transition rate, and $R_{ij}^{\mathcal{M}}$ is the absolute economic reward allocated to $\mathcal{M}$ when the transition occurs (refer to Appendix \ref{app:state_transitions} for specific values).

Conversely, the absolute expected total cost $\mathcal{C}_{\mathcal{M}}^S$ is the continuous rate cost incurred by keeping mining equipment running. Incorporating the normalized operational cost $c$, it can be expressed as:
\begin{equation} \label{eq:cost_base}
    \mathcal{C}_{\mathcal{M}}^S = \alpha_{\mathcal{M}} \cdot c \cdot \theta_{\mathcal{M}}^S
\end{equation}
where $\theta_{\mathcal{M}}^S \in [0, 1]$ is defined as the active time ratio during which miner $\mathcal{M}$ keeps its equipment running under strategy $S$, corresponding to the integral of the steady-state probabilities of its active states.

\subsubsection{Derivation of Utility Coefficients for Target Miner $\mathcal{T}$.}
When $S=\mathrm{Mine}$, $\mathcal{T}$ remains active throughout. It deterministically receives the reward for the race-resolving block it mines during transitions from States 0, 3, 4, and 5 back to State 0. Concurrently, in States 4 and 5, if the honest branch wins and the race-triggering block was discovered by $\mathcal{T}$, it receives the pending triggering block reward amplified by $e^{\mathrm{Mine}}$. Substituting this into Equation \eqref{eq:mrp_base}, its absolute expected revenue can be explicitly decomposed into race-resolving block rewards and race-triggering block rewards:
\begin{align*}
    \mathcal{R}_{\mathcal{T}}^{\mathrm{Mine}} =& \underbrace{\Big( \pi_0^{\mathrm{Mine}} + \pi_3^{\mathrm{Mine}} + \pi_4^{\mathrm{Mine}} + \pi_5^{\mathrm{Mine}} \Big) (\lambda \eta) \mathcal{R}_b}_{\text{resolvers in States }0,3,4,5} \\
    &+ \underbrace{\pi_4^{\mathrm{Mine}} \Big[ \lambda(1-\gamma)(\beta+\eta+\delta) \frac{\eta}{\eta+\delta} e^{\mathrm{Mine}} \Big] \mathcal{R}_b}_{\text{State 4 trigger}} \\
    &+ \underbrace{\pi_5^{\mathrm{Mine}} \Big[ \lambda(1-\gamma)(\eta+\delta) \frac{\eta}{\eta+\delta} e^{\mathrm{Mine}} \Big] \mathcal{R}_b}_{\text{State 5 trigger}}
\end{align*}

Based on the derivation of race winning probabilities in Appendix \ref{app:race_win_rate}, the rate factors for the honest branch winning in States 4 and 5 can be directly substituted with $p_{\mathcal{E}}^4$ and $p_{\mathcal{E}}^5$, respectively.
Meanwhile, according to the algebraic relationships among the components in the steady-state distribution formula \eqref{eq:pi_mine}, the sojourn probabilities of the deterring states and the subsequent racing states strictly satisfy the following flow conservation: $\pi_4^{\mathrm{Mine}} = (\eta+\delta)\pi_1^{\mathrm{Mine}}$, and $\pi_5^{\mathrm{Mine}} = (\eta+\delta)\pi_2^{\mathrm{Mine}}$.

By substituting these winning rate replacements and steady-state component relationships into the original absolute expected revenue equation, and dividing both sides by the target's baseline output $\eta \lambda \mathcal{R}_b$, we define the normalized effective block generation rate. The final form aligns with Equation \eqref{eq:v_t_mine} in the main text:
\begin{equation*}
    v_{\mathcal{T}}^{\mathrm{Mine}} \triangleq \frac{\mathcal{R}_{\mathcal{T}}^{\mathrm{Mine}}}{\eta \lambda \mathcal{R}_b} = \sum_{i \in \{0,3,4,5\}} \pi_i^{\mathrm{Mine}} + e^{\mathrm{Mine}} \Big( \pi_1^{\mathrm{Mine}} p_{\mathcal{E}}^4 + \pi_2^{\mathrm{Mine}} p_{\mathcal{E}}^5 \Big).
\end{equation*}

When $S=\mathrm{SPV}$, $\mathcal{T}$'s output during the deterring phase inevitably triggers the trap and becomes invalid; when $S=\mathrm{Stop}$, $\mathcal{T}$ does not participate in any game during the deterring phase, producing no triggering blocks and earning no additional rewards. Mathematically, all terms related to triggering block rewards in the above equation become zero. Therefore, its effective block generation rates degenerate to $v_{\mathcal{T}}^{\mathrm{SPV}} = \sum_{i \in \{0,3,4,5\}} \pi_i^{\mathrm{SPV}}$ and $v_{\mathcal{T}}^{\mathrm{Stop}} = \theta_{\mathcal{T}}^{\mathrm{Stop}}$.

\subsubsection{Derivation of Utility Coefficients for Attacker $\mathcal{A}$.}
The total expected revenue of $\mathcal{A}$ can be decoupled into two components: the direct block revenue from winning private branches, $\mathcal{R}_{\mathrm{blk}}^S$, and the share revenue stolen via infiltration, $\mathcal{R}_{\mathrm{shr}}^S$.

\textbf{Effective Block Generation Rate $v_{\mathcal{A}}^S$:} 
This coefficient is derived from $\mathcal{R}_{\mathrm{blk}}^S$. $\mathcal{A}$ deterministically receives the reward for its own race-resolving block only during transitions from States 3 and 4 back to State 0. Additionally, in States 3 and 4, if the private branch wins, $\mathcal{A}$ receives the pending triggering block $B_{\mathrm{pri}}$ mined in State 0. Its absolute expected revenue can similarly be decomposed as:
\begin{align*}
    \mathcal{R}_{\mathrm{blk}}^S &= \underbrace{\Big( \pi_3^S + \pi_4^S \Big) (\lambda \alpha) \mathcal{R}_b}_{\text{resolvers in States }3,4} \\
    &+ \underbrace{\pi_3^S \Big[ \lambda \big( \alpha + \gamma(\eta+\delta) \big) \Big] \mathcal{R}_b}_{\text{State 3 trigger}} \\
    &+ \underbrace{\pi_4^S \Big[ \lambda \big( \alpha + \gamma(\beta+\eta+\delta) \big) \Big] \mathcal{R}_b}_{\text{State 4 trigger}}
\end{align*}
Based on the derivation of race winning probabilities in Appendix \ref{app:race_win_rate}, the rate factors for the private branch winning in States 3 and 4 can be directly substituted with $\lambda p_{\mathcal{A}}^3$ and $\lambda p_{\mathcal{A}}^4$, respectively. By dividing both sides of the equation by the attacker's baseline output $\alpha \lambda \mathcal{R}_b$, we define the normalized effective block generation rate, yielding a final form consistent with the main text:
\begin{equation*}
    v_{\mathcal{A}}^S \triangleq \frac{\mathcal{R}_{\mathrm{blk}}^S}{\alpha \lambda \mathcal{R}_b} = \pi_3^S + \pi_4^S + \frac{1}{\alpha} \big( \pi_3^S p_{\mathcal{A}}^3 + \pi_4^S p_{\mathcal{A}}^4 \big).
\end{equation*}

\textbf{Share Extraction Rate $s_{\mathcal{A}}^S$:}
This is the most non-linear parameter in the model. When the attacker injects an infiltrating hash power of $r\alpha$ into the victim pool, its share proportion within the pool is defined as $F(r) = \frac{r\alpha}{\beta + r\alpha}$.

To mitigate the errors introduced by Jensen's Inequality due to the concavity of $F(r)$, we employ a local smoothing method based on the block discovery paths to compute the local time-weighted infiltration ratio within specific settlement windows:
\begin{equation*}
    \bar{r}_1^S = \frac{r_1\pi_0^S + r_2\pi_1^S}{\pi_0^S + \pi_1^S}, \quad \bar{r}_2^S = \frac{r_1\pi_0^S + r_2\pi_2^S}{\pi_0^S + \pi_2^S}
\end{equation*}
This orthogonal partitioning significantly compresses the variance of $r$ within the local windows, effectively reducing the otherwise prominent non-linear error into negligible higher-order terms. This achieves high-precision approximation while preserving the closed-form nature of the formulas.

$\mathcal{A}$'s parasitic share revenue stems from four core settlement scenarios: (1) The victim pool's race-resolving block output in State 0; (2) Spanning States 0 and 1, the settlement of the victim pool's pending triggering block $B_{\mathrm{vic}}$ when the honest branch wins in State 3; (3) Spanning States 0 and 2, the shares of the victim's race-resolving block when forced to discard the infiltration block in State 2, or the settlement of the pending triggering block $B_{\mathrm{inf}}$ (mined in State 0) when the infiltration branch wins in State 5; (4) The pool's race-resolving blocks mined during the full-scale infiltration competition in State 5. Expanding these according to their raw state transition rates yields:
\begin{align*}
    \mathcal{R}_{\mathrm{shr}}^S &= \underbrace{\pi_0^S (\lambda \beta) F(r_1) \mathcal{R}_b}_{\text{State 0 payout}} \\
    &+ \underbrace{\pi_3^S \Big[ \lambda \big( \beta + (1-\gamma)(\eta+\delta) \big) \Big] e^S F(\bar{r}_1^S) \mathcal{R}_b}_{\text{State 3 settlement}} \\
    &+ \underbrace{\Big\{ \pi_2^S (\lambda \beta) e^S + \pi_5^S \lambda \big( \alpha+\beta+\gamma(\eta+\delta) \big) \Big\} F(\bar{r}_2^S) \mathcal{R}_b}_{\text{States 2/5 settlement}} \\
    &+ \underbrace{\pi_5^S \lambda(\alpha+\beta) F(1) \mathcal{R}_b}_{\text{State 5 resolver}}
\end{align*}
Based on the derivation of race winning probabilities in Appendix \ref{app:race_win_rate}, the rate factor for the honest branch winning in State 3 above can be directly substituted with $\lambda p_{\mathcal{V}}^3$, and the rate factor for the infiltration branch winning in State 5 can be substituted with $\lambda p_{\mathcal{A}}^5$.
Meanwhile, because the attacker commits its full hash power to infiltration in State 5 ($r=1$), its share proportion becomes $F(1) = \frac{\alpha}{\alpha+\beta}$. Substituting this into the fourth line yields:
\begin{equation*}
    \pi_5^S \lambda (\alpha+\beta) \left( \frac{\alpha}{\alpha+\beta} \right) \mathcal{R}_b = \pi_5^S (\lambda \alpha) \mathcal{R}_b
\end{equation*}

Finally, by substituting the aforementioned winning rate replacements and dividing both sides of the equation by the attacker's baseline output $\alpha \lambda \mathcal{R}_b$, we extract the normalized share extraction function $f(r) \triangleq \frac{F(r)}{\alpha} = \frac{r}{\beta + r\alpha}$. We thus define the normalized share extraction rate, whose final form aligns with the main text:
\begin{align*}
    s_{\mathcal{A}}^S \triangleq \frac{\mathcal{R}_{\mathrm{shr}}^S}{\alpha \lambda \mathcal{R}_b} =& \pi_0^S \beta f(r_1) \\
    &+ e^S \Big[ \pi_3^S p_{\mathcal{V}}^3 f(\bar{r}_1^S) + \pi_2^S \beta f(\bar{r}_2^S) \Big] + \pi_5^S \Big[ p_{\mathcal{A}}^5 f(\bar{r}_2^S) + 1 \Big].
\end{align*}

\section{Optimization of Infiltration Ratios}\label{app:optimization}
\subsection{Optimization of the Global Infiltration Ratio}
The initial infiltration ratio $r_1$ determines the probability of branching into the deterring state and the baseline theft revenue. Because the attacker's primary objective is to disrupt system liveness, and maximizing its own utility is secondary, solving for $r_1$ constitutes a priority-based conditional optimization problem. 

We define the set of initial infiltration ratios that satisfy
the economic feasibility of deterrence as

\[
F=\{r_1\in[0,1]|\Delta U_T^{Stop}(r_1,r_2^*)<0\}.
\]

When \(F\neq\emptyset\), the attacker maximizes its utility
subject to inducing target shutdown. Otherwise, when
\(F=\emptyset\), meaning that deterrence is economically
infeasible, the attacker falls back to unconstrained revenue
maximization:

\[
r_1^*=
\begin{cases}
\arg\max_{r_1\in F}U_A^{Stop}(r_1),
&F\neq\emptyset,\\
\arg\max_{r_1\in[0,1]}U_A^{Mine}(r_1),
&F=\emptyset.
\end{cases}
\]

Because this non-linear optimization lacks a simple closed-form solution, we employ a Grid Search method to solve for $r_1^*$ in our experimental evaluation.

\subsection{Optimization of the Local Infiltration Ratio} 

Based on Equation \eqref{eq:u_t_diff} and the derivations in Appendix \ref{app:steady_state}, both the system's steady-state distribution $\boldsymbol{\pi}^S$ and the target miners' shutdown utility difference $\Delta U_{\mathcal{T}}^\mathrm{Stop}$ are mathematically independent of $r_2$. This is because the infiltrating hash power only submits PPoWs and deliberately discards valid FPoWs, contributing nothing to global block discovery. Thus, the attacker's entire hash power $\alpha$ remains effectively offline, making the single-block reward inflation effect $e^S$ a constant strictly independent of $r_2$. Under the proportional payout mechanism of the victim pool, the local profit function can be expressed as:
\begin{equation*}
\Pi_{\mathrm{inf}}(r_2) = \underbrace{\lambda\beta e^S\mathcal{R}_b}_{\text{Inflated pool total output}} \underbrace{\left( \frac{r_2\alpha}{\beta+r_2\alpha} \right)}_{\text{Share proportion}} - \underbrace{c r_2\alpha}_{\text{Infiltration power cost}}
\end{equation*}
Substituting the baseline profitability equation $\lambda\mathcal{R}_b=\omega_b c$ and factoring out common terms, the profit function simplifies to:
\begin{equation*}
\Pi_{\mathrm{inf}}(r_2) = c\alpha\left[ \frac{\omega_b e^S\beta r_2}{\beta+r_2\alpha} - r_2 \right]
\end{equation*}

To find the extremum, we take the first derivative with respect to $r_2$:
\begin{equation*}
\frac{d\Pi_{\mathrm{inf}}}{dr_2} = c\alpha\left[ \frac{\omega_b e^S\beta^2}{(\beta+r_2\alpha)^2} - 1 \right]
\end{equation*}

Setting the derivative to zero to solve for the steady-state extremum:
\begin{equation*}
(\beta+r_2\alpha)^2 = \omega_b e^S\beta^2
\end{equation*}

Since the system's hash power parameters $\alpha, \beta$ and the economic factors $\omega_b, e^S$ are all strictly positive, we take the positive square root of both sides and simplify:
\begin{equation*}
r_2\alpha = \beta\left(\sqrt{\omega_b e^S}-1\right) \implies r_2 = \frac{\beta\left(\sqrt{\omega_b e^S}-1\right)}{\alpha}.
\end{equation*}

Finally, considering that the attacker's infiltration ratio must satisfy the realistic boundary constraint $r_2 \in [0, 1]$, we apply truncation to obtain the closed-form solution for the optimal infiltration ratio:
\begin{equation*}
r_2^* = \max\left( 0, \min\left( 1, \frac{\beta\left(\sqrt{\omega_b e^S}-1\right)}{\alpha} \right) \right).
\end{equation*}

\section{Proof of SPV as a Strictly Dominated Strategy}\label{app:spv_dominated}

Let the target's SPV utility difference be defined as the normalized utility difference between adopting $S=\mathrm{Mine}$ and $S=\mathrm{SPV}$. Substituting the respective utility equations and expanding yields:
\begin{align*}
\Delta U_{\mathcal{T}}^\mathrm{SPV} &\triangleq U_{\mathcal{T}}^{\mathrm{Mine}} - U_{\mathcal{T}}^{\mathrm{SPV}} \\
&=c \left( \omega_b v_{\mathcal{T}}^{\mathrm{Mine}} - 1 \right)-c \left( \omega_b v_{\mathcal{T}}^{\mathrm{SPV}} - 1 \right)\\
&=c\omega_b\left(v_{\mathcal{T}}^{\mathrm{Mine}} - v_{\mathcal{T}}^{\mathrm{SPV}}\right).
\end{align*}
Since both the normalized operational cost and the baseline profitability factor are strictly positive ($c>0, \omega_b>0$), the sign of this utility difference is entirely determined by the difference in their effective block generation rates.

Based on the coefficient derivations in Section \ref{sec:expected_utility} and the probability normalization condition $\sum_{i=0}^{5} \pi_i^S = 1$, subtracting the two effective block generation rates yields:
\begin{align*}
    v_{\mathcal{T}}^{\mathrm{Mine}} - v_{\mathcal{T}}^{\mathrm{SPV}}
    = &\big( \pi_1^{\mathrm{SPV}} + \pi_2^{\mathrm{SPV}} \big) - \big( \pi_1^{\mathrm{Mine}} + \pi_2^{\mathrm{Mine}} \big) \\
    &+ e^{\mathrm{Mine}} \Big( \pi_1^{\mathrm{Mine}} p_{\mathcal{E}}^4 + \pi_2^{\mathrm{Mine}} p_{\mathcal{E}}^5 \Big).
\end{align*}

According to Lemma \ref{lem:prob_order}, the difference between the first two parenthetical terms is strictly positive. Furthermore, since the inflation coefficient $e^{\mathrm{Mine}} > 1$, and both the steady-state probabilities $\pi_i^S$ and the winning probabilities $p_{\mathcal{E}}^i$ are strictly positive, the trailing term is also strictly positive. Therefore, $v_{\mathcal{T}}^{\mathrm{Mine}} - v_{\mathcal{T}}^{\mathrm{SPV}} > 0$ strictly holds, which implies:
\begin{equation*}
    \Delta U_{\mathcal{T}}^\mathrm{SPV} > 0 \iff U_{\mathcal{T}}^{\mathrm{Mine}} > U_{\mathcal{T}}^{\mathrm{SPV}}.
\end{equation*}
This proves that the utility of target miners choosing SPV mining is always strictly lower than the utility of normal mining. Consequently, $S=\mathrm{SPV}$ constitutes a strictly dominated strategy.

\section{Proofs of Theorems}\label{app:theorems}

\subsection{Proof of Theorem \ref{thm:threshold_dominance}}\label{app:threshold-proof}
\begin{proof}
According to Definition \ref{def:alpha_crit}, the sole game-theoretic criterion for the target miners $\mathcal{T}$ to decide whether to shut down depends on the utility difference between normal mining and stopping to cut losses. Assuming other network and economic parameters are constant, we denote this utility difference as a function of the attacker's hash power $\alpha$ and the infiltration ratio combination $(r_1, r_2)$: $\Delta U_{\mathcal{T}}^\mathrm{Stop}(\alpha, r_1, r_2)$. The target shuts down if and only if $\Delta U_{\mathcal{T}}^\mathrm{Stop}(\alpha, r_1, r_2) < 0$.

\heading{BDoS as a Special Case of PDoS:} A BDoS attack corresponds to a degenerate case where the attacker completely abandons infiltration. In this scenario, all infiltration-related states in the CTMC collapse, and the target utility difference induced by BDoS is identically equal to the value of PDoS evaluated at the zero-infiltration boundary point, i.e., $\Delta U_{\mathcal{T}}^\mathrm{Stop} \big|_{\mathrm{BDoS}}(\alpha) \equiv \Delta U_{\mathcal{T}}^\mathrm{Stop}(\alpha, 0, 0)$.

\heading{Theoretical Extremum Inequality:} PDoS grants the attacker the degrees of freedom to dynamically optimize within the continuous two-dimensional space $(r_1, r_2) \in [0,1]^2$. To calculate the critical deterrence threshold, selecting the parameter combination that minimizes the target's utility difference yields the actual utility difference of PDoS: $\Delta U_{\mathcal{T}}^\mathrm{Stop} \big|_{\mathrm{PDoS}}(\alpha) \triangleq \min_{r_1, r_2} \Delta U_{\mathcal{T}}^\mathrm{Stop}(\alpha, r_1, r_2)$. By the fundamental properties of extrema in real analysis, the global minimum of a function over a closed region must be less than or equal to its value evaluated at any specific boundary point. Therefore, the following extremum inequality strictly holds:
\begin{align} \label{eq:delta-inequality}
  \Delta U_{\mathcal{T}}^\mathrm{Stop} \big|_{\mathrm{PDoS}}(\alpha) &\triangleq \min_{r_1, r_2} \Delta U_{\mathcal{T}}^\mathrm{Stop}(\alpha, r_1, r_2) \notag \\
  &\le \Delta U_{\mathcal{T}}^\mathrm{Stop}(\alpha, 0, 0) \equiv \Delta U_{\mathcal{T}}^\mathrm{Stop} \big|_{\mathrm{BDoS}}(\alpha).
\end{align}

\heading{Inclusion Relation of Stopping Sets:} We define the hash power sets (Stopping Sets) sufficient to induce the target to shut down under the two attacks:
\begin{align*}
 \mathcal{A}_{\mathrm{BDoS}} &= \left\{ \alpha \mid \Delta U_{\mathcal{T}}^\mathrm{Stop} \big|_{\mathrm{BDoS}}(\alpha) < 0 \right\}, \\
 \mathcal{A}_{\mathrm{PDoS}} &= \left\{ \alpha \mid \Delta U_{\mathcal{T}}^\mathrm{Stop} \big|_{\mathrm{PDoS}}(\alpha) < 0 \right\}.
\end{align*}
For any given hash power $\alpha \in \mathcal{A}_{\mathrm{BDoS}}$, by definition, we have $\Delta U_{\mathcal{T}}^\mathrm{Stop}(\alpha, 0, 0) < 0$. Combining this with the extremum inequality \eqref{eq:delta-inequality}, we deduce:
\begin{align*}
    0 &> \Delta U_{\mathcal{T}}^\mathrm{Stop}(\alpha, 0, 0)\\ 
    &\ge \min_{r_1, r_2} \Delta U_{\mathcal{T}}^\mathrm{Stop}(\alpha, r_1, r_2) 
    = \Delta U_{\mathcal{T}}^\mathrm{Stop} \big|_{\mathrm{PDoS}}(\alpha).
\end{align*}
This implies that this $\alpha$ must also satisfy the shutdown criterion for PDoS, meaning $\alpha \in \mathcal{A}_{\mathrm{PDoS}}$. From this set-theoretic derivation, we obtain $\mathcal{A}_{\mathrm{BDoS}} \subseteq \mathcal{A}_{\mathrm{PDoS}}$. The critical hash power $\alpha_{\mathcal{A}}^*$ is defined as the infimum of the corresponding stopping set. According to the fundamental properties of set theory, the infimum of a superset must be less than or equal to the infimum of its subset. Since $\mathcal{A}_{\mathrm{BDoS}} \subseteq \mathcal{A}_{\mathrm{PDoS}}$, it immediately follows that:
\begin{equation*}
    \alpha_{\mathcal{A}}^* \big|_{\mathrm{PDoS}} = \inf \mathcal{A}_{\mathrm{PDoS}} \;\le\; \inf \mathcal{A}_{\mathrm{BDoS}} = \alpha_{\mathcal{A}}^* \big|_{\mathrm{BDoS}}.
\end{equation*}
\end{proof}

\subsection{Proof of Theorem \ref{thm:net_cost_dominance}}\label{app:net-cost-proof}
\begin{proof}
According to Definition \ref{def:min_cost}, $\mathcal{C}_{\mathcal{A}}^*$ represents the theoretical minimum net cost achievable by the attacker.

A BDoS attack corresponds to a boundary case where the attacker completely abandons infiltration. Thus, its optimal net cost rate is identically equal to the value at that specific strategy point:
\begin{equation*}
    \mathcal{C}_{\mathcal{A}}^* \big|_{\mathrm{BDoS}}(\alpha) \equiv \mathcal{C}_{\mathcal{A}}^S(\alpha, 0, 0).
\end{equation*}

For a PDoS attack, the decision space is a two-dimensional continuous closed region $(r_1, r_2) \in [0,1]^2$. Since the BDoS strategy point $(0,0)$ lies within this global strategy space, according to the fundamental property of the minimum in real analysis, the global minimum of a continuous function over a closed region must be less than or equal to its value at any specific point within that region. Therefore, the optimal net cost rate of PDoS must satisfy:
\begin{align*}
    \mathcal{C}_{\mathcal{A}}^* \big|_{\mathrm{PDoS}}(\alpha) &\triangleq \min_{(r_1, r_2) \in [0,1]^2} \mathcal{C}_{\mathcal{A}}^S(\alpha, r_1, r_2) \\
    &\le \mathcal{C}_{\mathcal{A}}^S(\alpha, 0, 0) = \mathcal{C}_{\mathcal{A}}^* \big|_{\mathrm{BDoS}}(\alpha).
\end{align*}
\end{proof}

\subsection{Proof of Theorem \ref{thm:self-sustain}}\label{app:self-sustain-proof}
\begin{proof}
An attacker achieves self-sustainability if its minimum net cost rate satisfies $\mathcal{C}_{\mathcal{A}}^* \le 0$. Based on Definition \ref{def:min_cost} and Equation \eqref{eq:net_cost}, we expand this inequality explicitly as:
\begin{equation*}
    \min_{r_1, r_2} \left\{ c \Big[ \theta_{\mathcal{A}}^S(r_1, r_2) - \omega_b \big( v_{\mathcal{A}}^S(r_1, r_2) + s_{\mathcal{A}}^S(r_1, r_2) \big) \Big] \right\} \;\le\; 0.
\end{equation*}
Since the normalized operational cost $c > 0$, it can be divided out. The extremum condition is equivalent to the existence of a pair $(r_1, r_2) \in [0,1]^2$ such that:
\begin{equation*}
     \theta_{\mathcal{A}}^S(r_1, r_2) - \omega_b \big( v_{\mathcal{A}}^S(r_1, r_2) + s_{\mathcal{A}}^S(r_1, r_2) \big) \;\le\; 0.
\end{equation*}
Given that the sum of the effective block generation rate and the share extraction rate is strictly positive, $v_{\mathcal{A}}^S(r_1, r_2) + s_{\mathcal{A}}^S(r_1, r_2) > 0$, we rearrange the terms and divide by the positive sum to obtain:
\begin{equation*}
    \omega_b \;\ge\; \frac{\theta_{\mathcal{A}}^S(r_1, r_2)}{v_{\mathcal{A}}^S(r_1, r_2) + s_{\mathcal{A}}^S(r_1, r_2)}.
\end{equation*}
The necessary and sufficient condition for this existence to hold is that $\omega_b$ must be greater than or equal to the minimum value of this fractional term within the domain, i.e.,:
\begin{equation*}
    \omega_b \;\ge\; \min_{r_1, r_2} \frac{\theta_{\mathcal{A}}^S(r_1, r_2)}{v_{\mathcal{A}}^S(r_1, r_2) + s_{\mathcal{A}}^S(r_1, r_2)}.
\end{equation*}
The right side of the inequality is a system parameter extremum independent of $\omega_b$. We define it as the lower bound of the profitability factor for PDoS:
\begin{equation*}
    \Omega_{\mathcal{A}}^* \big|_{\mathrm{PDoS}} \;\triangleq\; \min_{r_1, r_2} \frac{\theta_{\mathcal{A}}^S(r_1, r_2)}{v_{\mathcal{A}}^S(r_1, r_2) + s_{\mathcal{A}}^S(r_1, r_2)}.
\end{equation*}

To prove the absolute advantage of PDoS over BDoS, we introduce the strategy space constraint from optimization theory. A BDoS attack lacks an infiltration dimension, and its strategy space is strictly constrained to the zero-infiltration boundary. Thus, the critical profitability factor required for it to self-sustain degenerates into a fixed value under this constraint:
\begin{equation*}
    \Omega_{\mathcal{A}}^* \big|_{\mathrm{BDoS}} \;\equiv\; \frac{\theta_{\mathcal{A}}^S(0, 0)}{v_{\mathcal{A}}^S(0, 0) + s_{\mathcal{A}}^S(0, 0)}.
\end{equation*}

A PDoS attack is essentially a feasible region relaxation of the BDoS strategy space, as it allows the variable combination $(r_1, r_2)$ to optimize freely within the two-dimensional continuous region $[0,1]^2$. According to the fundamental principles of optimization theory, the global optimal solution of a less constrained minimization problem must be less than or equal to the solution of any sub-problem with additional rigid constraints. Therefore, it follows that:
\begin{equation*}
   \Omega_{\mathcal{A}}^* \big|_{\mathrm{PDoS}} \;\le\; \Omega_{\mathcal{A}}^* \big|_{\mathrm{BDoS}}.
\end{equation*}
\end{proof}

\subsection{Proof of Theorem \ref{thm:deterrence_resilience}}\label{app:deterrence-proof}
\begin{proof}
When $\omega_b$ becomes sufficiently large, the utility of target miners persisting in mining increases, eventually leading to the failure of attack deterrence.
The attack is effective if and only if the target's shutdown utility difference satisfies $\Delta U_{\mathcal{T}}^\mathrm{Stop}<0$. Based on Equation \eqref{eq:u_t_diff}, this condition expands to:
\begin{equation*}
    c \Big[ \omega_b \big( v_{\mathcal{T}}^{\mathrm{Mine}} - v_{\mathcal{T}}^{\mathrm{Stop}} \big) - \big( 1 - v_{\mathcal{T}}^{\mathrm{Stop}} \big) \Big] < 0.
\end{equation*}
Since the normalized operational cost $c > 0$, it can be divided out. We now prove that $v_{\mathcal{T}}^{\mathrm{Mine}} - v_{\mathcal{T}}^{\mathrm{Stop}} > 0$.
According to the derivation in Section \ref{sec:expected_utility} and the probability normalization condition $\sum_{i=0}^{5} \pi_i^S = 1$, subtracting the two yields:
\begin{align*}
    v_{\mathcal{T}}^{\mathrm{Mine}} - v_{\mathcal{T}}^{\mathrm{Stop}}
    = &\big( \pi_1^{\mathrm{Stop}} + \pi_2^{\mathrm{Stop}} \big) - \big( \pi_1^{\mathrm{Mine}} + \pi_2^{\mathrm{Mine}} \big) \\
    &+ e^{\mathrm{Mine}} \Big( \pi_1^{\mathrm{Mine}} p_{\mathcal{E}}^4 + \pi_2^{\mathrm{Mine}} p_{\mathcal{E}}^5 \Big)
\end{align*}
According to Lemma \ref{lem:prob_order}, the difference between the first two terms is strictly positive. Since the inflation coefficient $e^{\mathrm{Mine}} > 1$, and both the probabilities $\pi_i^S$ and the winning rates $p_{\mathcal{E}}^i$ are positive, the trailing term is also strictly positive. Therefore, $v_{\mathcal{T}}^{\mathrm{Mine}} - v_{\mathcal{T}}^{\mathrm{Stop}} > 0$. Thus, the condition is equivalent to the existence of $(r_1, r_2) \in [0,1]^2$ such that:
\begin{equation*}
    \omega_b \;<\; \frac{1 - v_{\mathcal{T}}^{\mathrm{Stop}}(r_1, r_2)}{v_{\mathcal{T}}^{\mathrm{Mine}}(r_1, r_2) - v_{\mathcal{T}}^{\mathrm{Stop}}(r_1, r_2)}.
\end{equation*}
The necessary and sufficient condition for this existence to hold is that $\omega_b$ must be less than or equal to the maximum value of this fractional term within the domain, i.e.,:
\begin{equation*}
    \omega_b \;<\; \max_{r_1, r_2} \frac{1 - v_{\mathcal{T}}^{\mathrm{Stop}}(r_1, r_2)}{v_{\mathcal{T}}^{\mathrm{Mine}}(r_1, r_2) - v_{\mathcal{T}}^{\mathrm{Stop}}(r_1, r_2)}.
\end{equation*}
The right side of this inequality is an extremum of system parameters independent of $\omega_b$, which we define as the upper bound of the profitability factor for PDoS:
\begin{equation*}
    \Omega_{\mathcal{T}}^* \big|_{\mathrm{PDoS}} \;\triangleq\; \max_{r_1, r_2} \frac{1 - v_{\mathcal{T}}^{\mathrm{Stop}}(r_1, r_2)}{v_{\mathcal{T}}^{\mathrm{Mine}}(r_1, r_2) - v_{\mathcal{T}}^{\mathrm{Stop}}(r_1, r_2)}.
\end{equation*}

Similarly, analyzing from the perspective of strategy space constraints: The strategy space of BDoS is constrained at $(r_1, r_2) = (0,0)$. Thus, the critical profitability factor required to make a BDoS attack effective degenerates to:
\begin{equation*}
    \Omega_{\mathcal{T}}^* \big|_{\mathrm{BDoS}} \;\equiv\; \frac{1 - v_{\mathcal{T}}^{\mathrm{Stop}}(0, 0)}{v_{\mathcal{T}}^{\mathrm{Mine}}(0, 0) - v_{\mathcal{T}}^{\mathrm{Stop}}(0, 0)}.
\end{equation*}
Since PDoS is a complete relaxation of the BDoS feasible region, according to fundamental principles of optimization theory, the global optimal solution of an unconstrained maximization problem must be greater than or equal to the solution of any sub-problem with additional rigid constraints. Therefore, the extremum inequality strictly holds:
\begin{equation*}
    \Omega_{\mathcal{T}}^* \big|_{\mathrm{PDoS}} \;\ge\; \Omega_{\mathcal{T}}^* \big|_{\mathrm{BDoS}}.
\end{equation*}
\end{proof}

\subsection{Proof of Lemma \ref{lemma:mono-utility} }\label{app:mono-utility-proof}
\begin{proof}
In the partial shutdown scenario, let the probability of the system being in a non-deterring state be $\pi_{\mathrm{act}}(x)$, and the probability of being in a deterring state be $\pi_{\mathrm{det}}(x)$. According to the normalization condition:
\begin{equation}\label{eq:pi_det}
    1 - \pi_{\mathrm{act}}(x) = \pi_{\mathrm{det}}(x) = \pi_1(x) + \pi_2(x).
\end{equation}

When individual $\mathcal{T}_i$ chooses $S_{\mathcal{T}_i} = \mathrm{Stop}$, they only remain active and generate blocks during the non-deterring state. Their effective block generation rate and active time ratio are both $v_{\mathcal{T}_i}^{\mathrm{Stop}}(x) = \theta_{\mathcal{T}_i}^{\mathrm{Stop}}(x) = \pi_{\mathrm{act}}(x)$. Substituting this into the normalized utility \eqref{eq:u_m}, their shutdown utility is:
\begin{equation*}
    U_{\mathcal{T}_i}^{\mathrm{Stop}}(x) = c (\omega_b-1)\pi_{\mathrm{act}}(x).
\end{equation*}

When individual $\mathcal{T}_i$ chooses $S_{\mathcal{T}_i} = \mathrm{Mine}$, they remain active throughout the process, i.e., $\theta_{\mathcal{T}_i}^{\mathrm{Mine}} = 1$. Introducing the dynamic MEV inflation coefficient $e(x)$, their effective block generation rate includes normal output and triggering block output during the deterring phase:
\begin{equation*}
    v_{\mathcal{T}_i}^{\mathrm{Mine}}(x) = \pi_{\mathrm{act}}(x) + e(x) \Big[ \pi_1(x) p_{\mathcal{T}_i}^4(x) + \pi_2(x) p_{\mathcal{T}_i}^5(x) \Big]
\end{equation*}
When the triggering blocks mined during the deterrence period participate in the race, the hash power of their discoverers supports their own branch, while the remaining hash power is affected by $\gamma$. Thus, the microscopic winning probabilities in the race are respectively:
\begin{align*}
    p_{\mathcal{T}_i}^4(x) &= \eta_i + (1-\gamma) \big( \beta + x\eta + \delta - \eta_i \big), \notag \\
    p_{\mathcal{T}_i}^5(x) &= \eta_i + (1-\gamma) \big( x\eta + \delta - \eta_i \big).
\end{align*}

Based on the component relations of the steady-state probability \eqref{eq:pi_x}, the conditional probability ratios strictly satisfy $\frac{\pi_1(x)}{\pi_{\mathrm{det}}(x)} = 1-r_1$ and $\frac{\pi_2(x)}{\pi_{\mathrm{det}}(x)} = r_1$. The weighted average winning rate in the deterring state is defined as:
\begin{align}\label{eq:p_bar}
    \bar{p}_{\mathcal{T}_i}(x) &\triangleq \frac{\pi_1(x) p_{\mathcal{T}_i}^4(x) + \pi_2(x) p_{\mathcal{T}_i}^5(x)}{\pi_{\mathrm{det}}(x)} \notag \\
    &= (1-r_1)p_{\mathcal{T}_i}^4(x) + r_1 p_{\mathcal{T}_i}^5(x) \notag \\
    &= \eta_i + (1-\gamma) \Big[ (1-r_1)\beta + x\eta + \delta - \eta_i \Big].
\end{align}

Substituting into the normalized utility \eqref{eq:u_m}, their mining utility is:
\begin{equation*}
    U_{\mathcal{T}_i}^{\mathrm{Mine}}(x) = c \left( \omega_b \Big[ \pi_{\mathrm{act}}(x) + e(x) \pi_{\mathrm{det}}(x) \bar{p}_{\mathcal{T}_i}(x) \Big] - 1 \right).
\end{equation*}

Subtracting the two yields the individual shutdown utility difference:
\begin{align}\label{u_ti_stop}
    \Delta U_{\mathcal{T}_i}^{\mathrm{Stop}}(x) &\triangleq U_{\mathcal{T}_i}^{\mathrm{Mine}}(x) - U_{\mathcal{T}_i}^{\mathrm{Stop}}(x) \notag \\
    &= c \cdot \pi_{\mathrm{det}}(x) \Big[ \omega_b \cdot e(x) \cdot \bar{p}_{\mathcal{T}_i}(x) - 1 \Big]
\end{align}

To analyze monotonicity, define the core revenue term as $H(x) = e(x) \bar{p}_{\mathcal{T}_i}(x)$. Since the hash power vacuum is $\Delta \alpha_\mathrm{act}(x)=\alpha + \eta - x\eta$, let constant $C = 1-\alpha-\eta = \beta + \delta$, the inflation coefficient can be simplified to:
\begin{equation}\label{eq:e_x}
    e(x) = 1 + M \left( \frac{1 - (C + x\eta)}{C + x\eta} \right) = (1 - M) + \frac{M}{C + x\eta}
\end{equation}
Express the microscopic winning rate as a linear function $\bar{p}_{\mathcal{T}_i}(x) = A + Bx$, where:
\begin{align*}
    A &= \eta_i + (1-\gamma)[(1-r_1)\beta + \delta - \eta_i] \\
    B &= (1-\gamma)\eta > 0
\end{align*}
Taking the derivative with respect to $x$ yields:
\begin{align*}
    H'(x) &= \frac{d}{dx} \left[ (1 - M)(A + Bx) + M \frac{A + Bx}{C + x\eta} \right] \\
    &= (1 - M)B + M \frac{BC - A\eta}{(C + x\eta)^2}
\end{align*}
Expanding and simplifying the numerator term:
\begin{align*}
    BC - A\eta &= (1-\gamma)\eta \Big[ (\beta+\delta) - \big( (1-r_1)\beta + \delta - \eta_i \big) \Big] - \eta_i\eta \\
    &= \eta \Big[ (1-\gamma)r_1\beta - \gamma\eta_i \Big]
\end{align*}
Under the non-atomic game assumption, the hash power of a microscopic miner $\eta_i \to 0$, and the limit of this term is:
\begin{equation*}
    \lim_{\eta_i \to 0} (BC - A\eta) = \eta(1-\gamma)r_1\beta > 0
\end{equation*}
Since $0<M<1, B > 0$ and $(C+x\eta)^2 > 0$, we deduce that $H'(x) > 0$. Thus, the core revenue term $H(x)$ is strictly monotonically increasing with respect to $x$.

According to Definition \ref{def:alpha_crit}, when $\alpha > \alpha_{\mathcal{A}}^*$, the collective shutdown utility difference of the target is negative. Based on the symmetry of the non-atomic game ($\eta_i \to 0$), the relative utility of an individual in the fully active state ($x=1$) must also necessarily be negative, i.e., $\Delta U_{\mathcal{T}_i}^{\mathrm{Stop}}(1)=c\cdot\pi_{\mathrm{det}}(1)[\omega_b H(1) - 1] < 0$. Since $c > 0$ and $\pi_{\mathrm{det}}(1) > 0$, it must hold that $\omega_b H(1) - 1 < 0$. Because $H(x)$ is strictly monotonically increasing, for $\forall x \in [0, 1]$ it universally satisfies:
\begin{equation*}
    \omega_b H(x) - 1 \le \omega_b H(1) - 1 < 0
\end{equation*}

For the deterring state probability $\pi_{\mathrm{det}}(x) = \frac{\alpha}{D(x)}$, combined with the derivative of the partition function \eqref{eq:d_x}, $\frac{d D(x)}{dx} = \eta(1+\alpha) > 0$, we have:
\begin{equation*}
    \frac{d \pi_{\mathrm{det}}(x)}{dx} = - \frac{\alpha}{D(x)^2} \frac{d D(x)}{dx} < 0.
\end{equation*}

Applying the product rule to calculate the derivative of the total utility difference:
\begin{align*}
    \frac{d \Delta U_{\mathcal{T}_i}^{\mathrm{Stop}}(x)}{dx} =& \underbrace{c}_{(+)} \Bigg( \underbrace{\frac{d \pi_{\mathrm{det}}(x)}{dx}}_{(-)} \underbrace{\Big[ \omega_b H(x) - 1 \Big]}_{(-)} \\
    &+ \underbrace{\pi_{\mathrm{det}}(x)}_{(+)} \omega_b \underbrace{H'(x)}_{(+)} \Bigg) > 0.
\end{align*}
Therefore, the utility difference is strictly monotonically increasing over $x \in [0, 1]$.
\end{proof}

\subsection{Proof of Theorem \ref{thm:unique-ne} }\label{app:unique-ne-proof}
\begin{proof}
According to Definition \ref{def:alpha_crit}, when $\alpha > \alpha_{\mathcal{A}}^*$, under the most optimistic state where target miners are fully active ($x=1$), the utility difference of continuing to mine remains strictly negative, i.e., $\Delta U_{\mathcal{T}_i}^{\mathrm{Stop}}(1) < 0$.
According to Lemma \ref{lemma:mono-utility}, $\Delta U_{\mathcal{T}_i}^{\mathrm{Stop}}(x)$ is strictly monotonically increasing over $x \in [0, 1]$. Therefore, for any state $x \in [0, 1]$, it universally holds that:
\begin{equation*}
    \Delta U_{\mathcal{T}_i}^{\mathrm{Stop}}(x) \le \Delta U_{\mathcal{T}_i}^{\mathrm{Stop}}(1) < 0 \implies U_{\mathcal{T}_i}^{\mathrm{Mine}}(x) < U_{\mathcal{T}_i}^{\mathrm{Stop}}(x)
\end{equation*}
This implies that for any microscopic miner $\mathcal{T}_i$, regardless of the strategy combination adopted by other miners, the utility of choosing $S_{\mathcal{T}_i} = \mathrm{Mine}$ is always strictly less than that of choosing $S_{\mathcal{T}_i}=\mathrm{Stop}$. Therefore, $S_{\mathcal{T}_i}=\mathrm{Stop}$ is an absolute strictly dominant strategy. Through the Iterated Elimination of Strictly Dominated Strategies (IESDS), the only strategy combination the network can possibly converge to is $\mathbf{S}^* = (\mathrm{Stop}, \mathrm{Stop}, \dots, \mathrm{Stop})$.
\end{proof}

\subsection{Proof of Theorem \ref{thm:death-spiral} }\label{app:death-spiral-proof}
\begin{proof}
\heading{Global Convergence of the Evolutionary Trajectory.}
According to the uniqueness proof of the Nash Equilibrium in \ref{app:unique-ne-proof}, when $\alpha > \alpha_{\mathcal{A}}^*$, the inequality $\Delta U_{\mathcal{T}_i}^{\mathrm{Stop}}(x) < 0$ holds strictly over the entire effective interval $x \in (0, 1]$. 
Substituting this into the replicator dynamic equation \eqref{eq:replicator}, for any system state $x(t) \in (0, 1)$, since $x(t)(1-x(t)) > 0$, the state change rate satisfies:
\begin{equation*}
    \dot{x}(t) < 0.
\end{equation*}
This implies that $x=1$ is an unstable trivial stationary point, and any infinitesimal perturbation in hash power withdrawal will activate $\dot{x} < 0$. Since $x(t)$ is strictly monotonically decreasing over time and bounded by the lower bound $0$, its evolutionary trajectory must converge to the unique limit point $x=0$, which is asymptotically stable and satisfies $\dot{x}=0$. Thus, $\lim_{t \to \infty} x(t) = 0$, and the system necessarily tends toward a total network shutdown.

\heading{Characteristics of Accelerated Collapse via Positive Feedback.}
To characterize the non-linear dynamics of network liveness collapse, we analyze the collapse acceleration $-\ddot{x}(t)$. Applying the chain rule to the differential equation \eqref{eq:replicator} with respect to time $t$, and extracting the sign, we obtain:
\begin{align*}
    \ddot{x}(t) &= \frac{d}{dt} \Big( x(1-x) \Delta U_{\mathcal{T}_i}^{\mathrm{Stop}}(x) \Big) \\
    &= \frac{d}{dx} \Big( x(1-x) \Delta U_{\mathcal{T}_i}^{\mathrm{Stop}}(x) \Big) \cdot \frac{dx}{dt} \\
    &= \Big[ (1-2x)\Delta U_{\mathcal{T}_i}^{\mathrm{Stop}}(x) + x(1-x) \frac{d \Delta U_{\mathcal{T}_i}^{\mathrm{Stop}}(x)}{dx} \Big] \cdot \dot{x}
\end{align*}
Defining the collapse velocity as $-\dot{x}(t) > 0$, the collapse acceleration can be decomposed into two terms:
\begin{equation*}
    -\ddot{x}(t) = -\dot{x}(t)(1-2x)\Delta U_{\mathcal{T}_i}^{\mathrm{Stop}}(x) -\dot{x}(t) x(1-x) \frac{d \Delta U_{\mathcal{T}_i}^{\mathrm{Stop}}(x)}{dx}.
\end{equation*}

In traditional evolutionary game theory (where the baseline utility difference $\Delta U_{\mathcal{T}_i}^{\mathrm{Stop}}$ is constant), only the first term exists. Since $\Delta U_{\mathcal{T}_i}^{\mathrm{Stop}} < 0$, this term is positive for $x > 0.5$ and becomes negative for $x < 0.5$, exhibiting the classic symmetric Logistic S-shaped decay.

However, in a PDoS attack, the hash power vacuum effect triggers dynamic coupling at the system level. Within the evolution interval $x \in (0, 1)$, it holds that:
1. The collapse velocity $-\dot{x}(t) > 0$;
2. The parabolic factor $x(1-x) > 0$;
3. According to Lemma \ref{lemma:mono-utility}, the utility difference is strictly monotonically increasing, i.e., $\frac{d \Delta U_{\mathcal{T}_i}^{\mathrm{Stop}}(x)}{dx} > 0$.

The product of these three terms strictly guarantees that for any $x \in (0, 1)$, we have:
\begin{equation*}
    -\dot{x}(t) x(1-x) \frac{d \Delta U_{\mathcal{T}_i}^{\mathrm{Stop}}(x)}{dx} > 0.
\end{equation*}
This proves that the hash power vacuum effect, unique to PDoS, continuously injects an additional strictly positive collapse acceleration into the system. This positive feedback mechanism ensures that the evolutionary trajectory of PDoS is always suppressed below the baseline Logistic curve, effectively eliminating the possibility of a soft landing and manifesting as an irreversible, accelerated avalanche effect.
\end{proof}

\subsection{Proof of Theorem \ref{thm:epsilon-robustness} }\label{app:epsilon-robustness-proof}
\begin{proof}
\heading{Existence and Uniqueness of the Critical Threshold.}
In the presence of a frictional barrier $\epsilon$, for the attacker to trigger a network-wide liveness collapse, they must create a utility gap deep enough at the initial fully active state ($x=1$), i.e., satisfying $\Delta U_{\mathcal{T}_i}^{\mathrm{Stop}}(1) < -\epsilon$. Since Lemma \ref{lemma:mono-utility} has already proven the strict monotonicity of $\Delta U_{\mathcal{T}_i}^{\mathrm{Stop}}(x)$, as long as the initial gap breaches $\epsilon$, the subsequent degradation of liveness will follow an irreversible accelerated trend.

According to Equation \eqref{u_ti_stop}, we examine the individual shutdown utility difference in the fully active state ($x=1$), which can be viewed as a function of the attack hash power $\alpha$:
\begin{equation*}
    \Delta U_{\mathcal{T}_i}^{\mathrm{Stop}}(\alpha) = c \cdot \pi_{\mathrm{det}}(\alpha) \Big[ \omega_b \cdot e(\alpha) \cdot \bar{p}_{\mathcal{T}_i}(\alpha) - 1 \Big]
\end{equation*}

From the uniqueness proof of the Nash Equilibrium in \ref{app:unique-ne-proof}, when $\alpha > \alpha_{\mathcal{A}}^*$, $\Delta U_{\mathcal{T}_i}^{\mathrm{Stop}}(\alpha) < 0$. Given $c > 0$ and $\pi_{\mathrm{det}}(\alpha) > 0$, it must hold that $[\omega_b H(\alpha) - 1] < 0$. Combining the system parameter relations derived earlier \eqref{eq:pi_det}, \eqref{eq:pi_x}, \eqref{eq:e_x}, \eqref{eq:p_bar} and the total hash power conservation condition ($\beta+\eta+\delta = 1-\alpha$), the analytical expressions for each component function are:
Deterring state probability $\pi_{\mathrm{det}}(\alpha) = \frac{\alpha}{1 + \alpha(1-\alpha-r_1\beta)}$; 
inflation coefficient $e(\alpha) = 1 + M\frac{\alpha}{1-\alpha}$; 
and weighted average winning rate $\bar{p}_{\mathcal{T}_i}(\alpha) = (1-\gamma)(1-\alpha-r_1\beta)$.

Differentiating the multiplier term $\pi_{\mathrm{det}}(\alpha)$ with respect to $\alpha$ yields:
\begin{equation*}
    \frac{d \pi_{\mathrm{det}}(\alpha)}{d \alpha} = \frac{1 + \alpha^2}{\big[ 1 + \alpha(1-\alpha-r_1\beta) \big]^2} > 0.
\end{equation*}

Next, we define and examine the core revenue term:
\begin{align*}
    H(\alpha) &= e(\alpha) \cdot \bar{p}_{\mathcal{T}_i}(\alpha)\\
    &= \left( 1 + M\frac{\alpha}{1-\alpha} \right) (1-\gamma) \big( (1-\alpha) - r_1\beta \big) \\
    &= (1-\gamma) \left( (1-\alpha) - r_1\beta + M\alpha - \frac{M r_1 \beta \alpha}{1-\alpha} \right)
\end{align*}
Differentiating with respect to $\alpha$ using the quotient rule for the last term:
\begin{align*}
    H'(\alpha) &= (1-\gamma) \left[ -1 + M - M r_1 \beta \frac{1 \cdot (1-\alpha) - \alpha(-1)}{(1-\alpha)^2} \right] \\
    &= (1-\gamma) \left[ -(1-M) - \frac{M r_1 \beta}{(1-\alpha)^2} \right]
\end{align*}
Since the MEV value proportion constant $M \in (0, 1)$ and propagation advantage $\gamma \in (0, 1)$, it follows that $-(1-M) < 0$. Since $M, r_1, \beta \ge 0$, the trailing term is also $\le 0$. Thus, $H'(\alpha) < 0$ strictly holds.

Combining the monotonicity of both terms, we apply the product rule to calculate the derivative of the total utility difference:
\begin{align*}
    \frac{d \Delta U_{\mathcal{T}_i}^{\mathrm{Stop}}(\alpha)}{d \alpha} =& \underbrace{c}_{(+)} \Bigg( \underbrace{\frac{d \pi_{\mathrm{det}}(\alpha)}{d \alpha}}_{(+)} \underbrace{\Big[ \omega_b H(\alpha) - 1 \Big]}_{(-)} \\
    &+ \underbrace{\pi_{\mathrm{det}}(\alpha)}_{(+)} \omega_b \underbrace{H'(\alpha)}_{(-)} \Bigg) < 0.
\end{align*}
This proves that in the fully active state, the individual relative shutdown utility difference strictly monotonically decreases with the increase of attack hash power.

Given the realistic frictional cost $\epsilon \in (0, c)$, we examine the value of the utility difference function at both ends of the interval:
When $\alpha = \alpha_{\mathcal{A}}^*$, the system is precisely at the break-even point without friction, i.e., $\Delta U_{\mathcal{T}_i}^{\mathrm{Stop}}(\alpha_{\mathcal{A}}^*) = 0 > -\epsilon$.
In the limit state, considering that the accumulable MEV inflation bonus within a single block must have a theoretical upper bound, i.e., $e(\alpha) \le E_{\max}$ universally holds. Since the microscopic winning rate limit of target miners is $\lim_{\alpha \to 1^-} \bar{p}_{\mathcal{T}_i}(\alpha) = 0$, the core revenue term satisfies $\lim_{\alpha \to 1^-} H(\alpha) \le \lim_{\alpha \to 1^-} E_{\max} \cdot \bar{p}_{\mathcal{T}_i}(\alpha) = 0$. Combining this with the limit of the deterrence sojourn probability $\lim_{\alpha \to 1^-} \pi_{\mathrm{det}}(\alpha) = 1$, the utility difference under the bounded inflation constraint strictly satisfies:
\begin{equation*}
    \lim_{\alpha \to 1^-} \Delta U_{\mathcal{T}_i}^{\mathrm{Stop}}(\alpha) = c \big( \omega_b \cdot 0 - 1 \big) = -c < -\epsilon
\end{equation*}

Since $\Delta U_{\mathcal{T}_i}^{\mathrm{Stop}}(\alpha)$ is a continuous function on the interval $\alpha\in[\alpha_{\mathcal{A}}^*, 1)$ and is strictly monotonically decreasing, according to the Intermediate Value Theorem, there must exist a unique solution $\alpha_{\mathcal{A}}^\epsilon \in (\alpha_{\mathcal{A}}^*, 1)$ such that $\Delta U_{\mathcal{T}_i}^{\mathrm{Stop}}(\alpha_{\mathcal{A}}^\epsilon) = -\epsilon$. When $\alpha > \alpha_{\mathcal{A}}^\epsilon$, the inequality $\Delta U_{\mathcal{T}_i}^{\mathrm{Stop}}(\alpha) < -\epsilon$ holds strictly, thereby breaching the frictional barrier.

\heading{Relative Exploitation Advantage of PDoS.}
Consistent with the previous strategy space constraints, a BDoS attack corresponds to the boundary case of zero infiltration $(r_1=0, r_2=0)$. Comparing the utility gap created by both at the same $\alpha$ in the fully active state $x=1$:
For the microscopic winning rate term:
\begin{align*}
    \bar{p}_{\mathcal{T}_i}(\alpha)\big|_{\mathrm{PDoS}} &= (1-\gamma)(1-\alpha-r_1\beta) \\
    \bar{p}_{\mathcal{T}_i}(\alpha)\big|_{\mathrm{BDoS}} &= (1-\gamma)(1-\alpha)
\end{align*}
The difference is $\bar{p}_{\mathcal{T}_i}(\alpha)\big|_{\mathrm{BDoS}} - \bar{p}_{\mathcal{T}_i}(\alpha)\big|_{\mathrm{PDoS}} = (1-\gamma)r_1\beta > 0$. Since both have identical inflation coefficients $e(\alpha)$ at the same $\alpha$, the smaller winning rate of PDoS leads to a larger internal revenue gap:
\begin{equation}\label{eq:proof_epsilon_1}
    \Big[\omega_b e(\alpha) \bar{p}_{\mathcal{T}_i}(\alpha)\big|_{\mathrm{PDoS}} - 1\Big] < \Big[\omega_b e(\alpha) \bar{p}_{\mathcal{T}_i}(\alpha)\big|_{\mathrm{BDoS}} - 1\Big] < 0
\end{equation}

For the deterrence time multiplier term:
\begin{align*}
    \pi_{\mathrm{det}}(\alpha)\big|_{\mathrm{PDoS}} &= \frac{\alpha}{1 + \alpha(1-\alpha-r_1\beta)} \\
    \pi_{\mathrm{det}}(\alpha)\big|_{\mathrm{BDoS}} &= \frac{\alpha}{1 + \alpha(1-\alpha)}
\end{align*}
Since the denominator $1 + \alpha(1-\alpha-r_1\beta) < 1 + \alpha(1-\alpha)$, it clearly holds that:
\begin{equation}\label{eq:proof_epsilon_2}
    \pi_{\mathrm{det}}(\alpha)\big|_{\mathrm{PDoS}} > \pi_{\mathrm{det}}(\alpha)\big|_{\mathrm{BDoS}} > 0
\end{equation}

Combining inequalities \eqref{eq:proof_epsilon_1} and \eqref{eq:proof_epsilon_2}: PDoS possesses a larger positive multiplier (longer deterrence time) multiplied by a deeper negative revenue term (greater single-round loss), resulting in a product with a strictly larger absolute value. That is, it strictly holds that:
\begin{align*}
    \Delta U_{\mathcal{T}_i}^{\mathrm{Stop}}(\alpha)\big|_{\mathrm{PDoS}} <& \Delta U_{\mathcal{T}_i}^{\mathrm{Stop}}(\alpha)\big|_{\mathrm{BDoS}} < 0 \\
    \implies &\Big|\Delta U_{\mathcal{T}_i}^{\mathrm{Stop}}(\alpha)\big|_{\mathrm{PDoS}}\Big| > \Big|\Delta U_{\mathcal{T}_i}^{\mathrm{Stop}}(\alpha)\big|_{\mathrm{BDoS}}\Big|
\end{align*}
This proves that under identical external attack hash power, the utility gap created by PDoS is strictly greater than that of BDoS, demonstrating a stronger deterrence penetration capability.
\end{proof}

\section{Additional Discussion}
\heading{Modeling Scope and Limitations}
When target miners shut down to avoid losses under attack deterrence, the effective active hash power $\alpha_{\mathrm{act}}$ decreases. This passively lengthens the average block interval, allowing each block to accumulate more transaction fees and continuously arriving high-value opportunities, such as MEV-like rewards, before it is eventually mined. As a result, liveness degradation simultaneously reduces throughput and endogenously amplifies the per-block reward faced by the remaining active miners. In PDoS, this effect is further combined with the attacker's continuous extraction of share rewards from the infiltrated victim pool, creating an additional subsidy channel that is absent in classical BDoS. In other words, while slowing down the chain, the attacker also exploits both reward inflation caused by the slower block rate and the victim pool's payout mechanism to offset its attack cost. This explains the counterintuitive phenomenon revealed in this paper: in high-fee, high-MEV, or bull-market environments, a higher per-block reward does not necessarily improve resistance to liveness attacks. Instead, it can lower the cost threshold of PDoS and make the attack easier to sustain.

We note that the above results are subject to three important limitations. First, we focus on short- to medium-term attack windows. This assumption is appropriate for PoW systems such as Bitcoin, where difficulty adjustment is relatively slow and strategic withholding or deterrence attacks unfold within a locally stationary window. However, static difficulty underestimates the target's long-term utility: persistent shutdown eventually reduces difficulty and restores mining profitability. Second, the analytical inflation coefficient $e^S$ is uncapped, whereas block capacity and demand bound real fee accumulation. Third, the infiltration-originated release path is protocol dependent. Under conventional pool-assigned jobs, the worker can construct the valid header but normally relies on the pool to reconstruct and publish the body after the retained FPoW is submitted. A pool that rejects stale jobs too quickly, changes templates aggressively, or conditionally escrows payouts reduces the useful race window. Stratum V2 custom jobs can give a miner the full template, but deployment is not universal~\cite{stratumv2jobs}.

We model target miners as rational agents driven by short-term expected payoff. This assumption is consistent with the classical analytical framework for incentive-based attacks such as BDoS, and also reflects the behavior of industrial miners operating under thin profit margins. Nevertheless, real-world participants may exhibit more complex behavior, including cross-chain migration, heterogeneous time horizons, and mixed short-term and long-term objectives. Extending our model to capture these richer strategic behaviors is an important direction for future work.

\heading{Reducing Race-Time Propagation Asymmetry.}
Our analysis shows that the attack becomes stronger as the rushing ability $\gamma$ increases. Therefore, any network-layer improvement that reduces the attacker's propagation advantage directly lowers its winning probability in racing states and weakens the credibility of the header-only deterrence signal. In practice, mining pools and large miners can improve relay diversity, strengthen anti-eclipse connectivity, deploy multi-homed networking, optimize block-body propagation paths, and shorten abnormal-fork detection time. These measures reduce the attacker's first-mover advantage during races and increase the hash-power threshold required for PDoS to successfully induce shutdown.

\end{document}